\documentclass[journal, letterpaper]{IEEEtran}

\usepackage{graphicx,color,subfigure}
\usepackage{ifthen}
\usepackage{hyperref}
\usepackage[cmex10]{amsmath}
\usepackage{bbm,amssymb,latexsym,url}

\definecolor{subtler}{rgb}{1,0,0.1}    

\newcommand{\be}{\begin{equation}}
\newcommand{\ee}{\end{equation}}
\newcommand{\ben}{\begin{equation*}}
\newcommand{\een}{\end{equation*}}
\newcommand{\ba}{\begin{eqnarray}}
\newcommand{\ea}{\end{eqnarray}}
\newcommand\Var {{\rm Var}}

\def\techrep{\cite{FisherTechrep08}}
\newtheorem{thm}{Theorem}
\newtheorem{cor}[thm]{Corollary}
\newtheorem{lem}[thm]{Lemma}
\newtheorem{prop}[thm]{Proposition}
\newtheorem{defn}[thm]{Definition}
\newtheorem{rem}[thm]{Remark}

\newcommand{\Real}{\mathbb R}

\newcommand{\tr}{\mathbf{tr}}

\newcommand{\T}{\mathrm{T}}
\newcommand{\E}{\mathbb E}
\newcommand{\Count}{1}
\newcommand{\uncon}{\mathbf{J}}
\newcommand{\con}{\boldsymbol{\mathcal{I}}}
\newcommand{\constr}{\mathbf{G}}
\newcommand{\diag}{\mathbf{diag}}
\newcommand{\rank}{\mathrm{rank}}
\newcommand{\ind}{\mathbbm 1}
\newcommand{\veth}{\boldsymbol{\hat \theta}}

\def\ppopt{\pp'} 
\def\si{m_i}       
\def\nf{N_{\!f}}       
\def\pp{p_p}     
\def\qp{q_p}	 
\def\pf{p_{\!f}}      
\def\qf{q_{\!f}} 
\def\np{n}
\def\naive{na\"{\i}ve\ }
\def\ppd{p_{p, \text{DS}}}

\def\bjk{b_{jk}}
\def\vth{\boldsymbol\theta}
\def\vgm{\boldsymbol\gamma}
\def\Bt{\mathbf{\tilde B}}
\def\Dt{\mathbf{\tilde D}}
\def\Jt{\mathbf{\tilde J}}
\def\onew{\mathbf{1}_W}
\def\bt{\mathbf{b}^\T_0}
\def\dt{\mathbf{d}}

\def\ppo{p_{p,1}}
\def\pfo{p_{\!f,1}}
\def\ppt{p_{p,2}}
\def\pft{p_{\!f,2}}

\graphicspath{{.}
{EPSFigures/}
}
\def\threeup{63mm}
\def\oneup{86mm}

\def\DSfamily                  {\includegraphics[width=83mm]{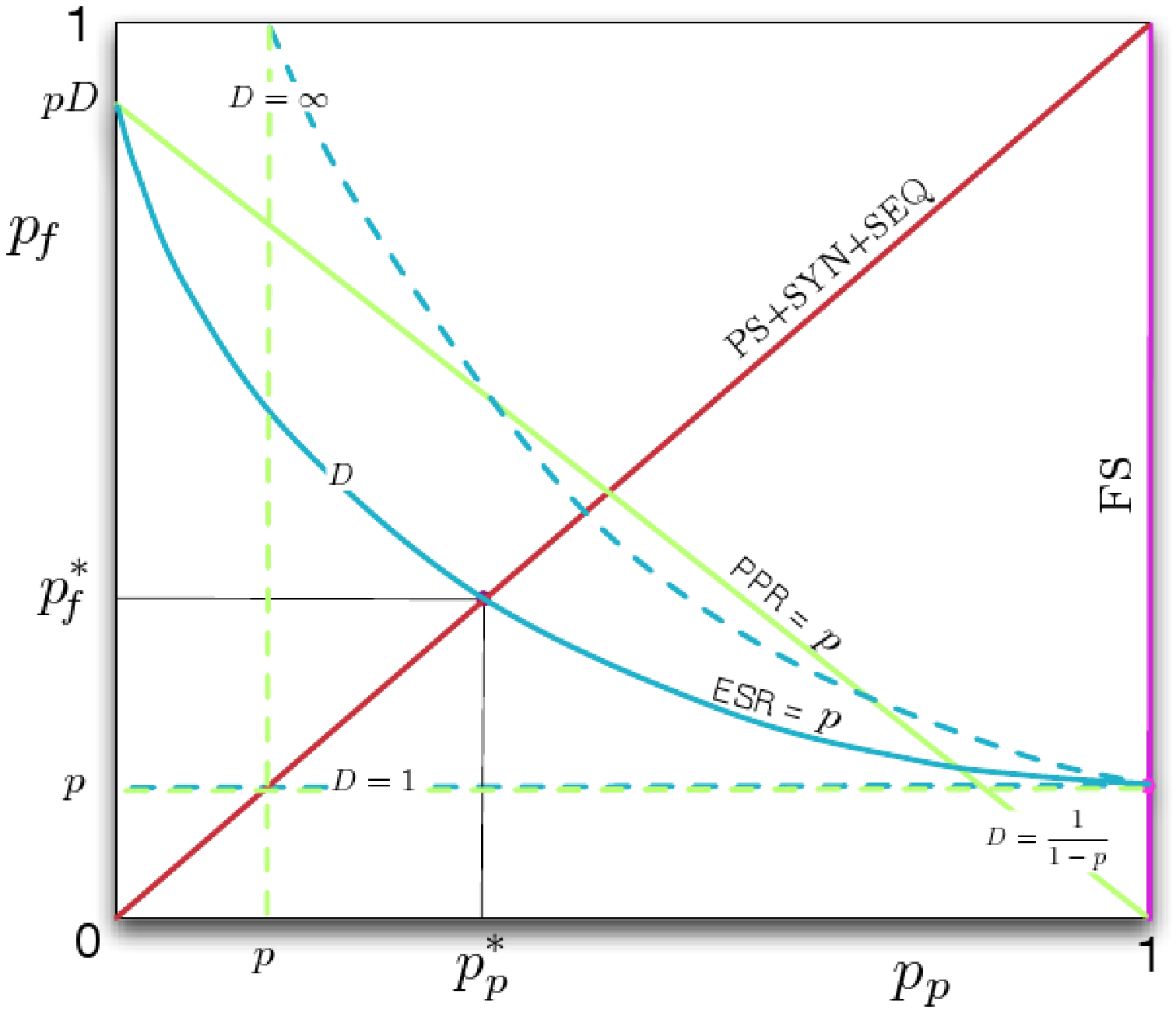}}

\def\shortComparePPR       {\includegraphics[width=90mm]{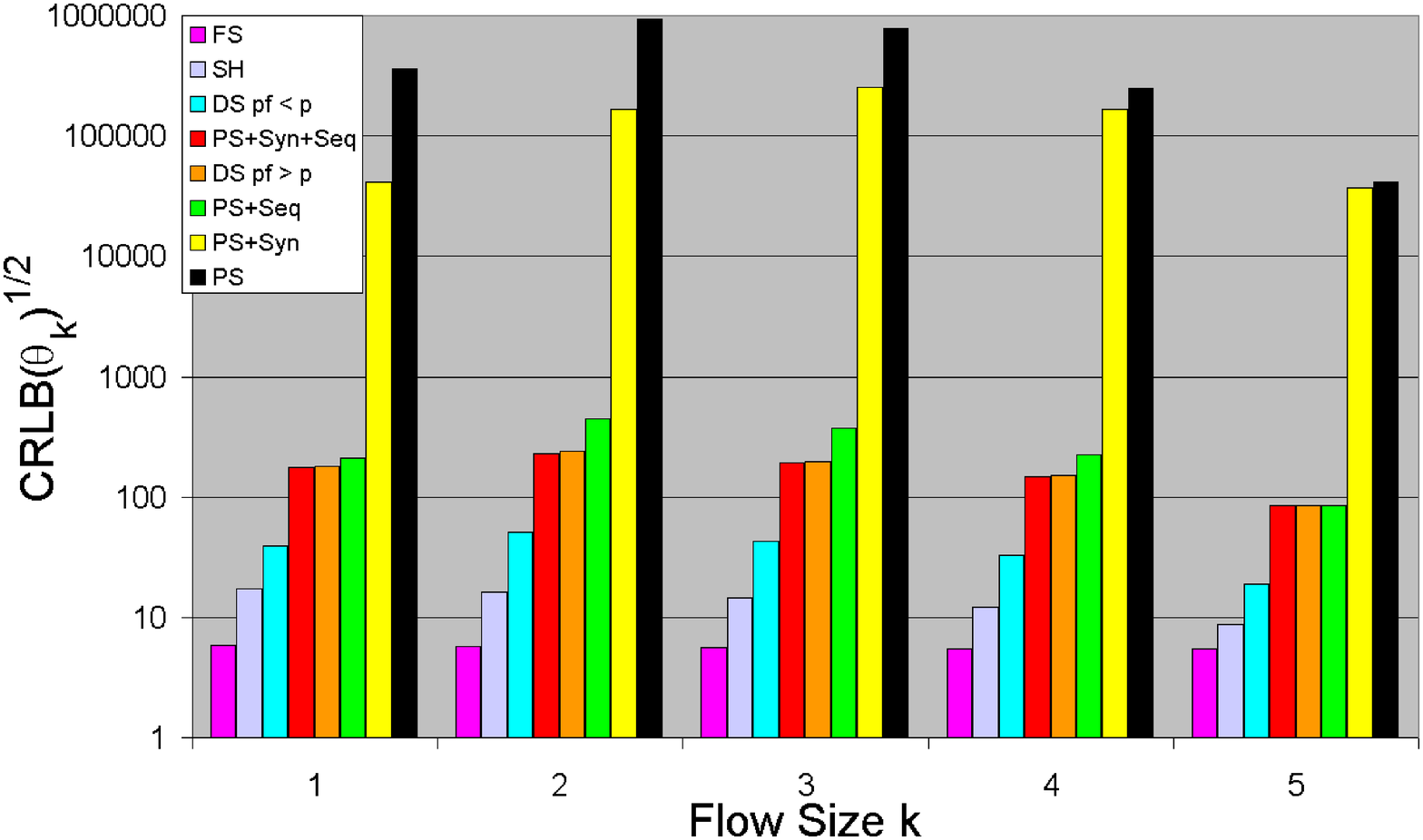}}
\def\shortCompareESR       {\includegraphics[width=\oneup]{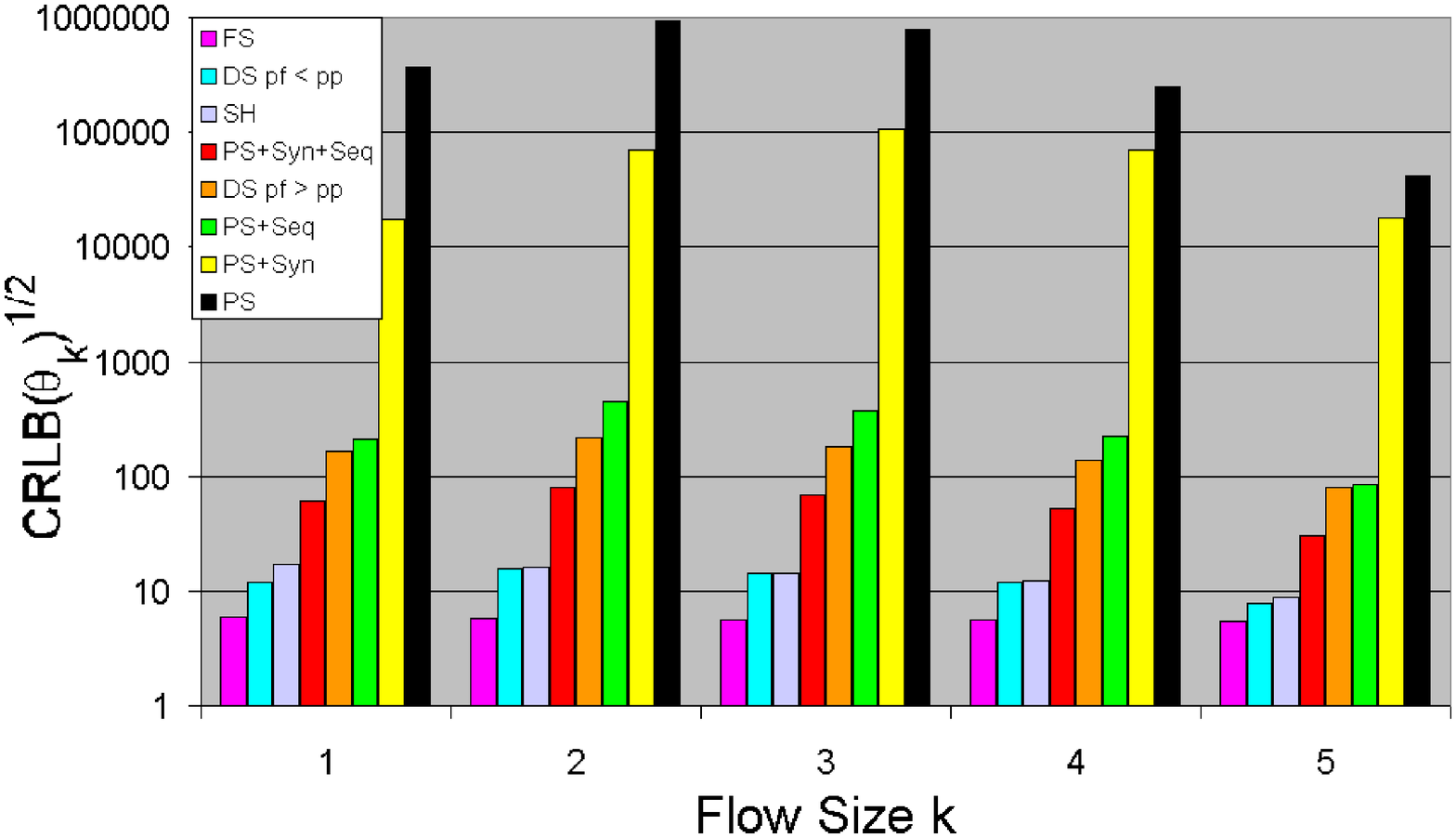}}
\def\sampvaryPPR           {\includegraphics[width=\oneup]{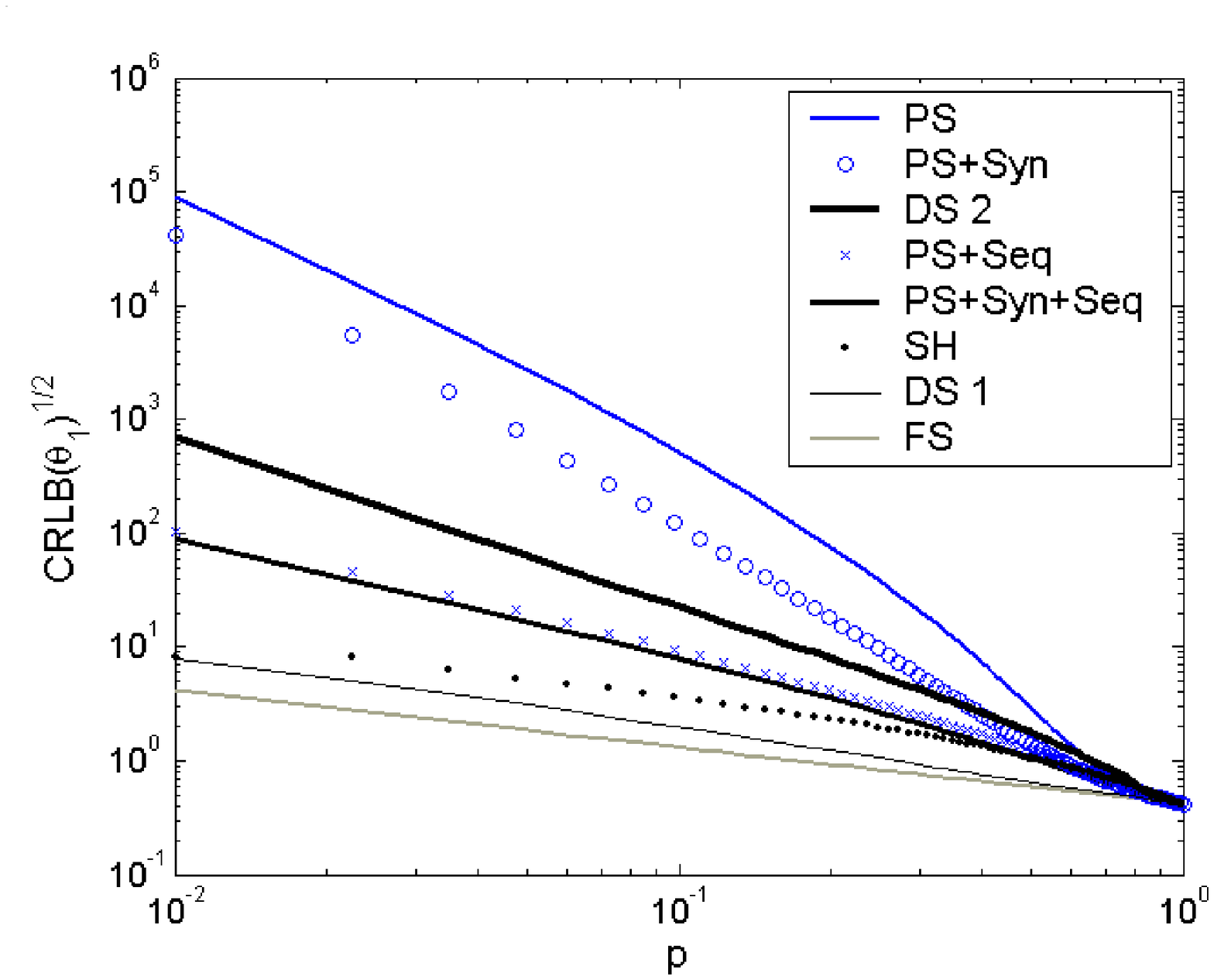}}

\def\longComparePPR        {\includegraphics[width=\threeup]{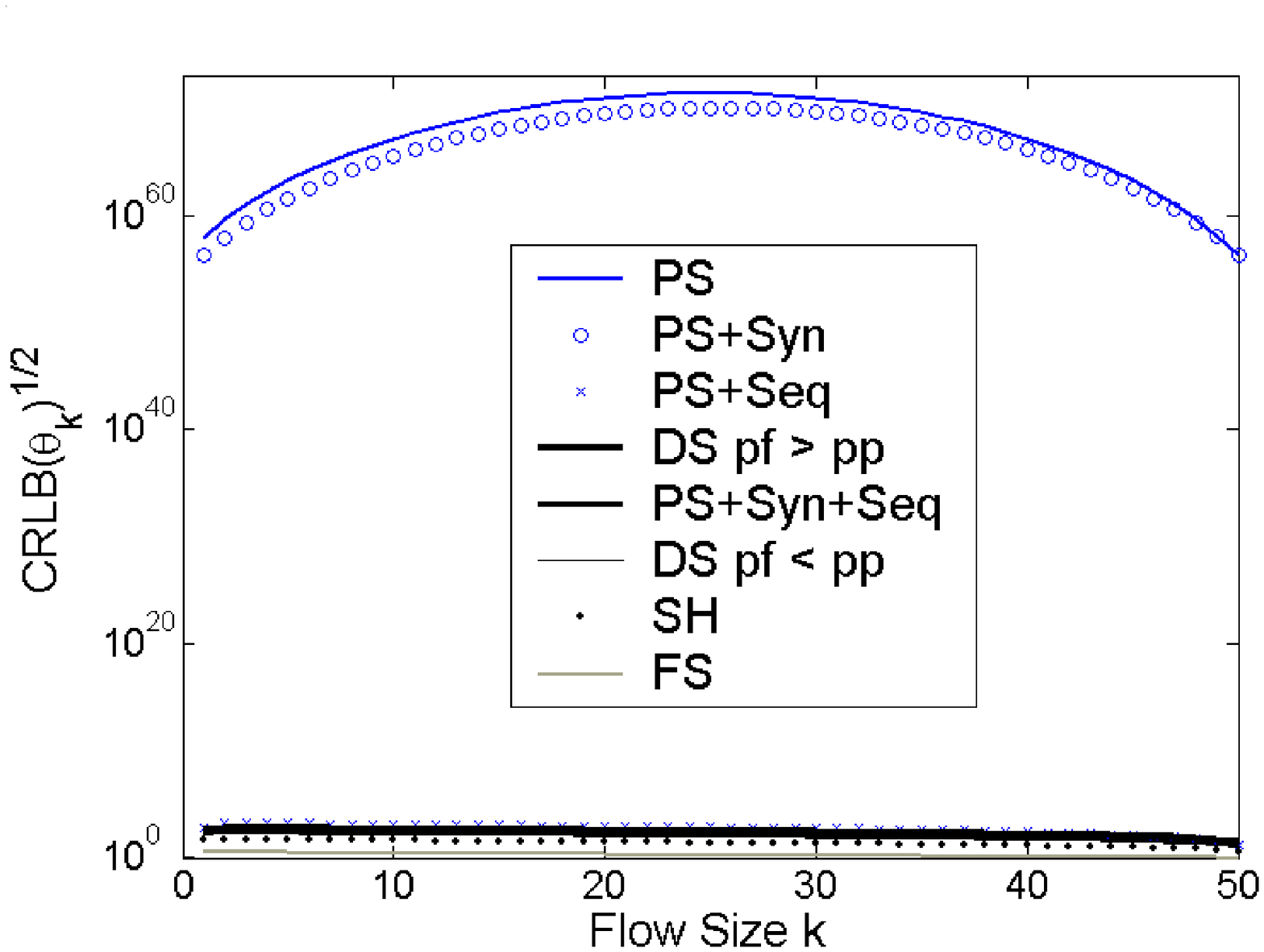}}
\def\longCompareESR        {\includegraphics[width=\threeup]{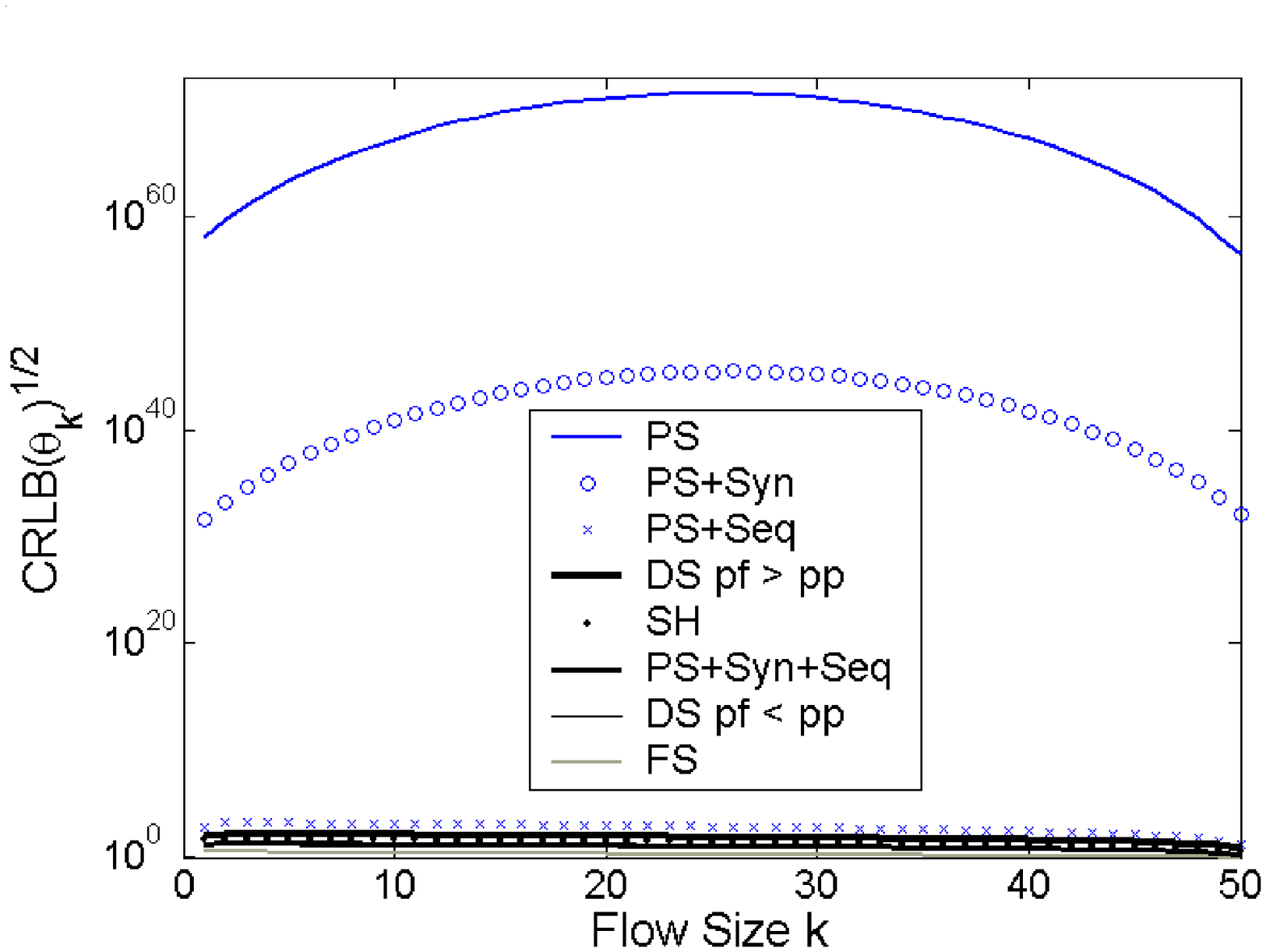}}
\def\longCompareESRzoom    {\includegraphics[width=\threeup]{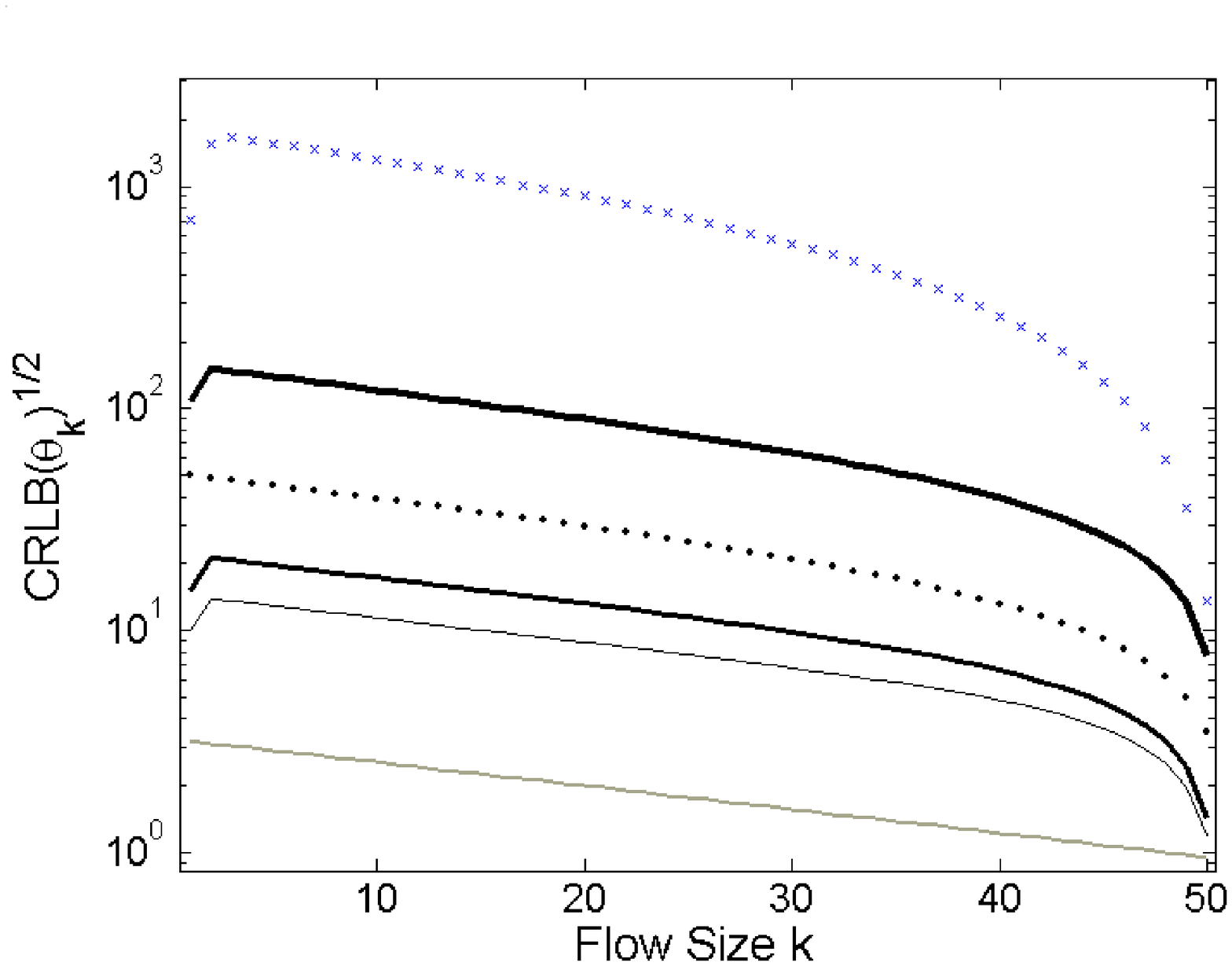}}

\def\PSPSSComp	           {\includegraphics[width=\oneup]{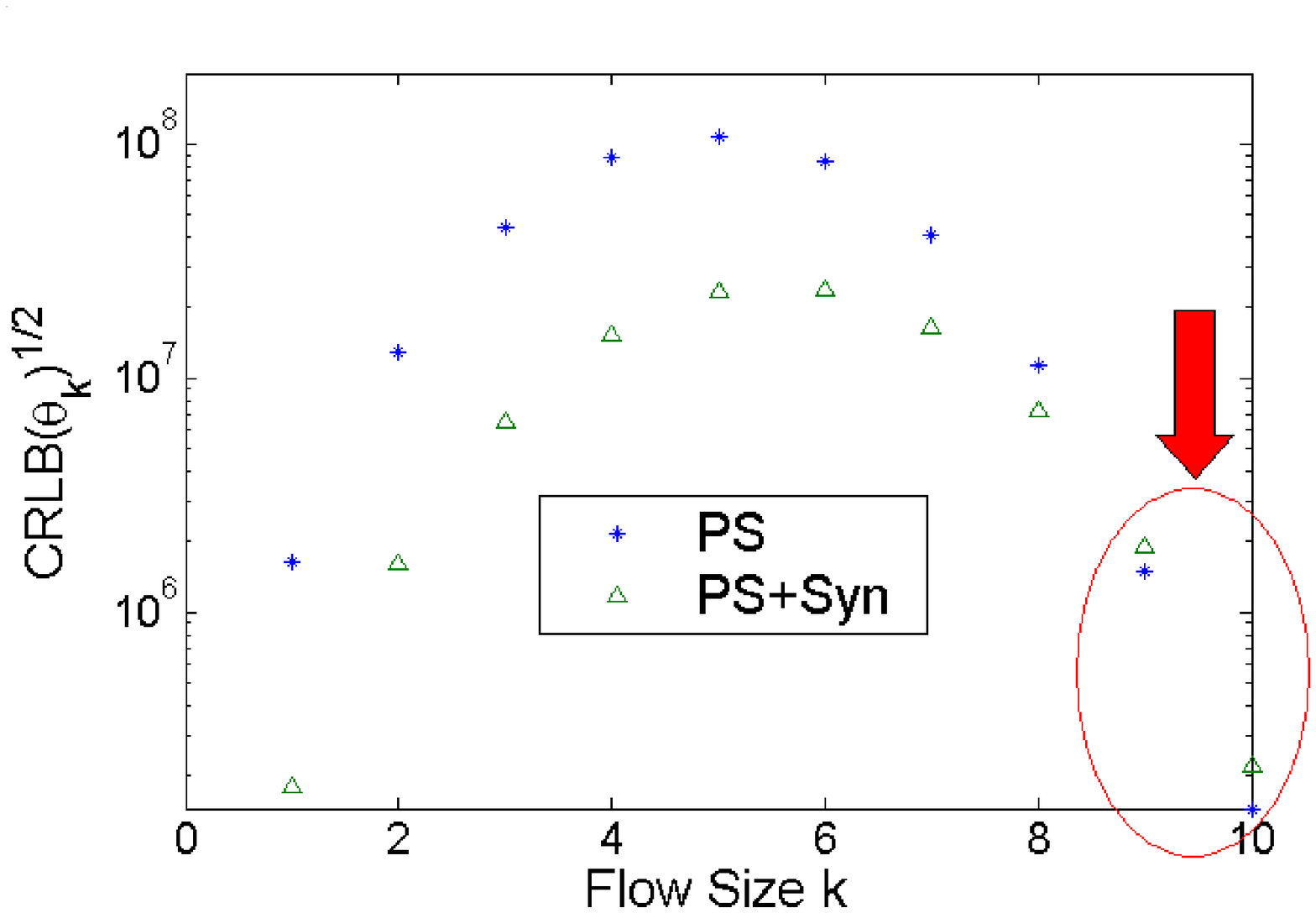}}	
\def\SHFSComp	           {\includegraphics[width=\oneup]{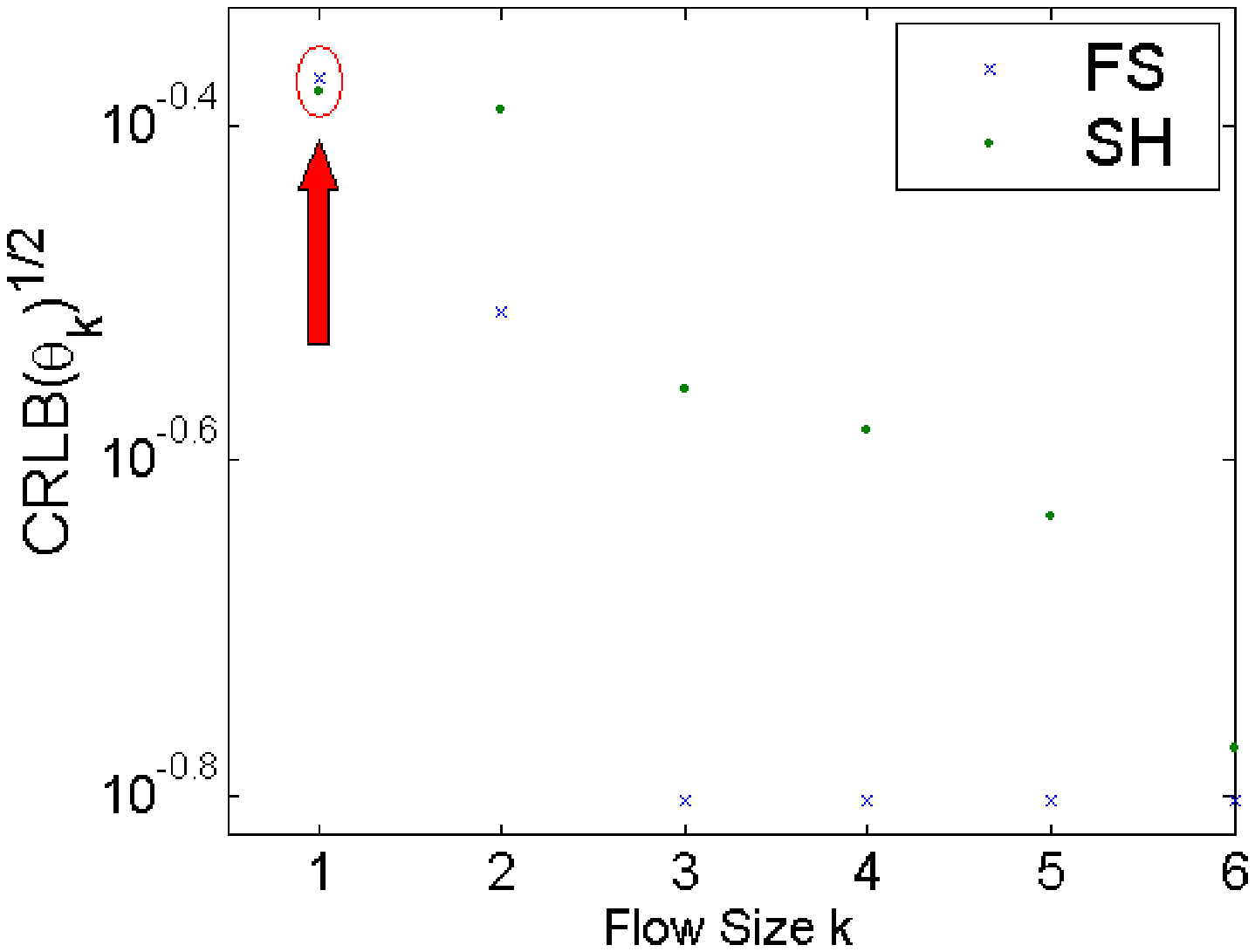}}

\def\datacompare           {\includegraphics[width=85mm]{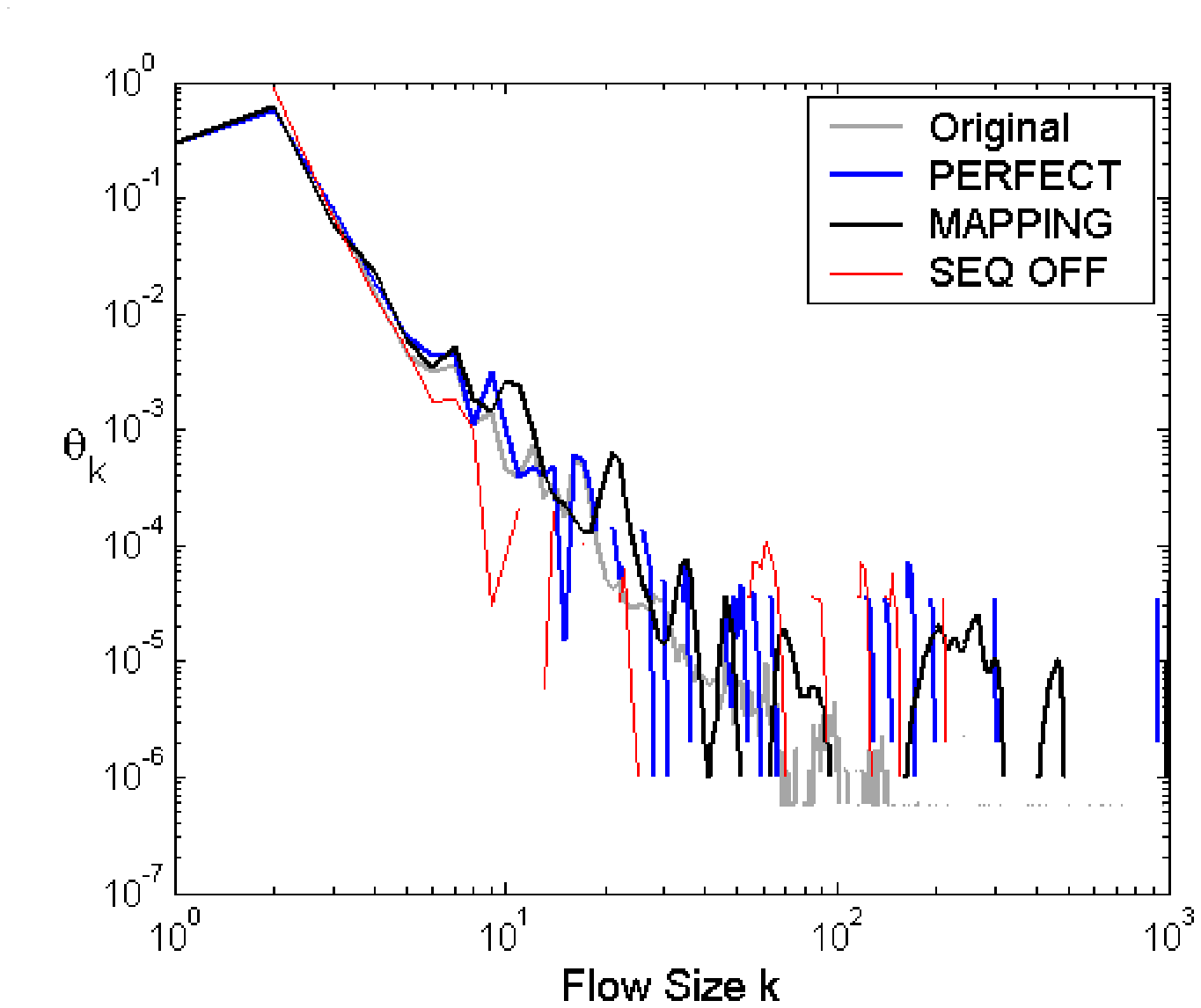}}
\def\ESRcompareone         {\includegraphics[width=\oneup]{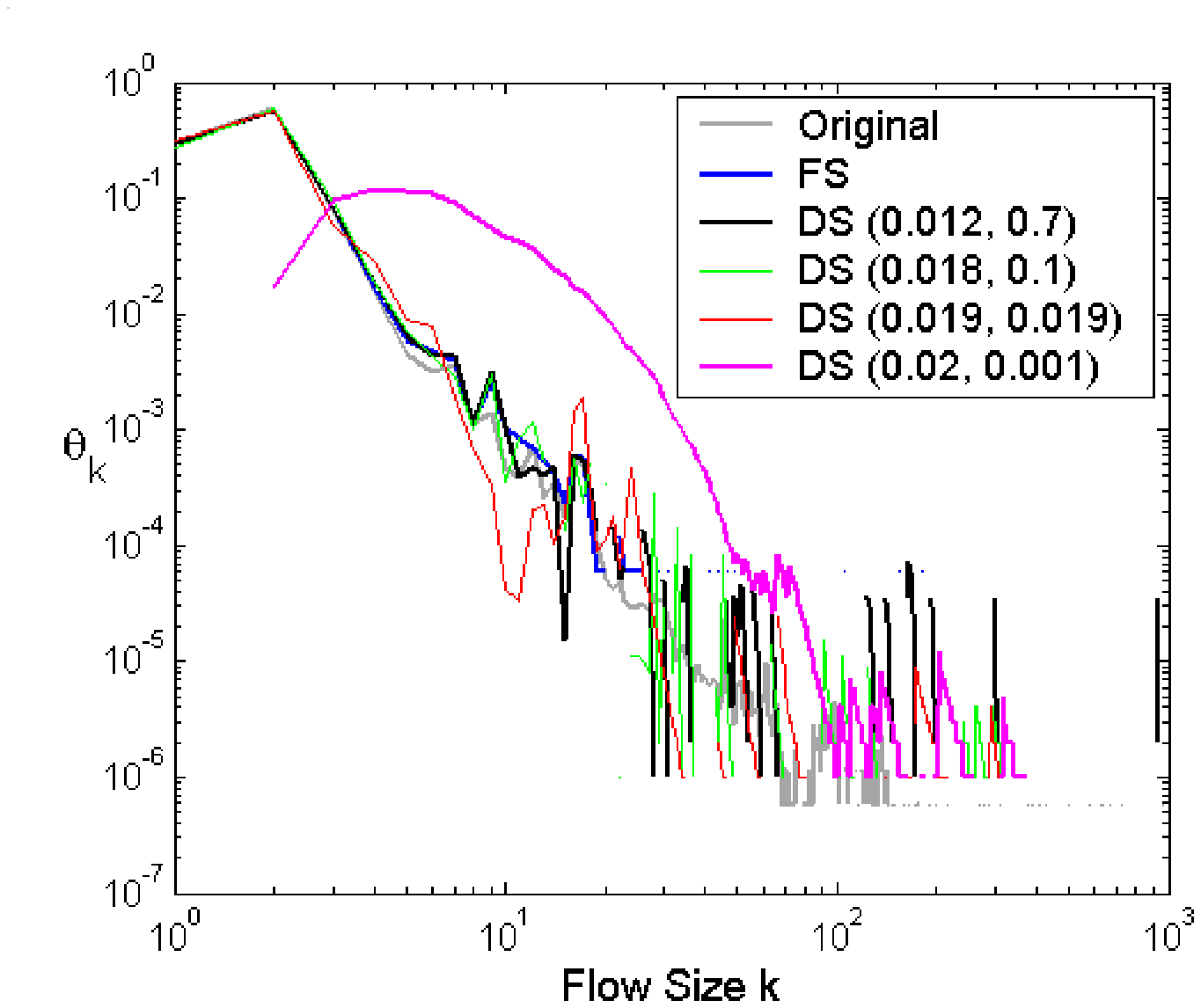}}
\def\ESRcomparetwo         {\includegraphics[width=\oneup]{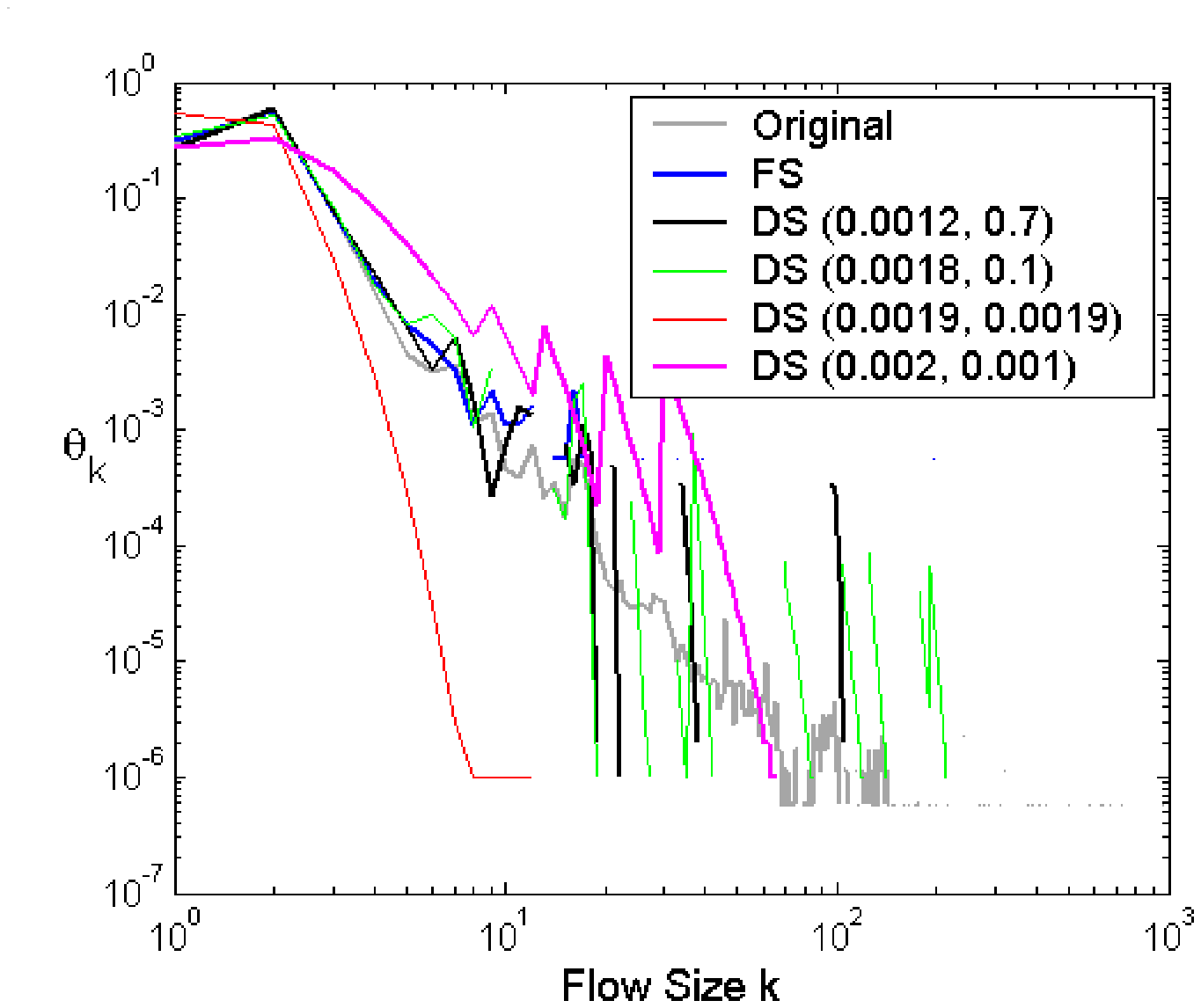}}
\def\ESRseqcompone         {\includegraphics[width=\oneup]{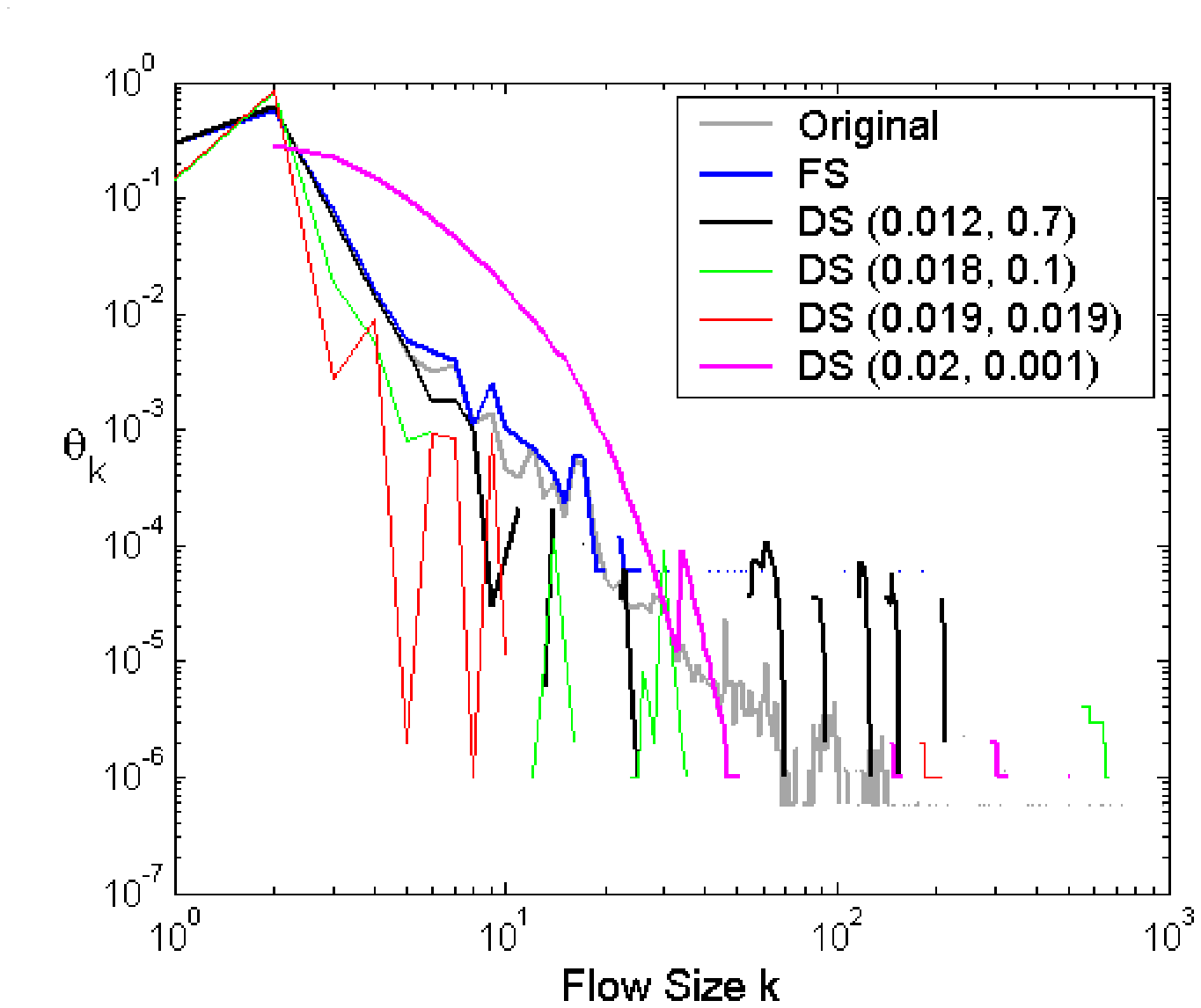}}
\def\ESRseqcomptwo	       {\includegraphics[width=\oneup]{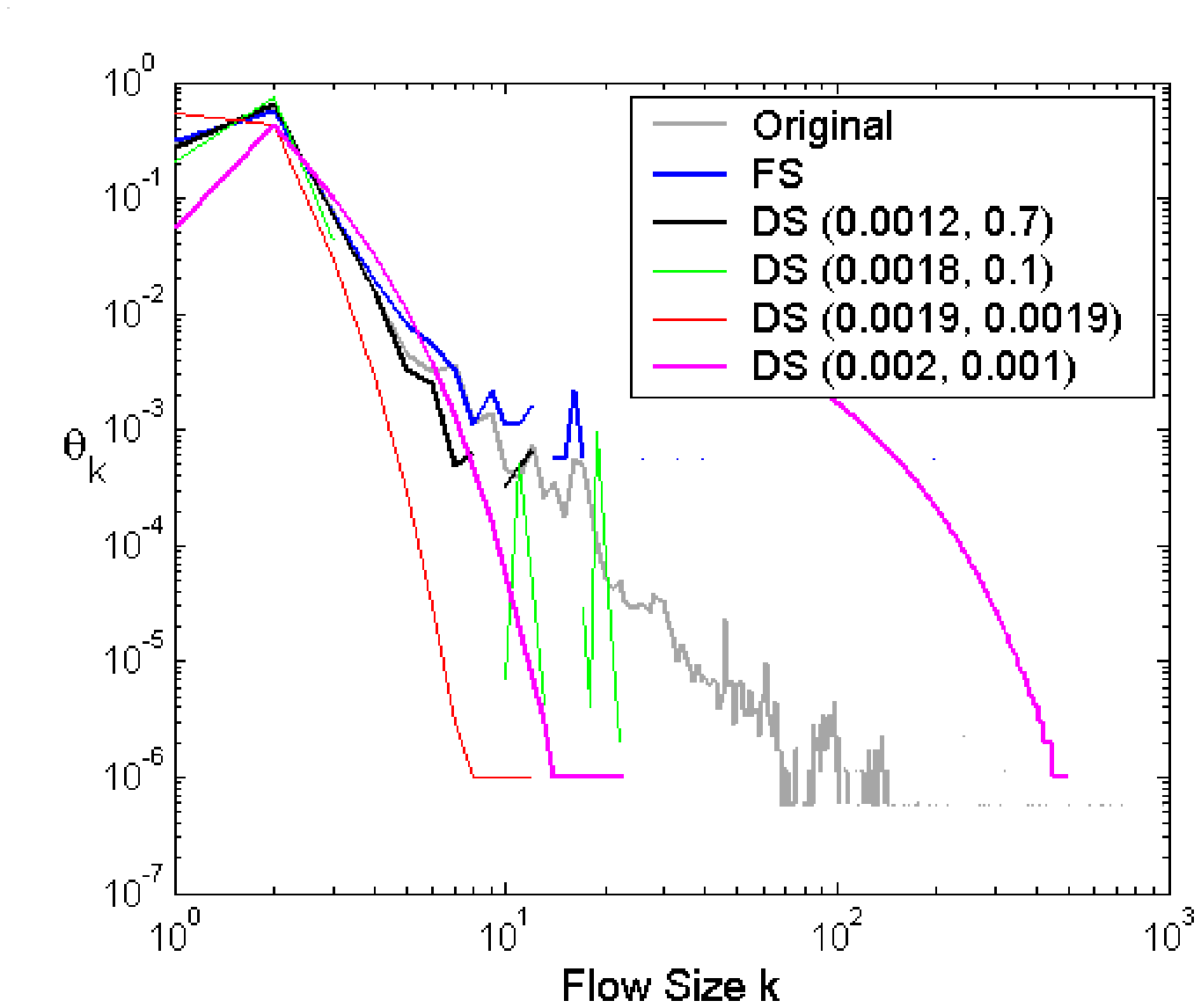}}
\def\ESRshcompare	       {\includegraphics[width=\oneup]{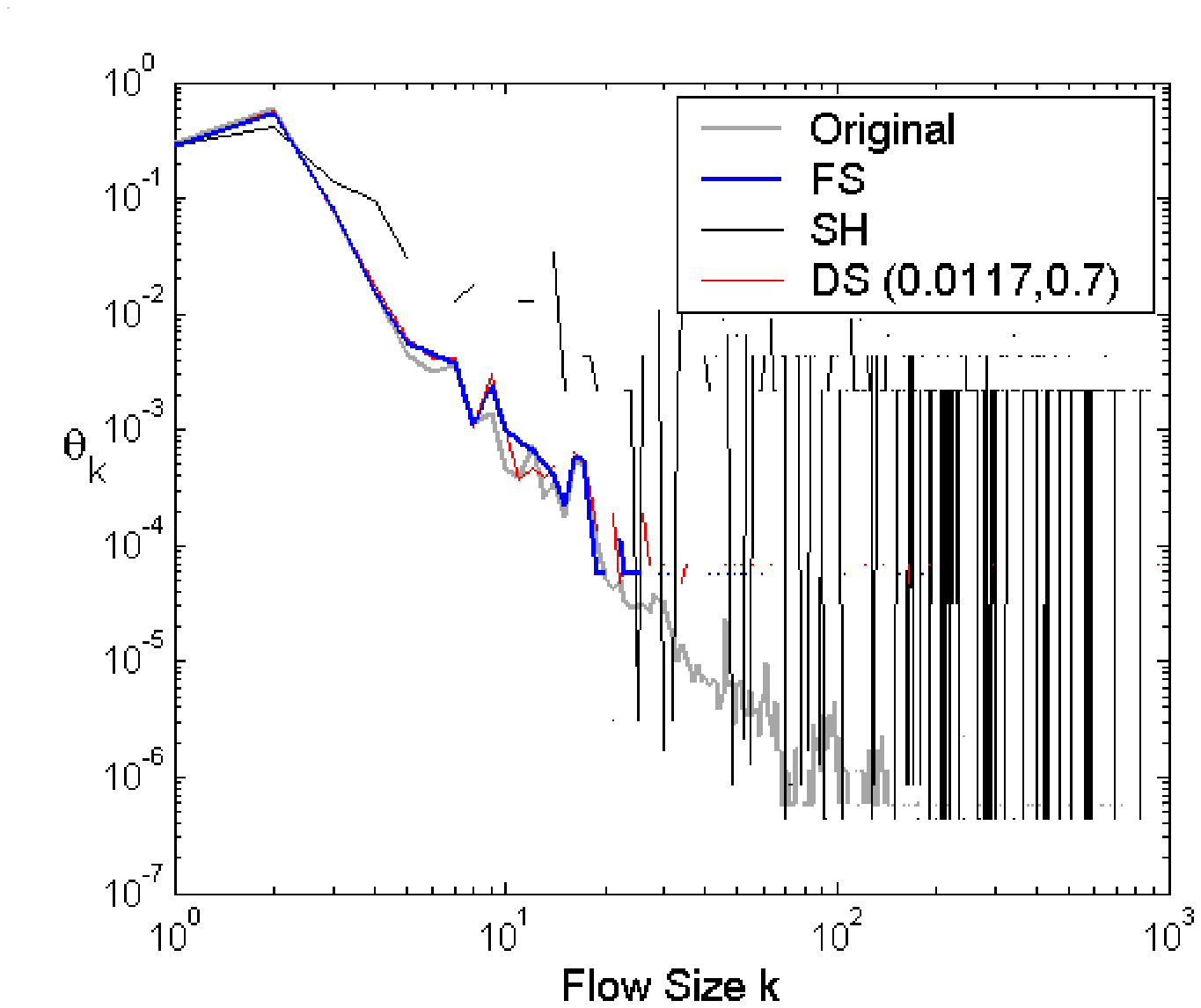}}

\def\ESRabiper 			{\includegraphics[width=\oneup]{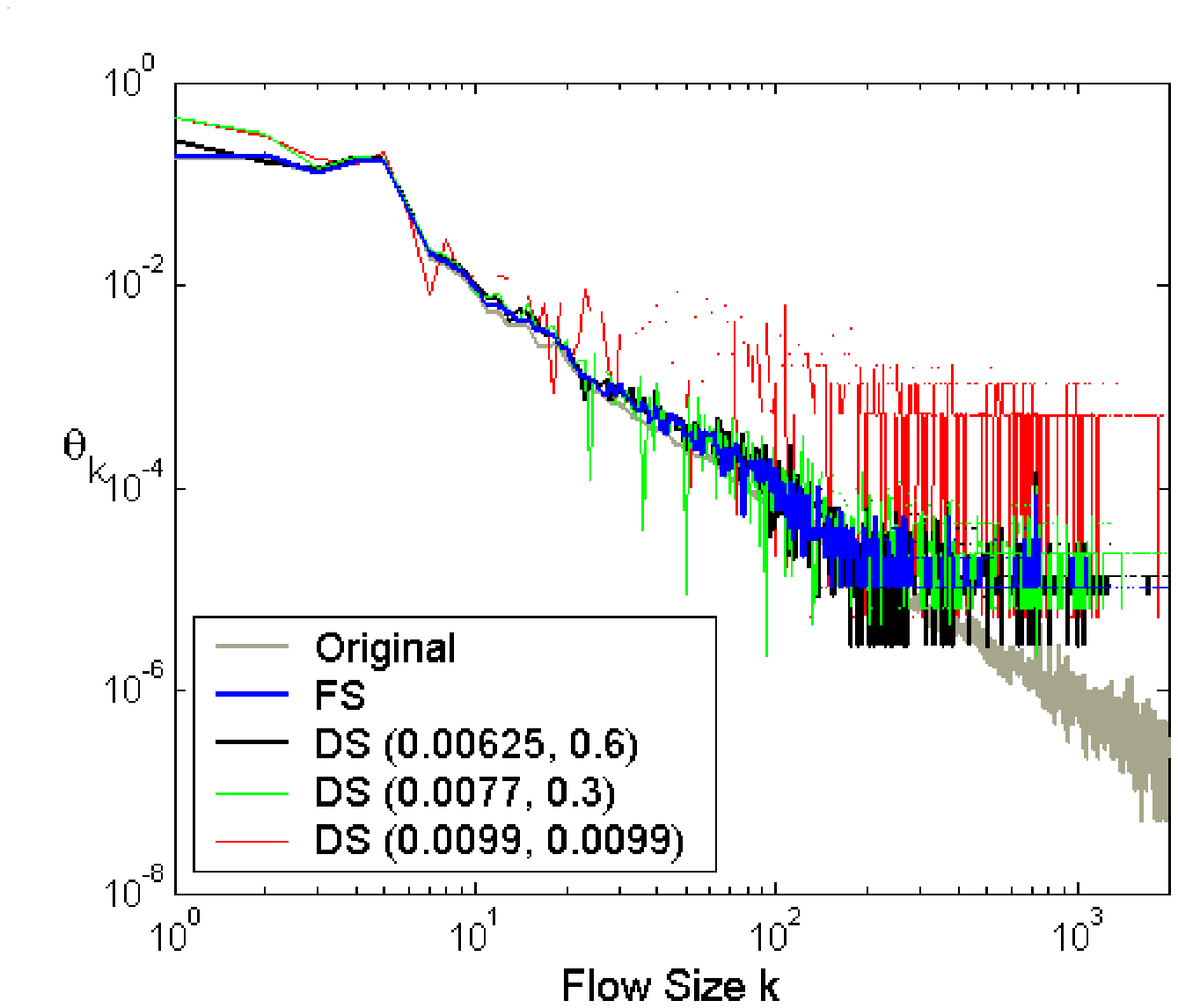}}
\def\ESRabiseq			{\includegraphics[width=\oneup]{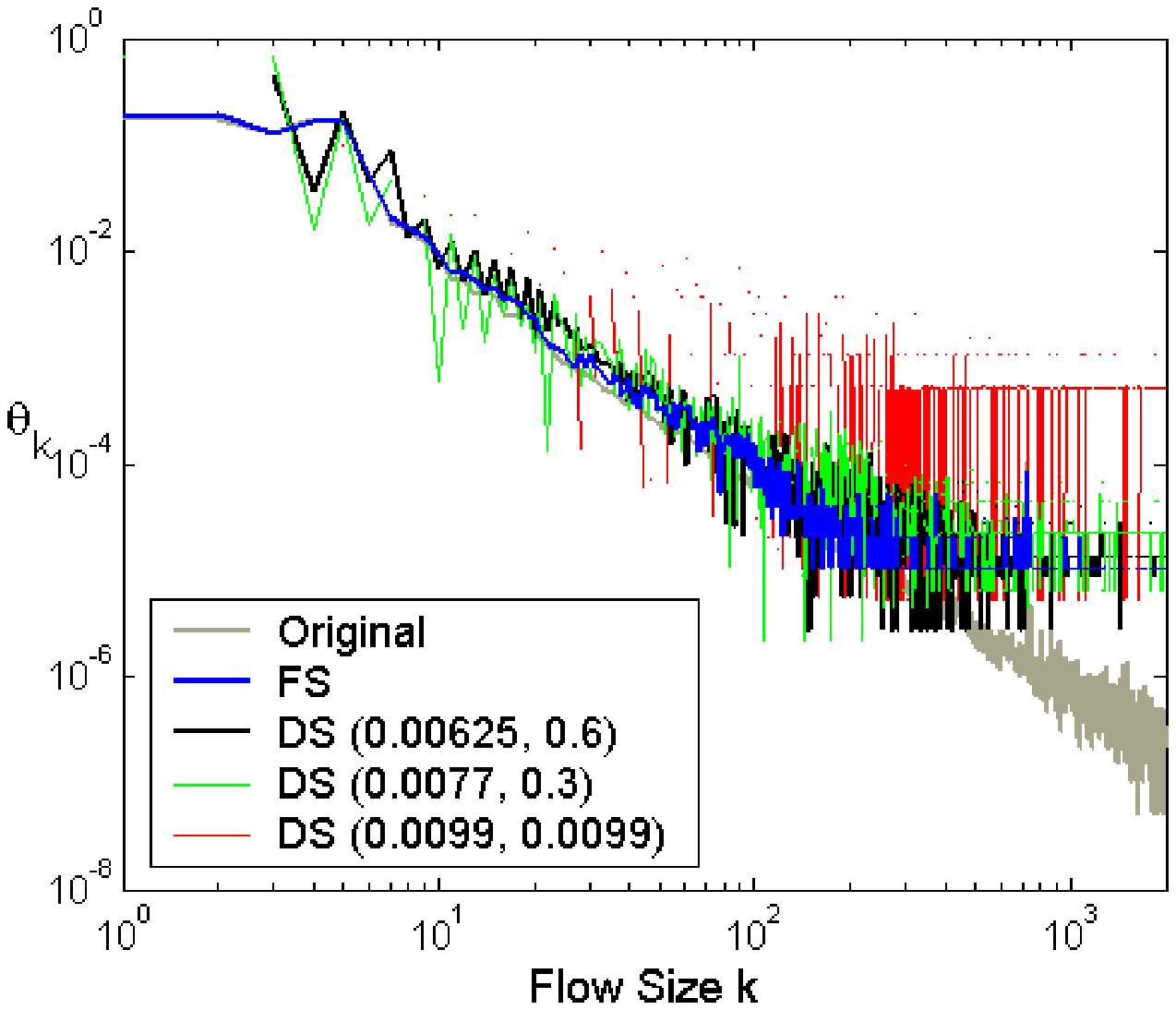}}
\def\ESRabish			{\includegraphics[width=\oneup]{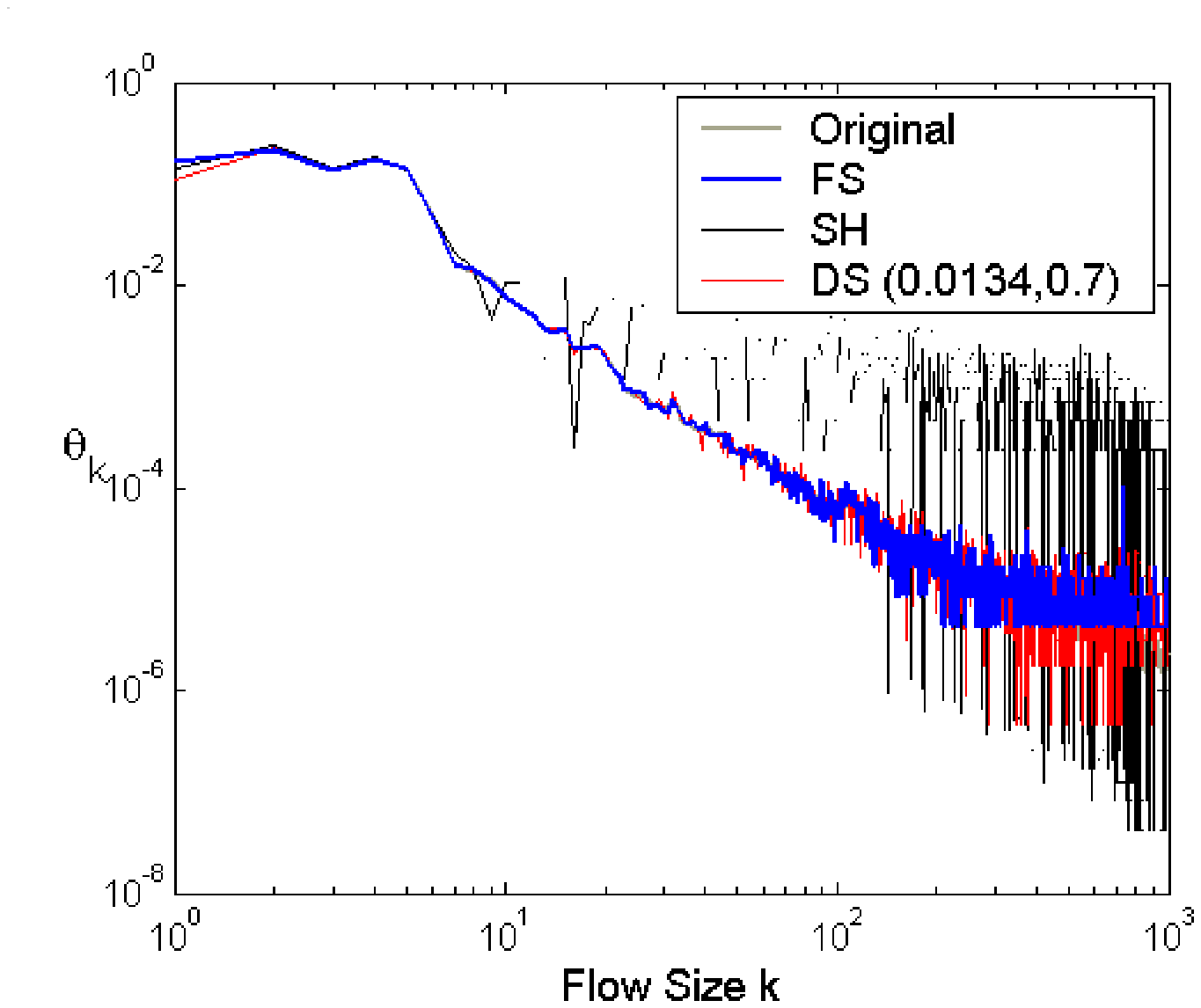}}
\def\ESRabicrlb			{\includegraphics[width=\oneup]{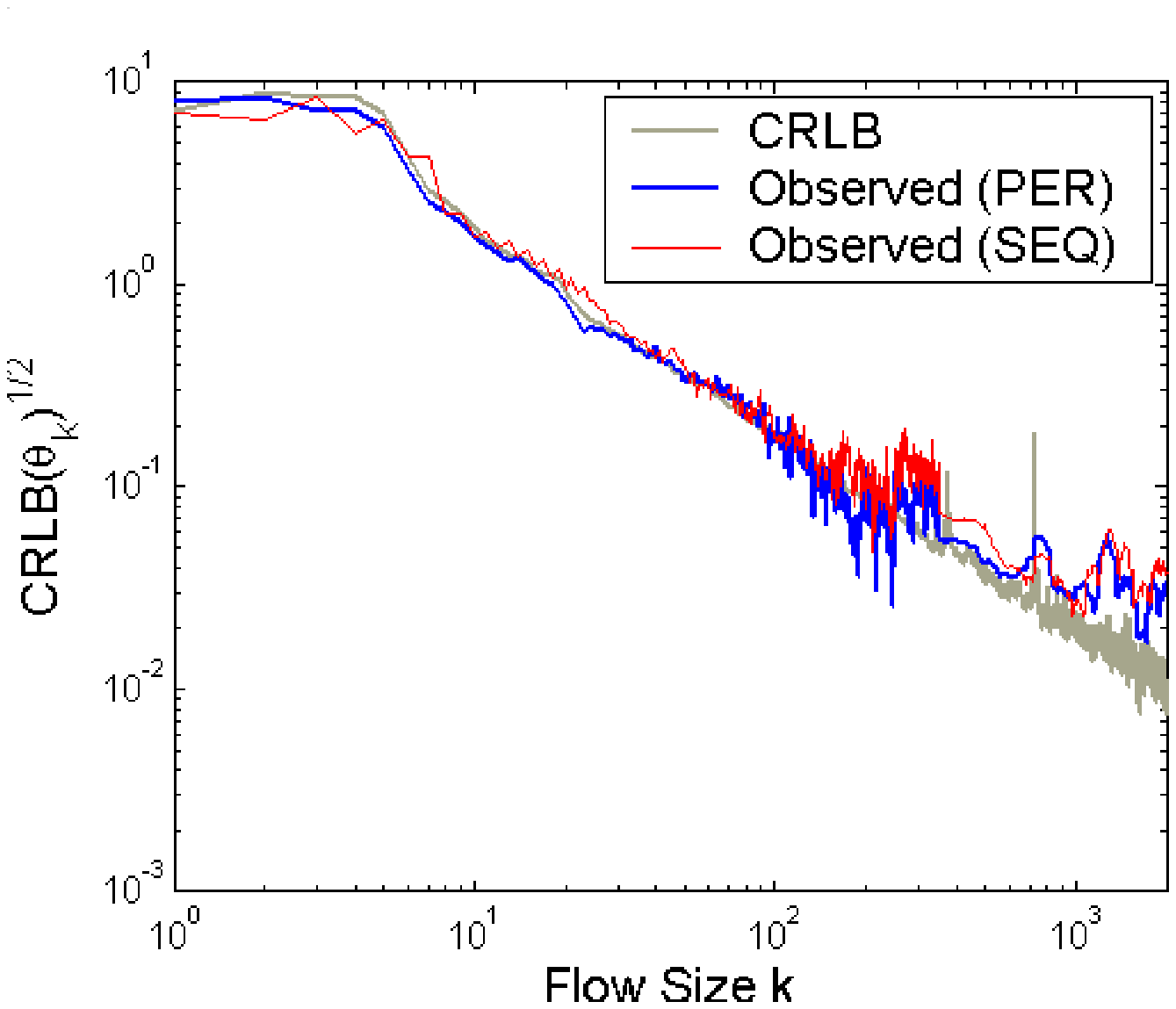}}

\begin{document}
\title{Fisher Information in Flow Size Distribution Estimation}

\author{Paul~Tune,~\IEEEmembership{Member,~IEEE,} and Darryl~Veitch,
\IEEEmembership{Fellow,~IEEE}
\thanks{The authors are with the 
Department of E\&E Engineering, The University of Melbourne, Australia (Email:
\{lsptune@ee.,dveitch@\}unimelb.edu.au).}}
\maketitle

\begin{abstract}
The flow size distribution is a useful metric for traffic modeling and management.
Its estimation based on sampled data, however, is problematic. Previous work has shown 
that flow sampling (FS) offers enormous statistical benefits over packet sampling but high
resource requirements precludes its use in routers. We present Dual Sampling (DS), a
two-parameter family, which, to a large extent, provide FS-like statistical performance by
approaching FS continuously, with just packet-sampling-like computational cost. Our work
utilizes a Fisher information based approach recently used to evaluate a number of
sampling schemes, excluding FS, for TCP flows. We revise and extend the approach 
to make rigorous and fair comparisons between FS, DS and others. We show how DS 
significantly outperforms other packet based methods, including Sample and Hold, the 
closest packet sampling-based competitor to FS. We describe a packet sampling-based
implementation of DS and analyze its key computational costs to show that router 
implementation is feasible. Our approach offers insights into numerous issues, including
the notion of `flow quality' for understanding the relative performance of methods, and 
how and when employing sequence numbers is beneficial. Our work is theoretical with some
simulation support and case studies on Internet data.
\end{abstract}


\section{Introduction}
\label{sec:intro}

The distribution of {\em flow size}, that is the number of packets in
a flow, is a useful metric for traffic modelling and management, and is important for
security because of the role small flows play in attacks. As is now well known however,
its estimation based on sampled data is problematic.

Currently, sampling decisions in routers are made on a per-packet basis, with only
sampled packets being subsequently assembled into (sampled) flows. Duffield
et~al.~\cite{duffsamplingToN2005} were the first to point out that simple \textit{packet
sampling} strategies such as `$1$ in $N$' periodic or i.i.d. (independent, identically
distributed) packet sampling have severe limitations, in particular a strong flow length
bias which allows the tail of the flow size distribution to be recovered, but
dramatically obscures the details of small flows. They explored the use of TCP SYN
packets to improve the resolution at the small flow end of the spectrum. Hohn
et~al.~\cite{thinning,thinningToN} explored these difficulties further and pointed out
that \textit{flow sampling}, where the sampling decision is made directly on flows,
resulting in all packets belonging to any sampled flows being collected, has enormous
statistical advantages. However, flow sampling has not been pursued further nor found
its way into routers, partly because it implies that lookups be performed on every packet,
which is very resource intensive.

More recently, Ribeiro et~al.~\cite{Ribeiro06fischer} explored the use of TCP sequence
numbers to improve estimation for TCP flows. The idea is that the presence of packets
which are not physically sampled can be inferred by noting the increasing byte count
given by the sequence number fields of sampled packets.  By using the Fisher information
as a metric of the effectiveness of sampling in retaining information about the original
flow sizes, they showed that this helps greatly to `fill in the holes' left by packet
sampling. However, they did not address whether these techniques out-perform flow
sampling (FS).

In this paper we revisit FS in the context of TCP flows. Our first contribution is to
explain how the approach of \cite{Ribeiro06fischer} can be reformulated and extended to 
include FS. This provides a framework for our second contribution, proofs that FS
outperforms existing methods by a large margin, though they themselves greatly improve
upon simple packet sampling. Our results are rigorous, based on explicit calculation and 
comparison of the Fisher Information matrices of competing schemes. With the statistical
reputation of FS thus reinforced, the challenge is to find methods which can somehow
approach or approximate flow sampling in order to benefit from its information theoretic
efficiency, but with lower resources requirements. We show how this can be done.

The computational problem for FS can be described as follows. To capture the variety
present in traffic flows and to provide the raw material for a variety of current (and
future) metrics, many flows must be sampled. This implies large flow tables which in turn
implies the use of slower but cheaper DRAM rather than the faster but expensive 
SRAM~\cite{VargheseNA}. However, DRAM is not fast enough to perform lookups for every
packet, as required by a straightforward implementation of FS, for today's high
capacity links. The question then becomes, how can flow sampling be implemented using
per-packet decisions, in other words using some form of \textit{packet} sampling?

The main contribution of this paper is the introduction of \textit{Dual Sampling} (DS), a
hybrid approach combining the advantages of both packet and flow sampling.
It is a two parameter sampling family which includes FS as a special case and allows FS
to be approached continuously, enabling a tradeoff of sampling efficiency against
computational cost. Computationally, it can be implemented via a modified form of
two-speed or `dual' packet sampling which circumvents the problem of slow DRAM. There
is a cost in terms of wasted samples, but we show that this can be borne in high speed
routers. Following \cite{Ribeiro06fischer}, DS benefits from the use of TCP sequence
numbers although it can also be used without them, and we provide insight into how and
when they have an impact. We show rigorously that DS outperforms the methods proposed in
\cite{Ribeiro06fischer}. We also compare and contrast DS with the well known `Sample and
Hold' scheme \cite{Estan03}. We show that Sample and Hold performs quite well, though not
as well as DS.

Finally, we introduce \textit{SYN+SEQ+FIN}, another sampling method which enables flow
sampling to be perfectly achieved (aside from errors in the mapping of byte to packet
counts) at very low computational cost, well below that even of packet sampling.
Its disadvantage is that it exploits the TCP FIN field, when not all TCP flows terminate
correctly with a FIN packet.

With its explicit use of TCP protocol information in most cases, our work applies to
TCP flows only.  However, the ideas and results could apply to other kinds of flows
provided that suitable substitutes could be found for connection startup (SYN),
`progress' (sequence numbers) and termination (FIN). TCP flows still constitute the
overwhelming majority of traffic in the Internet.

The rest of the paper is organized as follows. Section~\ref{sec:framework} describes our
sampling framework and derives the Fisher information matrix and its inverse explicitly.
Section~\ref{sec:methods} defines the sampling methods and derives their main properties. 
Section~\ref{sec:compare} compares the methods theoretically and derives further
properties explaining their performance, and Section~\ref{sec:seq_div} explores in more
detail how sequence numbers reduce estimation variance. Section~\ref{sec:compcompare}
introduces a simple model for computational cost and uses it to define and solve an 
optimization problem for sampling performance under constraints.  Section~\ref{sec:data}
applies the methods to real Internet data and shows that DS performs favorably with flow 
sampling in practice, and has better performance than Sample and Hold. A closed-form
unbiased estimator was proposed for DS and Sample and Hold which achieves the Cram\'er-Rao 
lower bound asymptotically, eliminating the need for iterative optimization algorithms. We 
conclude and discuss future work in Section~\ref{sec:conc}.

This paper is an extended and enhanced version of the conference paper
\cite{Fisher_IMC08}. The main additions relate to the inclusion of the Sample and Hold
method throughout the paper, several new theorems and counter-examples on method
comparison, and the inclusion of a new major data set.

\section{The Sampling Framework}
\label{sec:framework}

In this section we establish a framework to define and analyze sampling techniques
applied to an idealized view of TCP flows on a link. Nominally, we imagine that such
flows are defined by the usual 5-tuple of origin and destination IP addresses, port
numbers, and TCP protocol field together with a timeout. For the analysis we make a
number of simplifying assumptions:
\begin{enumerate}
\item flows begins with a SYN packet and have no others,
\item flows are not split (this can occur through timeouts or flow table clearing),
\item all necessary protocol information (5-tuple, SYN/FIN bits and sequence numbers)
can be observed, and
\item per-flow sequence numbers count packets, not bytes.
\end{enumerate}
Assumptions (iii) and (iv) will be discussed/relaxed when we deal with real data in
Section~\ref{sec:data}. Note that we \textbf{do} respect TCP's per-flow random
initialization of sequence numbers. Hence their absolute value holds no information on
the number of packets in a flow, only differences of sequence numbers matter.
This is crucial for the analysis.

\subsection{The Flow Model}
\label{ssec:flows}

We consider a measurement interval containing $\nf$ flows. Let $\si$ denote the
\textit{size} of flow $i$ (the number of packets it contains). It satisfies $1\le \si
\le W$, where $1 \le W<\infty$ is the maximum flow size. The total number of packets is
$\np=\sum^{\nf}_{i=1} \si$.

Let $M_j$ be the number of flows of size $j$, $1\le j\le W$, that is $M_j = \sum_{i:
\si=j} \Count$, and $\nf=\sum^W_{j=1} M_j$. The flow size `distribution' is the set
$\vth = \{\theta_1, \theta_2, \ldots, \theta_W\}$  of relative frequencies, that is
\begin{equation}
  \theta_j = \frac{M_j}{\nf}
  \label{theta_counts}
\end{equation}
where $0 \le \theta_j \le 1$ and $\sum^W_{j=1} \theta_j = 1$. Note that $\{M_j\}$ and
$\{\vth,\nf\}$ are equivalent and complete descriptions of the flows in the
measurement interval. They are sets of deterministic \textit{parameters}, not random
variables, effectively a deterministic flow size model.

Most of the literature on traffic sampling follows the above viewpoint, where the data is
deterministic, the only randomness being introduced through the sampling itself.
An exception is the work of Hohn and Veitch (\cite{thinning,thinningToN}) where randomness arises 
both through the traffic model and the sampling process, which makes the analysis considerably 
more difficult, but less generic.

\subsection{General Random Sampling}

The effect on flow size of random sampling can be described as follows: from an original
flow of size $k$, only $j$ packets, $0\le j\le k$, are \textit{sampled} (retained), with
probability
\[
     \bjk = \Pr(\textrm{sampled flow has $j$ pkts}\,|\,\textrm{original flow has $k$ pkts}).
\]
The operation of the sampling scheme is entirely defined by the $\bjk$, which can
be assembled into a $(W+1) \times W$ \textit{sampling matrix}  $\mathbf{B}$, whose
$(j+1,k)$-th element is $b_{jk}$.  Note that $\bjk = 0$ for $j > k$. By definition
$\mathbf{B}$ is a {\em (column) stochastic matrix}, that is each element obeys
$\bjk\ge0$, and each column sums to unity.

The experimental outcome can be described by
a set of random variables $\{M'_j\ |\ 0 \le j \le W\}$ where
$M'_j$ counts the number of sampled flows of size $j$. Thus,
\begin{equation*}
     M'_j = \sum_{i:\si'=j} \Count = \sum^{\nf}_{i=1} \ind(\si'=j).
\end{equation*}
Equivalently, let $\vth' = \{\theta'_j\ |\ 0 \le j \le W\}$ 
(note that the index $j$ includes 0) denote the empirical distribution of 
\textit{sampled flow} sizes, where
\begin{equation}
    \theta'_j = \frac{M'_j}{\nf},
    \label{theta_prime_counts}
\end{equation}
where $0 \le \theta'_j \le 1$ and $\sum^W_{j=0} \theta'_j = 1$. Note that $\theta'_0 \ge
0$ as some flows may not survive the thinning process.
Out of the original $\nf$ flows, only $\nf' =\nf(1-\theta'_0) = \nf - M'_0$ flows survive sampling.

Define a normalized set of fractions of the sampled flow sizes $\vgm =
\{\gamma_j\ |\ 1 \le j \le W\}$ by
\begin{equation}
	\gamma_j = \frac{\theta'_j}{\sum^W_{k=1} \theta'_k} = \frac{\theta'_j}{1 -
	\theta'_0},
	\label{gamma_norm}
\end{equation}
where $0 \le \gamma_j \le 1$ and $\sum^W_{j=1} \gamma_j = 1$. The set
$\vgm$ constitutes the \naive or directly measurable sampled flow
size distribution and is equivalent to the distribution of flows conditional on at
least a single packet from that flow being sampled. For $j \ge 1$, $\theta'_j$ is
related to $\gamma_j$ by $(1-\theta'_0)\gamma_j = (\nf'/\nf)\gamma_j = \theta'_j$.

\subsection{The Unconditional Formulation}

The sampled flow above includes the case, $j=0$, where the flow `evaporates'.
It seems natural to conclude however that such cases cannot be observed. This logically
leads to an analysis based on the observation of the $\gamma_j$ defined above
where  $j\ge1$, which is effectively
\textit{conditional}: sample flow distributions given that at least one packet is
sampled. This is the approach adopted in
\cite{duffsamplingToN2005,Ribeiro06fischer,YangMichailidisGlobe06} and in the literature
generally.  One of the key differences in our work is that we show that it is possible
to observe the $j=0$ case, leading to an \textit{unconditional} formulation which enjoys
many advantages.

To see how this is possible we return to general context of $\nf$ flows, \textit{each}
one of which will be sampled in this general sense.  As defined above $\nf'$ is the number of flows
of size at least 1 after sampling. The number of evaporated flows is just $\nf-\nf'$,
but typically $\nf$ is not known and is regarded as a `nuisance parameter' which must be
estimated. However, it can easily be measured by directly counting the total number of SYN
packets, which is just the number of flows $\nf$.  For methods which are already assuming an
ability to access and perform specific actions based on whether a packet is a SYN
or not, the additional assumption of being able to count all SYN packets is a natural one. It is also implementable, as a single additional counter which checks every packet and
conditionally increments based on a small number of header bits is not difficult
even at the highest speeds \cite{VargheseNA}, as we discuss in more detail in Section~\ref{sec:compcompare}.
In summary, by knowing $\nf$, every flow gives rise to a sampled flow, each one of which is
observable, either directly ($j\ge1$), or indirectly ($j=0$).
In other words, we can in effect observe the $\theta'_j$ over the full range from $j=0$ to $j=W$.

The chief advantage of the unconditional formulation is the very simple form of the
likelihood function for the experimental outcome $j$ for a single flow. This makes the
manipulation of the Fisher information far more tractable, leading to new analytic
results and insights. The other big advantage is that flow sampling can now be
included. In the conditional world flow sampling is perfect -- by definition, if a flow
is sampled at all, all its packets will be and so there is nothing to do! The
unconditional framework allows the missing part of the picture to be included,
enabling meaningful comparison.

\subsection{The Sampled Flow Distribution}

Our analysis is based on the idea of selecting a `typical' flow, and that flows are
mutually independent (a reasonable assumption if $\nf$ is very large). Since flows are
in fact deterministic, this is only meaningful if we introduce a supplementary random
variable $U$, a uniform over the $\nf$ flows available, which performs the random flow
selection. This variable, which acts `invisibly' behind the scenes (and is rarely
discussed), is not part of the random sampling scheme itself, but is essential as it
allows the $\vth$ parameters to be treated as probabilities, even though they are
not.
\begin{figure}[t]
	\centering
	\includegraphics[width=8cm, height=5.3cm]{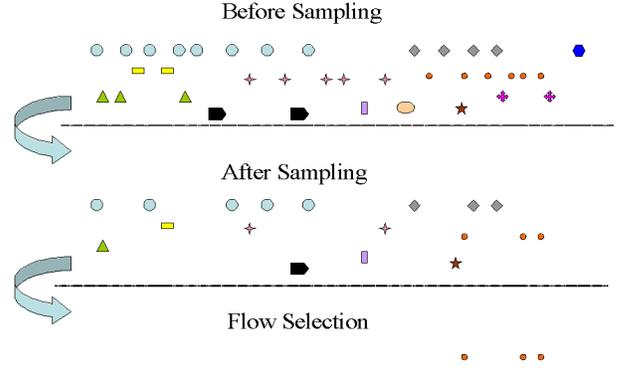}
	\caption{The flow sampling and selection process. Here a flow is selected which had $k=5$ packets originally and $j=3$ after sampling.}
	\label{fig:flow_sel}
\end{figure}
An example is given in Figure~\ref{fig:flow_sel} which shows $\nf=12$ flows before and after sampling, followed by a random flow selection.   In the interests of clarity a flow of size $j=3$ after sampling was selected, but it could have  been one of the evaporated flows ($j=0$).

With this background established, the discrete distribution for a sampled flow
originally of size $k$ is very simple:
\be
   d_j = \sum^W_{k=1} b_{jk}\theta_k , \quad 0\le j\le W.
   \label{eq:model}
\ee
This can be expressed in matrix notation as
\begin{equation}
  \mathbf{d} = \mathbf{B}\vth
  \label{eq:vec_model}
\end{equation}
where $\mathbf{d} = [d_0,d_1,d_2,\ldots,d_W]^\T$ is a $(W\!+\!1)\times 1$ column
vector, and $\vth$ a $W\times 1$ column vector. The probability $d_j$ is related
to the empirical fraction $\theta'_j$ for $j \ge 0$ by
\begin{align*}
	\E\lbrack \theta'_j \rbrack
    & = \frac{\E \lbrack \sum_{i:\si'=j} \Count \rbrack}{\nf}
	= \frac{\sum^{\nf}_{i=1} \E \lbrack \ind(\si'=j) \rbrack}{\nf} \\
    & = \frac{\sum^{\nf}_{i=1} \Pr(m'_i = j)}{\nf}
	= \frac{\nf d_j}{\nf} = d_j .
\end{align*}
The likelihood function for the parameters is simply
\begin{equation}
    f(j;\vth) = d_j , \quad 0\le j\le W.
    \label{eq:likelihood_fn}
\end{equation}
In the conditional framework commonly used $j=0$ is missing, and
normalization is then needed to ensure probabilities add to one. This
implies a division of random variables, which greatly complicates the likelihood.

\subsection{The Fisher Information of a Sampled Flow}
\label{ssec:Fisher}

The parameter vector $\vth$ is the unknown we would like to estimate from sampled flows.
Since here we are not concerned with specific estimators of $\vth$, but in the
effectiveness of the underlying sampling scheme, a powerful approach (introduced in
\cite{Ribeiro06fischer}) is to use the \textit{Fisher information} \cite[Section 11.10]
{CoverThomas06} to access its efficiency in collecting information about $\vth$.

We first introduce notation that will be used throughout this paper.
The expectation of a random variable $X$ is denoted by $\E\lbrack X \rbrack$, and the 
variance by $\mathrm{Var}(X)$. Matrices are written in bold-face upper case and vectors in
bold-face lower case. The \textit{transpose} of a matrix $\mathbf{A}$ is denoted by
$\mathbf{A}^{\!\T}$. The operator also applies to vectors. The operator $\tr(\mathbf{A})$
denotes the trace of the matrix $\mathbf{A}$. The matrix $\mathbf{I}_n$ denotes the $n
\times n$ identity matrix. The vector $\mathbf{1}_n= \lbrack 1,1,\ldots,1 \rbrack^\T$
denotes an $n \times 1$ vector of 1s. The vector $\mathbf{0}_n$ denotes the $n \times 1$
null vector, and the $m \times n$ null matrix is written as $\mathbf{0}_{m \times n}$.
Given an $n \times 1$ vector $\mathbf{x}$, $\diag(\mathbf{x})$ denotes an $n \times n$
matrix with diagonal entries $x_1, x_2,\ldots,x_n$.

\vspace{1mm}
\begin{defn}
An $n \times n$ real matrix $\mathbf{M}$ is \textit{positive definite} iff for all vectors
$\mathbf{z} \in \Real^n \backslash \{\mathbf{0}_n\}$,
$\mathbf{z}^\T \mathbf{M}\mathbf{z} > 0$, and is \textit{positive semidefinite} iff
$\mathbf{z}^\T \mathbf{M}\mathbf{z} \ge 0$.
\label{def:positive_def}
\end{defn}
\vspace{1mm}
We write $\mathbf{A} > 0$ or $0<\mathbf{A}$ to indicate that $\mathbf{A}$ is positive 
definite. For two matrices $\mathbf{A}$ and $\mathbf{B}$, we write $\mathbf{A}>\mathbf{B}$
to mean $\mathbf{A}-\mathbf{B}>0$ in the positive definite sense. Similarly, $\mathbf{A}
\ge\mathbf{B}$ and $\mathbf{A} - \mathbf{B} \ge 0$ each mean that $\mathbf{A} -
\mathbf{B}$ is positive semidefinite. The operator $|\cdot|$ returns the size of a vector
or set. All other definitions will be defined when needed.


The Fisher information is useful because its inverse is the Cram\'er-Rao lower bound
(CRLB), which lower-bounds the variance of any unbiased estimator of $\vth$. In fact the
Fisher information takes a different form depending on whether constraints are imposed
on the $\vth$ or not \cite{Hero96}. Inequality constraints are particularly problematic,
so we avoid them by assuming that each $\theta_k$ obeys $0<\theta_k<1$ (this ensures
that the CRLB optimal solution cannot include boundary values, which would create bias
and thereby invalidate the use of the unbiased CRLB). Assuming that flows exist for all
sizes, i.e.~that $\theta_k>0$ for all $k$, is reasonable given the huge number of
simultaneously active flows (up to a million) in high end routers. There is one more
constraint, the equality constraint $\sum_{k=1}^W \theta_k =1$, which must be included.
As this complicates the Fisher information, we first deal with the unconstrained case.

\subsection{The Unconstrained Fisher Information}

The Fisher information is based on the likelihood and is defined by
\begin{eqnarray}
  \nonumber
  \uncon(\vth) &=& \E \lbrack (\nabla_{\vth} \log f(j;\vth))
(\nabla_{\vth} \log f(j;\vth))^\T \rbrack  \\
                      &=& \sum^W_{j=0} (\nabla_{\vth} \log f(j;\vth))
(\nabla_{\vth} \log f(j;\vth))^\T d_j .
  \label{matrix_j}
\end{eqnarray}
Here $\nabla_{\vth} \log f(j;\vth) = (1/d_j)\lbrack b_{j1}, \ldots,
b_{jW} \rbrack^\T$ because of the simple form (\ref{eq:likelihood_fn}) of the
likelihood. This leads to the simple explicit expression
$(\uncon(\vth))_{ik} =\sum^W_{j=0} \frac{b_{ji}b_{jk}}{d_j}$, or equivalently
\be
   \uncon(\vth) = \mathbf{B}^\T \mathbf{D}(\vth)\mathbf{B}
   \label{eq:mtx_j_decomp}
\ee
where $\mathbf{D}(\vth)$  is a diagonal matrix with $(\mathbf{D}(\vth))_{jj} =
d^{-1}_{j-1}$.

We will need to find the inverse of $\uncon$, but since $\mathbf{B}$ is not square,
this cannot be done directly from (\ref{eq:mtx_j_decomp}) in terms of the inverse of
$\mathbf{B}$. However if we re-express $\mathbf{B}$ as
\be
\mathbf{B} =
  \begin{bmatrix}
       \mathbf{b}^\T_0 \\
       \mathbf{\tilde B}
  \end{bmatrix}
  \label{eq:bB}
\ee
where $\mathbf{b}_0^\T = \lbrack b_{01},\ldots,b_{0W} \rbrack$ is the top row of
$\mathbf{B}$ and
$$
\mathbf{\tilde B} =
\begin{bmatrix}
b_{11}       & b_{12}   & \cdots        & b_{1W}\\
0            & b_{22}   & \cdots        & b_{2W}\\
\vdots       & \vdots   & \ddots        & \vdots\\
0            & 0        & \cdots        & b_{WW}
\end{bmatrix},
$$
then we can write
\begin{equation}
  \uncon(\vth) = \frac{1}{d_0} \mathbf{b}_0 \bt + \tilde \uncon(\vth)
  \label{eq:mtx_j_alt}
\end{equation}
where $\tilde \uncon(\vth) = \mathbf{\tilde B}^\T \mathbf{\tilde D}(\vth) \mathbf{\tilde B}$,
 $\mathbf{\tilde D}(\vth) = \diag(d^{-1}_1, \ldots,d^{-1}_W)$, and $d_0=\bt\vth$.
Since $\Bt$ and $\Dt(\vth)$ are square, the inverse of $\tilde \uncon(\vth)$ is just
$\tilde \uncon^{-1}(\vth) = \Bt^{-1} \Dt^{-1}(\vth) (\Bt^{-1})^T$, that is
\be
   (\tilde\uncon^{-1}(\vth))_{ik} =\sum^W_{j=1} b'_{ij}b'_{kj}{d_j}
   \label{eqn:Jtildeinv}
\ee
where $b'_{jk}=(\Bt^{-1})_{jk}$.
By Lemma~\ref{thm:pos_def_prod}(ii) in Appendix \ref{app:other_lemmas}, $\tilde 
\uncon(\vth)$ is positive definite. We can now give the inverse of $\Jt$.
\smallskip

\begin{prop}
The inverse of $\uncon(\vth)$ is given by
\begin{equation*}
  \uncon^{-1}(\vth) = \tilde \uncon^{-1}(\vth) -
  \frac{1}{d_0 + \bt \tilde \uncon^{-1}(\vth) \mathbf{b}_0}
  \tilde \uncon^{-1}(\vth) \mathbf{b}_0 \bt \tilde \uncon^{-1}(\vth).
\end{equation*}
\label{prop:sub}
\end{prop}
\begin{proof}
The matrix inversion lemma applies (see Lemma~\ref{lem:matrix_inversion} in Appendix 
\ref{sec:our_lemmas}) with $\mathbf{R}=\Jt$ and $\mathbf{T}=1/{d_0}$ nonsingular. Since $d_0 + 
\bt \tilde \uncon^{-1}(\vth) \mathbf{b}_0>\bt \tilde \uncon^{-1}(\vth) \mathbf{b}_0 >0$ as $\Jt^{-1}$ is 
positive definite, the result immediately follows.
\end{proof}
The diagonal elements of the matrix $\uncon^{-1}(\vth)$ will be important in later sections,
and the explicit formula is given below:
\be
(\uncon^{-1}(\vth))_{jj} = \sum^W_{k=1} b'^2_{jk}{d_k} - \frac{\Big( \sum_{k=j}^W d_k b'_{jk} 
\sum_{\ell=1}^k b_{0\ell} b'_{\ell k}\Big)^2}{d_0 + \sum_{k=1}^W d_k \Big(\sum_{\ell =1}^k 
b_{0\ell} b'_{\ell k}\Big)^2}. 
\label{eq:gen_inverse_diag}
\ee

Again, this explicit inverse, valid for any general sampling matrix $\mathbf{B}$,
is made possible by the very simple form of the likelihood function in equation 
(\ref{eq:likelihood_fn}). We now specialize the above result for sampling matrices that satisfy 
particular conditions. Although some of these matrices exhibit a dependence on $\vth$, we
drop this dependence for notational simplicity when the context is clear.

\smallskip
\begin{cor}
If for some constant $q$  the sampling matrix $\mathbf{B}$ satisfies $\mathbf{b}_0 = q \mathbf{1}_W$
then (setting $p=1-q$)
\begin{align*}
  \uncon^{-1} = \tilde \uncon^{-1} - \frac{q}{p}\vth \vth^\T.
\end{align*}
\label{cor:inv_uncon_syn}
\end{cor}
\begin{proof}
\noindent The given condition implies that $d_0=\bt\vth=q\mathbf{1}^T_{\!W}\vth=q$, and
$\onew^\T \mathbf{\tilde B}^{-1}  = (1/p) \onew^\T$ since $\mathbf{B}$ is column stochastic.
Next,
\begin{align*}
\mathbf{b}_0^\T \tilde \uncon^{-1} \mathbf{b}_0 
	 &= q^2 \mathbf{1}^\T_W \tilde \uncon^{-1} \mathbf{1}_W 
      = q^2 \mathbf{1}^\T_W \mathbf{\tilde B}^{-1} \mathbf{\tilde D}^{-1}
(\mathbf{\tilde B}^{-1})^T \mathbf{1}_W \\
     & = (q^2/p^2) \mathbf{1}^\T_W \mathbf{\tilde D}^{-1} \mathbf{1}_W \\
     & = (q^2/p^2) \sum^W_{j=1} d_j = (q^2/p^2) (1 - d_0) = q^2/p.
\end{align*}

Let $\tilde{\dt} = \mathbf{\tilde B} \vth= \lbrack d_1, d_2, \ldots d_W \rbrack^\T$.
Then from Proposition~\ref{prop:sub},
\begin{align*}
\uncon^{-1}
& = \tilde \uncon^{-1} - \frac{1}{d_0 + q^2/p}\tilde
\uncon^{-1}\mathbf{b}_0 \mathbf{b}_0^\T \tilde \uncon^{-1}\\
&= \tilde \uncon^{-1} - \frac{q^2/p^2}{q + q^2/p}
\mathbf{\tilde B}^{-1} \mathbf{\tilde D}^{-1} \mathbf{1}_W
\mathbf{1}^\T_W \mathbf{\tilde D}^{-1} (\mathbf{\tilde B}^\T)^{-1} \\
& = \tilde \uncon^{-1} - \frac{q^2/p^2}{q + q^2/p} \mathbf{\tilde B}^{-1}
\tilde\dt \tilde\dt^\T (\mathbf{\tilde B}^\T)^{-1} \\
& = \tilde \uncon^{-1} - \frac{q}{p} \vth \vth^\T.
\end{align*}
\end{proof}

\smallskip
The matrix $\tilde\uncon$ corresponds to the information carried by the outcomes $1\le
j\le W$ only.  We expect $\uncon$ to carry more information through the knowledge of
$\nf$ which gives access to $j=0$, and therefore  $\uncon^{-1}$  to have reduced uncertainty, corresponding (through the CRLB) to a reduced variance.
The following result confirms this intuition (proof in Appendix \ref{sec:our_lemmas}).

\smallskip
\begin{thm}
An upper bound for $\uncon^{-1}(\vth)$ is
\begin{equation*}
   \uncon^{-1}(\vth) \le \tilde \uncon^{-1}(\vth).
\end{equation*}
Equality holds if and only if $\mathbf{b}_0 = \mathbf{0}_{W \times 1}$.
\label{thm:upper_bound_inv}
\end{thm}
\smallskip
The reduction in uncertainty is given by the second term in the expression for
$\uncon^{-1}(\vth)$ in Proposition~\ref{prop:sub}.

\subsection{The Constrained Fisher Information and CRLB}
\label{ssec:constrained_fim}

Intuitively, constraints on the parameters should increase the Fisher information since
they tell us something more about them, `for free'. In fact, \cite{Gorman90} shows that
this is only true for equality constraints. Since we are assuming that $0<\theta_k<1$, the
only active constraint is $\sum_{k=1}^W \theta_k =1$. Its \textit{gradient} is
\be
   \constr (\vth) = \nabla_{\vth} g(\vth)
   \label{gradient_matrix}
\ee
where $g(\vth) = \sum^W_{j=1} \theta_j - 1$.

\vspace{1mm}
The inverse constrained Fisher information \cite{Gorman90} is
\be
   \con^{+} = \uncon^{-1} - \uncon^{-1} \constr \left(\constr^\T \uncon^{-1}
 \constr \right)^{-1} \constr^\T \uncon^{-1}
   \label{eq:inverse_fim}
\ee
where $\con^{+}$ denotes the Moore-Penrose pseudo-inverse \cite[Chapter 20, pp.~493-514]
{Harville97} of the constrained Fisher information matrix $\con$. The matrix $\con^{+}$
is rank $W-1$ due to the single equality constraint and is thus singular (see
\cite[Remark 2]{Gorman90}). This somewhat formidable expression can be simplified in our
case, as we now show.

\smallskip
\begin{lem}
$\uncon\, \mathbf{diag}(\theta_1,\ldots,\theta_W)\onew = \onew$.
\label{lem:prod_rowstoch}
\end{lem}
\begin{proof}
The row sum of row $i$ of $\uncon\,\mathbf{diag}(\theta_1, \ldots,\theta_W)$ is
\be
  \nonumber
  \sum^W_{k=1} \sum^W_{j=0} \frac{b_{ji}b_{jk}}{d_j} \theta_k
    = \sum^W_{j=0} \frac{b_{ji}}{d_j} \sum^W_{k=1} b_{jk}\theta_k
    = \sum^W_{j=0} \frac{b_{ji}}{d_j} d_j = 1
\ee
since $\mathbf{B}$ is column stochastic.
\end{proof}

It is easy to see that $\constr = \mathbf{1}_W$. Hence 
$\uncon^{-1}\constr = \uncon^{-1}\onew = \mathbf{diag}(\theta_1,\ldots,\theta_W)\onew = \vth$ 
from Lemma~\ref{lem:prod_rowstoch}. It is then straightforward to verify that
$\left(\constr^\T \uncon^{-1} \constr \right)$ is simply the number $1$, and further
that (\ref{eq:inverse_fim}) can be reduced to
\be
   \con^{+} = \uncon^{-1} - \vth\vth^\T.
   \label{eq:inverse_simple}
\ee
The remarkable thing here is that the constraint term $\vth\vth^\T$ depends on $\vth$
only, and so is constant for all sampling matrices $\mathbf{B}$, a great advantage when
comparing different methods.

Since we are assuming flows are sampled independently, the Fisher information for $N$
flows is just $N\uncon$, and the inverse becomes $\con^{+}/N$. For any unbiased
estimator $\boldsymbol{\hat\theta}$ of $\vth$, the CRLB then states that
\be
   E[(\boldsymbol{\hat\theta} - \vth)(\boldsymbol{\hat\theta} - \vth)^{\mathrm{T}}] \ge
\frac{\con^{+}(\vth)}{N} .
   \label{crlb}
\ee
Because of independence we study $N=1$. In practice all flows are sampled and so $N=\nf$.
\begin{rem}
There is an interpretation to the simple structure of the matrix $\mathbf{P} =
\uncon^{-1} \constr (\constr^\T \uncon^{-1} \constr )^{-1} \constr^\T
\uncon^{-1}$. The scalar value $\onew^\T \con^{+} \onew$ is equivalent to $\mathrm{Var}
(\sum^W_{i=1} \hat \theta_i)$. By the equality constraint, we expect the best estimator
of $\sum^W_{i=1} \hat \theta_i$ to have a variance of 0, since the estimator already
knows that $\sum^W_{i=1} \theta_i=1$.
This corresponds to a CRLB of zero, namely
$\onew^T \con^{+} \onew = \onew^T\uncon^{-1}\onew - \onew^T\vth \vth^\T \onew
                                         = \onew^T \mathbf{diag}(\theta_1,\ldots,\theta_W)\onew -1 = 0$.
Thus, $\mathbf{P}=\vth \vth^\T$ is the form of the correction term to the unconstrained covariance matrix $\uncon^{-1}$ needed to satisfy the constraint.
\end{rem}

\section{The Sampling Methods}
\label{sec:methods}

In this section we define the sampling methods we consider and derive
their main properties.
We begin with methods which have been described elsewhere, including
simple packet and flow sampling, as well as others
exploiting protocol information, in particular those proposed in
\cite{Ribeiro06fischer,duffsamplingToN2005}. Apart from their inherent interest, we
revisit these because in the unconditional framework these methods are now all
\textbf{different} to before.
More importantly, we also derive inverses analytically which has not been possible before,
and thereby obtain a number of important insights.
We also include the widely cited `Sample and Hold' \cite{Estan03} whose Fisher information has not previously been studied.  We then introduce our new method, Dual Sampling.

To better see the connection between the usual framework and ours,
recall that $b_{jk}$ is always a conditional probability with respect
to the size $k$ of the original flow. Typically however, it is also
made conditional with respect to $j$, but we do not so here. Hence, if
$\mathbf{B_c}$ is the usual $j$-conditional matrix, then
$\mathbf{B_c}\mathbf{C} = \Bt$ where $\mathbf{C} = \mathbf{I}_W - \diag(b_{01},\ldots,
b_{0W})$, i.e.~the matrix $\mathbf{C}^{-1}$ does the conditioning.

We use the decomposition of (\ref{eq:bB}) to describe each sampling matrix $\mathbf{B}$.
In each case we define $\mathbf{B}$ and $\Bt$, give the inverse $\Bt^{-1}$ of $\Bt$, and
give explicit expressions for the diagonal terms $(\uncon^{-1})_{jj}$, or in some case
for the entire inverse $\uncon^{-1}$. The importance of the diagonal terms will become
very clear in Section~\ref{sec:compare}.

\subsection{Packet Sampling (PS)}

By this we mean the simplest form of sampling, \textit{i.i.d.~packet
sampling}, where each packet is retained independently with probability $\pp$ and
otherwise dropped with $\qp = 1-\pp$. For the purpose of simplicity, we treat both
`1 in $N$' periodic sampling and i.i.d.~random sampling under the same framework, as
both methods were shown to be statistically indistinguishable in practice \cite{duffsamplingToN2005}.

The chief benefit of PS is its simplicity, and the fact that it can be implemented at
high speed because a sampling decision can be made without even inspecting the packet.
The chief disadvantage is the fact that it has a strong length bias, small flows are
very likely to evaporate.

\noindent It is easy to see that $b_{jk} = \dbinom{k}{j}\pp^j\qp^{k-j}$, or
\be
\nonumber
\mathbf{B} =
\begin{bmatrix}
\qp &   \qp^2    &  \qp^3     &   \qp^4  &    \cdots  & \qp^W  \\
\pp & 2\pp\qp    &  3\pp\qp^2 &   4\pp\qp^3  & \cdots  & \dbinom{W}{1} \pp\qp^{W-1} \\
0   &  \pp^2     &  3\pp^2\qp &   6\pp^2\qp^2  & \cdots  &\dbinom{W}{2}\pp^2\qp^{W-2} \\
\vdots   &  \vdots & \vdots & \vdots & \ddots & \vdots \\
0   &  0    &  0 & \cdots & 0 & \pp^W
\end{bmatrix}.
\ee

Before finding the inverse of $\mathbf{\tilde B}$, it is first instructive to note the
following general results by Strum \cite{Strum77}. Let $\mathcal{B}(x,y)$ be an $(W+1)
\times (W+1)$ matrix with the following structure
$$
\mathcal{B}(x,y) =
\begin{bmatrix}
1   & x                    & x^2      & \cdots        & x^W      \\
0   & y                    & 2xy      & \cdots        & \dbinom{W}{1}x^{W-1}y     \\
0   & 0                    & y^2      & \cdots        & \dbinom{W}{2}x^{W-2}y^2   \\
\vdots & \vdots & \vdots   & \ddots & \vdots                    \\
0   & 0                    & 0        & \cdots        & y^W
\end{bmatrix}
$$
which is known as the \textit{binomial matrix}. Note that $\mathcal{B}(0,1)$ reduces to
the identity matrix.  From \cite{Strum77} we have

\smallskip
\begin{lem}
If $y \ne 0$, then $\mathcal{B}(x,y)$ is invertible and
$$ \left[\mathcal{B}(x,y)\right]^{-1} = \mathcal{B}(-xy^{-1}, y^{-1}). $$
\label{lem:binom_mtxinv}
\end{lem}
Using these results we can find the inverse of $\mathbf{\tilde B}$ (recall that $b'_{jk} = 
(\mathbf{\tilde{B}}^{-1})_{jk}$.)

\smallskip
\begin{thm}
The inverse of $\mathbf{\tilde B}$ is given by
\begin{equation*}
\hspace{-2mm}\mathbf{\tilde B}^{-1} =
\begin{bmatrix}
\pp^{-1}  &  -2\pp^{-2} \qp  & 3\pp^{-3} \qp^2   & \cdots & W\pp^{-W}(-\qp)^{W-1} \\
0             &   \pp^{-2}          &-3\pp^{-3} \qp      & \cdots & \vdots \\
\vdots      &   \vdots            &  \vdots              & \ddots & \vdots \\
0    &  0  &  \cdots & \cdots & \pp^{-W}
\end{bmatrix}
\end{equation*}
that is $b'_{jk} = (-1)^{k-j}\dbinom{k}{j}\qp^{k-j}\pp^{-k}$.
\label{thm:ps_Binv}
\end{thm}
\begin{proof}
Here $\bt=\lbrack \qp, \qp^2,\ \ldots\ \qp^W \rbrack$. We have
$$
\mathcal{B}(\qp,\pp) =
\begin{bmatrix}
1                 & \bt \\
\mathbf{0}_W  & \mathbf{\tilde B}
\end{bmatrix}.
$$
Since
$\mathcal{B}(\qp,\pp) \lbrack \mathcal{B}(\qp,\pp) \rbrack ^{-1} = \mathbf{I}_{W+1}$,
for some $\mathbf{k}$ we have
$$
\begin{array}{cccc}
\begin{bmatrix}
1             & \bt \\
\mathbf{0}_W  & \mathbf{\tilde B}
\end{bmatrix}
&
\begin{bmatrix}
1             & \mathbf{k} \\
\mathbf{0}_W  & \mathbf{\tilde B}^{-1}
\end{bmatrix}
& = &
\begin{bmatrix}
1             & \mathbf{0}_W^\T \\
\mathbf{0}_W  & \mathbf{I}_W
\end{bmatrix}.
\end{array}
$$
Furthermore, from Lemma~\ref{lem:binom_mtxinv} we have $\lbrack\mathcal{B}
(\qp,\pp) \rbrack ^{-1} =  \mathcal{B}(-\qp\pp^{-1}, \pp^{-1})$. Thus
$\mathbf{\tilde B}^{-1}$ is essentially a principal $W \times W$ submatrix of
$\mathcal{B}(-\qp\pp^{-1}, \pp^{-1})$.
\end{proof}

\noindent For PS  $\uncon^{-1}$ is difficult to write in a compact form and will be
omitted. It is however feasible to give using equation \eqref{eq:gen_inverse_diag}
\begin{align}
\nonumber
 ({\mathbf{J}}^{-1})_{jj} &= \sum_{k=j}^W \dbinom{k}{j}^2 \qp^{2(k-j)} \pp^{-2k}
	d_k \\
	&\hspace{5mm}- \frac{\Big( \sum_{k=j}^W  (-1)^{2k-j-1} d_k \dbinom{k}{j} \qp^{2k-j} 
	\pp^{-2k} \Big)^2}{\sum_{k=0}^W \qp^{2k} \pp^{-2k} d_k}	.
\label{eq:jinv_ps}
\end{align}
The above form is derived in Appendix \ref{app:sampling}.

\subsection{Packet Sampling with Sequence Numbers (PS+SEQ)}

First PS with parameter $\pp$ is performed as above. Sequence numbers are then
used as follows. Let $s_l$ be the lowest sequence number among the sampled packets, and
$s_h$ the highest.  All packets in-between these can now reliably inferred, hence
$j=s_h-s_l+1$ is the number of sampled packets returned. This is called
``ALL-seq-sflag'' in \cite{Ribeiro06fischer}.

The chief benefit of PS+SEQ is the fact that a potentially large number of packets can
be `virtually' observed without having to physically sample them. The disadvantage is
the additional processing involved to perform the inference. Also, the technique is of limited value if flows are too short (as we discuss later).

If $j=0,1$ then sequence numbers cannot help and $\bjk$ is as for PS. Otherwise, note
that the $j$ packets must occur in a contiguous block bordered by $s_l$ and $s_h$.
There are $k-j+1$ possible positions for such a block, each characterized by $k-j$
unsampled packets outside it and the borders $s_l$ and $s_h$. It follows that $b_{jk} =
(k-j+1)\pp^2\qp^{k-j}$ for $2<j\le k$. Hence
$$
\mathbf{B} =
\begin{bmatrix}
\qp &   \qp^2    &  \qp^3   &   \qp^4  &    \cdots  & \qp^W  \\
\pp   & 2\pp\qp  &  3\pp\qp^2  &   4\pp\qp^3  & \cdots  & W \pp\qp^{W-1} \\
0   &  \pp^2     &  2\pp^2\qp   &   3\pp^2\qp^2  & \cdots  & (W-1)\pp^2\qp^{W-2} \\
0   & 0    &   \pp^2   &  2\pp^2 \qp   &    \cdots  & (W-2)\pp^2\qp^{W-3} \\
\vdots   &  \vdots & \vdots & \vdots & \ddots & \vdots \\
0   &  0    &  0 & \cdots & 0 & \pp^2
\end{bmatrix}.
$$
\begin{thm}
The inverse of $\mathbf{\tilde B}$ is
$$
\mathbf{\tilde B}^{-1} =
\begin{bmatrix}
\pp^{-1}  & -2\qp\pp^{-2} & \qp^2\pp^{-2} & 0  & \cdots  & 0 \\
0   &  \pp^{-2} &  -2\qp\pp^{-2}  & \qp^2\pp^{-2}  & \cdots  & 0 \\
0   & 0    &   \pp^{-2}   &  -2\qp\pp^{-2}   &    \cdots  & 0 \\
\vdots   &  \vdots & \vdots & \vdots & \ddots & \vdots \\
0   &  0    &  0 & \cdots & 0 & \pp^{-2}
\end{bmatrix}.
$$
\label{thm:psseq_Binv}
\end{thm}
\begin{proof}
Observe that $\mathbf{\tilde B} = \mathbf{S}\mathbf{T}$ where
$$\mathbf{S} = \diag(\pp, \pp^2, \pp^2, \ldots, \pp^2)$$
is a $W \times W$ matrix and
$$
\mathbf{T} =
\begin{bmatrix}
1   & 2\qp  &  3\qp^2  &   4\qp^3  & \cdots  & W \qp^{W-1} \\
0   &  1     &  2\qp   &   3\qp^2  & \cdots  & (W-1)\qp^{W-2} \\
0   & 0    &   1   &  2\qp   &    \cdots  & (W-2)\qp^{W-3} \\
\vdots   &  \vdots & \vdots & \vdots & \ddots & \vdots \\
0   &  0    &  0 & \cdots & 0 & 1
\end{bmatrix}.
$$

A straightforward computation yields
$$
\mathbf{T}^{-1} =
\begin{bmatrix}
1  & -2\qp & \qp^2 & 0  & \cdots  & 0 \\
0   &  1 &  -2\qp  & \qp^2  & \cdots  & 0 \\
0   & 0    &   1   &  -2\qp   &    \cdots  & 0 \\
\vdots   &  \vdots & \vdots & \vdots & \ddots & \vdots \\
0   &  0    &  0 & \cdots & 0 & 1
\end{bmatrix},
$$
and $\mathbf{S}^{-1} = \diag(\pp^{-1}, \pp^{-2}, \pp^{-2}, \ldots, \pp^{-2})$. Thus,
$\mathbf{\tilde B}^{-1} = \mathbf{T}^{-1}\mathbf{S}^{-1}$, which proves our result.
\end{proof}

The diagonal elements of $\uncon^{-1}$ are given by
\begin{align*}
({\mathbf{J}}^{-1})_{11} &= \pp^{-2} d_1 + 4\qp^2\pp^{-4}d_{2} + \qp^4
\pp^{-4}d_{3} - \tiny{\text{$\frac{(\qp^2\pp^{-4}d_1 + 2\qp^3\pp^{-4}d_2)^2}{r}$}} \\
({\mathbf{J}}^{-1})_{22} &= \pp^{-4} d_2 + 4\qp^2\pp^{-4}d_{3} + \qp^4
\pp^{-4}d_{4} - \tiny{\text{$\frac{\qp^4\pp^{-8}d^2_2}{r}$}}\\
({\mathbf{J}}^{-1})_{jj} &=  \pp^{-4} d_j +
4\qp^2\pp^{-4}d_{j+1} + \qp^4 \pp^{-4}d_{j+2}, \ 3\le j\le W 
\end{align*}
where $r = d_0 + \qp^2\pp^{-2}d_1 + \qp^4\pp^{-4}d_2$, and for convenience we set $d_j
= 0$ for $j > W$.

\subsection{Packet Sampling with SYN Sampling (PS+SYN)}

First PS with parameter $\pp$ is performed as above. A post-processing phase then
discards all packets belonging to sampled flows which lack a SYN packet (or more
accurately, maps them to sampled flows with $j=0$). This was introduced in
\cite{duffsamplingToN2005} and called ``SYN-pktct'' in \cite{Ribeiro06fischer}.

The chief benefit of PS+SYN is that the flow length bias of PS is averted by keeping
flows based on the presence of the SYN, which is flow length independent. The chief
disadvantage is the fact that it is wasteful: if $\pp=0.01$ then 99\% of packets which were 
initially sampled belong to `failed' flows and are subsequently discarded!

A flow evaporates iff its SYN is not sampled, hence $b_{0k}=\qp$. For $j\ge1$ the SYN
must first be sampled, which occurs with probability $\pp$, and conditional on this
$j-1$ more packets must be sampled from the remaining $k-1$ using i.i.d. sampling. Hence
$b_{jk} = \pp\cdot\dbinom{k-1}{j-1}\pp^{j-1}\qp^{k-j}$ for $j\ge1$, giving
$$
\mathbf{B} =
\begin{bmatrix}
\qp       & \qp        & \qp      &  \cdots  &  \qp\\
\pp       & \pp q_p    & \pp\qp^2 &  \cdots  &  \pp q_p^{W-1}  \\
 0        & \pp^2      & 2\pp^2\qp&  \cdots  & (W-1)\pp^2 q_p^{W-2}\\
\vdots    & \vdots     & \vdots   &  \ddots  & \vdots              \\
 0        & 0          &  0       &  \cdots  & \pp^W
\end{bmatrix}.
$$

\smallskip
\begin{thm}
The inverse of $\mathbf{\tilde B}$ is given by
$$
\mathbf{\tilde B}^{-1} = \frac{1}{\pp}
\begin{bmatrix}
1      &  -\qp\pp^{-1} &\qp^2\pp^{-2} &\!\cdots\!& (-\qp)^{W-1}\pp^{-W+1} \\
0      &   \pp^{-1}      & -2\qp\pp^{-2} &\!\cdots\!&  (-\qp)^{W-2}\pp^{-W+1} \\
\vdots &   \vdots       &  \vdots           &\!\ddots\!& \vdots \\
0      &  0                  &  \cdots           &\!\cdots\!& \pp^{-W+1}
\end{bmatrix}
$$
that is $b'_{jk} = (-1)^{k-j}\dbinom{k-1}{j-1}\qp^{k-j}\pp^{-k}$.
\label{thm:pssyn_Binv}
\end{thm}
\begin{proof}
Note that we can express $\mathbf{\tilde B}$ in terms of a $W \times
W$ binomial matrix $\mathcal{B}(x,y)$ such that $\mathbf{\tilde B} =
\pp\mathcal{B}(\qp,\pp)$. Then by Lemma \ref{lem:binom_mtxinv}, the inverse is
given by $\mathbf{\tilde B} = (1/\pp)\mathcal{B}(-\qp\pp^{-1}, \pp^{-1})$.
\end{proof}

\smallskip
It is easy to see that the condition of Corollary~\ref{cor:inv_uncon_syn} is satisfied with $p=\pp$.
Hence $\uncon^{-1} = \tilde \uncon^{-1} - \frac{\qp}{\pp}\vth \vth^\T$,
and the diagonal entries for $1 \le j \le W$
are
\be
   ({\mathbf{J}^{-1}})_{jj} = \sum_{k=j}^W	\dbinom{k-1}{j-1}^2 \qp^{2
	(k-j)} \pp^{-2k} d_k - \frac{\qp}{\pp} \theta^2_j.
\label{eq:jinv_pssyn}
\ee

\subsection{Packet Sampling with SYN and SEQ (PS+SYN+SEQ)}

First sampling is performed according to PS+SYN with parameter $\pp$,
and then on each resulting sampled flow  the sequence number post-processing is
performed as per PS+SEQ. This is called ``SYN-seq'' in \cite{Ribeiro06fischer}.
PS+SYN+SEQ is a hybrid of PS+SYN and PS+SEQ and combines the advantages
and disadvantages of both.

If $j=0,1$ then sequence numbers cannot help and $\bjk$ is as for PS+SYN. Otherwise, by
combining the arguments above, it is easy to see that $b_{jk} = \pp\cdot\pp\qp^{k-j}$
for $j>1$, giving
\begin{equation}
\mathbf{B} =
\begin{bmatrix}
\qp &   \qp   &  \qp   &   \qp  &    \cdots  & \qp  \\
\pp   &   \pp\qp  &  \pp\qp^2  &   \pp\qp^3  &    \cdots  & \pp\qp^{W-1} \\
0     &  \pp^2    &  \pp^2\qp   &   \pp^2\qp^2  & \cdots  & \pp^2\qp^{W-2} \\
0      & 0    &  \pp^2   &  \pp^2\qp  &    \cdots  & \pp^2\qp^{W-3} \\
\vdots     &   \vdots   &  \vdots & \vdots & \ddots & \vdots \\
0      &  0    &  0 & \cdots & 0 & \pp^2
\end{bmatrix}.
\label{eq:pssynseq_b}
\end{equation}

\begin{thm}
The inverse of $\mathbf{\tilde B}$ is given by
$$
\mathbf{\tilde B}^{-1} = \frac{1}{\pp}
\begin{bmatrix}
1        &   -\frac{\qp}{\pp}   &  0   & 0  & \cdots  & 0 \\
0 	 &  \frac{1}{\pp}     &   -\frac{\qp}{\pp}   &  0  &     \cdots  & 0\\
0        &  0	&     \frac{1}{\pp}     &   -\frac{\qp}{\pp}   &   \cdots  & 0\\
\vdots   &   \vdots   &  \vdots & \vdots & \ddots & \vdots \\
0        &  0  &  0 & \cdots &   0 		& \frac{1}{\pp}
\end{bmatrix}
$$
\label{thm:pssynseq_Binv}
\end{thm}
\begin{proof}
A straightforward computation shows that $ \Bt\Bt^{-1} = \mathbf{I}_W$.
\end{proof}

Since $\mathbf{b}_0 =  \qp\onew$,
Corollary~\ref{cor:inv_uncon_syn} applies and states that
$\uncon^{-1} = \tilde \uncon^{-1} - \frac{\qp}{\pp}\vth \vth^\T$.
The diagonal elements can be explicitly written, but we defer this to
Section~\ref{ssec:DS}.

\subsection{Flow Sampling (FS)}

In \textit{i.i.d.~flow sampling} \cite{thinningToN}, flows are
retained independently with probability $\pf$ and otherwise dropped
with $\qf = 1-\pf$.

The chief benefit of FS is the fact that flows which are sampled
retain their full complement of packets, eliminating completely the
difficulties in inverting sampled flow sizes back to original sizes.
The chief disadvantage is that each packet requires a lookup
in a flow table to see if it belongs to be flow which has been
sampled.

A flow evaporates iff its SYN is not sampled, hence $b_{0k}=\qf$.
If a flow has been selected,  which occurs with
probability $\pf$, then conditional on this $j=k$ with certainty, that
is $b_{jk}=1$ if $j=k$, else 0, for $j\ge1$:

\begin{equation}
\mathbf{B} =
\begin{bmatrix}
\qf      &  \qf     &  \qf     & \qf     & \cdots     & \qf \\
\pf 	 &  0       &  0       &  0      &     \cdots & 0\\
0        &  \pf	    &     0    &   0     &   \cdots   & 0\\
\vdots   &   \vdots &  \vdots  & \vdots  & \ddots     & \vdots \\
0        &  0       &  0       &  0      &   \cdots   &  \pf
\end{bmatrix}.
\label{eq:flow_b}
\end{equation}

The inverse of $\Bt$ is just $\Bt^{-1}=\mathbf{I}_W/\pf$, and $\uncon$
takes the elegant form $\uncon(\vth) = \qf \mathbf{1}_W\mathbf{1}_W^\T + \pf
\mathbf{diag}(\theta^{-1}_1, \ldots,\theta^{-1}_W)$.

Clearly Corollary~\ref{cor:inv_uncon_syn} applies with $p=\pf$.
Since
\begin{align*}
\tilde \uncon &= \pf \diag(\theta^{-1}_1, \ldots,\theta^{-1}_W),
\end{align*}
the unconstrained inverse can therefore be expressed as
\be
  \uncon^{-1}(\vth) = \frac{1}{\pf}\diag(\vth) + (1 - \frac{1}{\pf}) \vth\vth\,^\T .
  \label{eq:fsinv_alt}
\ee

\begin{rem}
By using equations (\ref{eq:inverse_simple}) and (\ref{eq:fsinv_alt}), the inverse
constrained Fisher information matrix for FS is given by
\begin{equation*}
   \con^{+} = \frac{1}{\pf}\diag(\vth) - \frac{1}{\pf}\vth \vth~^\T,
\end{equation*}
\noindent with the diagonals being $(\con^{+})_{kk} = (1/\pf)\theta_k (1-\theta_k)$.
This is just an appropriately scaled version (by $1/\pf$) of the inverse
Fisher information of a multinomial model. This makes sense given that FS works by picking out
whole flows from the original flow set in an i.i.d fashion and that the complete
likelihood function can be modeled using a multinomial model.
\end{rem}

\subsection{Packet Sampling with SYN, FIN,\! SEQ (PS+SYN+FIN+SEQ)}

In this scheme SYN packets are retained independently with
probability $\pf$ and otherwise dropped with $\qf = 1-\pf$, and the
FIN packets corresponding to sampled SYNs are also sampled, but no others.
Sequence numbers are then used to infer flow sizes.

This scheme has two great advantages:  like FS the flows sampled are
sampled perfectly, and moreover this could be achieved by physically sampling only
two packets per flow, based on looking for SYN and FIN flags on a per
packet basis, which is feasible at high speed.
The disadvantage is that a moderate minority of flows do not terminate
correctly with a FIN, and/or the FIN may be not observable.
Furthermore, flows consisting of a single SYN (such as in a SYN attack) would be entirely missed.
For this reason we choose not to study it further.

Information theoretically, PS+SYN+FIN+SEQ is identical to flow sampling provided we assume $\theta_1\approx0$.

\subsection{Sample and Hold (SH)}

Here packets are first sampled as for PS with probability $\pp$, however for each flow if a packet is sampled, then all subsequent packets in the flow will be. Hence the total number of packets sampled is much higher than the parameter $\pp$.  The scheme was introduced in \cite{Estan03}.

The chief benefit of SH is that provided just a single packet from a flow is PS-sampled, then typically many will be finally sampled. This conditional behaviour is much more effective than methods using SEQ where at least two packets must be PS-sampled, and even then fewer packets are finally recouped.
Essentially SH skips a geometric number of the first packets in a flow and then captures all the rest.
It therefore efficiently skips small flows and accurately recovers the size of large flows.
This amplified flow length bias (even stronger than for PS) makes it well suited for the heavy hitter problem (i.e.~accurately measuring the very largest flows) for which it was originally designed.
For flow size estimation more generally however, it is a disadvantage for most flow sizes.
The other disadvantage is the need to check, for \textbf{each packet}, whether it belongs to a flow which has already been sampled. This makes it very costly in a true sampling implementation.
Indeed Estan and Varghese implemented it using lossy sketching techniques \cite{Estan03}.

A flow evaporates iff none of its packets are sampled, hence $b_{0k}=\qp^k$. Otherwise,
$b_{jk} = \pp\qp^{k-j}$, thus:
\begin{equation}
\mathbf{B} =
\begin{bmatrix}
\qp      &  \qp^2     &  \qp^3     & \qp^4     & \cdots     & \qp^W \\
\pp 	 &  \pp\qp       &  \pp\qp^2       &  \pp\qp^3      &     \cdots & \pp\qp^{W-1}\\
0        &  \pp	    &     \pp\qp    &   \pp\qp^2     &   \cdots   & \pp\qp^{W-2}\\
\vdots   &   \vdots &  \vdots  & \vdots  & \ddots     & \vdots \\
0        &  0       &  0       &  0      &   \cdots   &  \pp
\end{bmatrix}.
\label{eq:sh_b}
\end{equation}
It is interesting to note that the first row is as for PS, the second as for PS+SYN,  and subsequent rows like PS+SYN scaled by $1/\pp$.

\smallskip
\begin{thm}
The inverse of $\mathbf{\tilde B}$ is given by
$$
\Bt^{-1} = \frac{1}{\pp}
\begin{bmatrix}
1        &   -\qp   &  0   & 0  & \cdots  & 0 \\
0 	 &  1     &   -\qp   &  0  &     \cdots  & 0\\
0        &  0	&     1     &   -\qp   &   \cdots  & 0\\
\vdots   &   \vdots   &  \vdots & \vdots & \ddots & \vdots \\
0        &  0  &  0 & \cdots &   1 		& -\qp\\
0        &  0  &  0 & \cdots &   0 		& 1
\end{bmatrix}
$$
\label{thm:sh_Binv}
\end{thm}
\begin{proof}
\ It is easily verified that $ \Bt\Bt^{-1} = \mathbf{I}_W$.
\end{proof}

It is not difficult to show that Proposition~\ref{prop:sub} reduces to ${\mathbf{J}}^{-1}
= \Jt^{-1} -  \mathbf{C}$, where the only non-zero element of $\mathbf{C}$ is
$\mathbf{C}_{11}= \pp^{-2}\qp^2d_1^2(\pp^2d_0+\qp^2 d_1)^{-1}=d_0/\pp$ since $d_1 = \frac{\pp}{\qp}d_0$.
Furthermore, using (\ref{eqn:Jtildeinv}) one can
show that $\Jt^{-1}$ is tridiagonal (and symmetric) with upper off-diagonal terms
$(\Jt^{-1})_{k,k+1}= -\pp^{-2} \qp d_{k+1} $, $k<W$, and
diagonal elements
\ba
      \nonumber
	({\mathbf{J}}^{-1})_{11}\!\!\!\!&=&\!\! \frac{1}{\pp^2}(d_1+ \qp^2d_{2}) - \frac{1}{\pp} d_0\\ 
	 \label{eq:SHone}
	&=& \theta_1 + \frac{1}{\pp} \sum_{k=2}^W \qp^{k-1} \theta_k\\ \nonumber
	({\mathbf{J}}^{-1})_{jj}\!\! &=&\!\! \frac{1}{\pp^2}(d_j + \qp^2d_{j+1}), \ 2 \le j\le W-1 \\
	 \label{eq:SHj}
                             &=& \frac{\theta_j}{\pp}  + \frac{1+\qp}{\pp} \sum_{k=j+1}^W \qp^{k-j} \theta_k \\\nonumber
	({\mathbf{J}}^{-1})_{WW}\!\! &=&\!\! \frac{d_W}{\pp^2} = \frac{\theta_W}{\pp},
\ea

\noindent  where we have used the property of $\mathbf{B}$ that $d_j = \pp \theta_j + \qp d_{j+1}$ for $1 \le j \le W-1$.

\subsection{Dual Sampling (DS)}
\label{ssec:DS}

DS can be defined simply as follows. First, at the packet level it consists of two PS
schemes running in parallel, one which operates only on SYN packets with sampling
probability $\pf$, and the other only on non-SYN packets with sampling probability
$\pp$. In a second phase, sampled flows which lack a SYN are discarded, and
sequence numbers are used to infer additional `virtual' packets, as in PS+SYN+SEQ.
Thus, at one level DS is simply a generalization of PS+SYN+SEQ, and reduces to it when
$\pf=\pp$. However, the generalization is significant as it also includes FS as the
special case $\pp=1$, and interpolates continuously between the two. This is illustrated
in Figure~\ref{fig:DSfamily} which depicts the $(\pp,\pf)$ parameter space, and marks
the special cases.

Dual Sampling is `dual' in two senses. Computationally it can be viewed as the original
PS sampling being split into two, at the (low) cost of per-packet switching based on
some bit checking to determine which PS applies. Information theoretically, the sampling
is now split into two parts with very different natures, each controlled by a dedicated
parameter: a FS-like direct sampling of flows, and a PS-like in-flow sampling. Here
$\pf$ controls the number of sampled flows, and $\pp$ their `quality'.

The derivation of the sampling matrix mirrors closely that of PS+SYN+SEQ.  The result
shows clearly how $\pf$ and $\pp$ act in a modular fashion. The flow sampling component
controls the top row of $\mathbf{B}$ and factors $\Bt$, whereas the packet sampling
component determines the internal structure of $\Bt$.
\begin{equation}
\nonumber
\mathbf{B} =
\begin{bmatrix}
\qf &   \qf   &  \qf   &   \qf  &    \cdots  & \qf  \\
\pf   &   \pf\qp  &  \pf\qp^2  &   \pf\qp^3  &    \cdots  & \pf\qp^{W-1} \\
0     &  \pf\pp    &  \pf\pp\qp   &   \pf\pp\qp^2  & \cdots  & \pf\pp\qp^{W-2} \\
0      & 0    &  \pf\pp   &  \pf\pp\qp  &    \cdots  & \pf\pp\qp^{W-3} \\
\vdots     &   \vdots   &  \vdots & \vdots & \ddots & \vdots \\
0      &  0    &  0 & \cdots & 0 & \pf\pp
\end{bmatrix}.
\label{eq:dps_b}
\end{equation}
The separation of the FS and PS roles in $\Bt$ is clearly reflected in its inverse:
$$
\mathbf{\tilde B}^{-1} = \frac{1}{\pf}
\begin{bmatrix}
1        &   -\frac{\qp}{\pp}   &  0   & 0  & \cdots  & 0 \\
0 	 &  \frac{1}{\pp}     &   -\frac{\qp}{\pp}   &  0  &     \cdots  & 0\\
0        &  0	&     \frac{1}{\pp}     &   -\frac{\qp}{\pp}   &   \cdots  & 0\\
\vdots   &   \vdots   &  \vdots & \vdots & \ddots & \vdots \\
0        &  0  &  0 & \cdots &   0 		& \frac{1}{\pp}
\end{bmatrix}.
$$

The similarity between $\mathbf{\tilde B}^{-1}$ for SH and DS is striking.
Indeed, whereas SH uses sampling to select which flows it will focus on and then holds to them, DS reverses these operations and thus could be described as a `Hold and Sample' scheme.
However, although they each combine PS and FS features, the methods remain significantly different. In particular SH is strongly flow length biased whereas DS is unbiased.

Once again Corollary~\ref{cor:inv_uncon_syn} for the inverse Fisher information holds, showing
that $\uncon^{-1} = \tilde \uncon^{-1} - \frac{\qf}{\pf}\vth \vth^\T$.
Similarly to SH, from (\ref{eqn:Jtildeinv}) one can show that $\Jt^{-1}$ is tridiagonal with upper off diagonal terms $(\Jt^{-1})_{k,k+1}= -(\pf\pp)^{-2} \qp d_k $, $k<W$.
The diagonal elements (here $2 \le j \le W-1$) are given by
\ba
	\nonumber
	\hspace{-4mm} ({\mathbf{J}}^{-1})_{11}\!\!
	&=&\!\! \frac{1}{\pf^2\pp^2}(\pp^2 d_1+\qp^2d_2) - \frac{\qf}{\pf} \theta_1^2\\
  \label{eq:DS_diags1}
  	&=&\!\! \frac{\theta_1}{\pf} + \frac{1}{\pf\pp}\sum_{k=2}^W \qp^{k-1}\theta_k  - \frac{\qf}{\pf}\theta_1^2\\
	\nonumber
	\hspace{-4mm}  ({\mathbf{J}}^{-1})_{jj}\!\!
	&=&\!\! \frac{1}{\pf^2\pp^2}(d_{j}+\qp^2d_{j+1}) - \frac{\qf}{\pf} \theta_j^2\\\nonumber
       &=&	\frac{\theta_j}{\pf\pp} + \frac{1}{\pf\pp}\sum_{k=j+1}^W \qp^{k-j}(1+\qp)\theta_k - \frac{\qf}{\pf} \theta_j^2\\
	\hspace{-4mm}  ({\mathbf{J}}^{-1})_{W\!W}\!\!
	   	&=&\!\! \frac{d_W}{\pf^2\pp^2} - \frac{\qf}{\pf} \theta_W^2 = \frac{\theta_W}{\pf\pp} - \frac{\qf}{\pf} \theta_W^2.
  \label{eq:DS_diags}
\ea
By setting $\pf=\pp$ we obtain those for PS+SYN+SEQ.

\section{Comparisons}
\label{sec:compare}

In this section we compare and contrast the performance of the different methods, using two
normalizations which are the key to a fair comparison. We
show that a positive semidefinite comparison holds for certain cases, and justify
why we ultimately resort to comparison of the diagonals of the covariance matrix.
We provide a partial ranking of the methods. Proofs of most key results in this
section are deferred to Appendix \ref{app:comps}.

\subsection{Normalization}
\label{ssec:norm}
\begin{figure}[t]
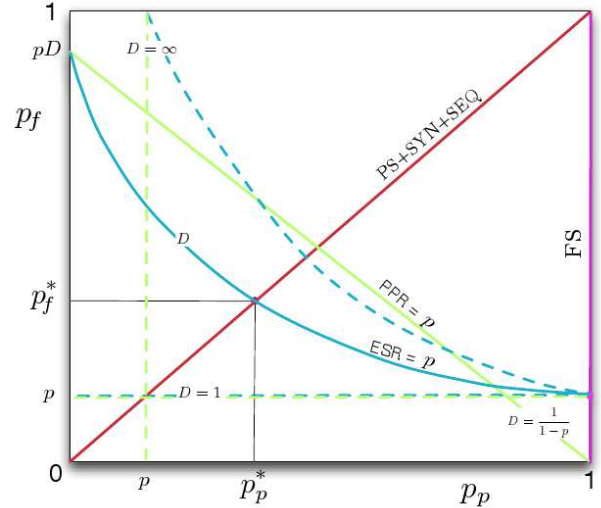

  \DSfamily
  \caption{Parameter space $(\pp,\pf)$ of DS.
   Fixing ESR$=p$ constrains the family to the solid (blue) curve $\pf(\pp;p)$, where
    it reduces to PS+SYN+SEQ at $\pp=\pp^*$ and FS at $\pp=1$.
   Fixing PPR$=p$ constrains the family to a straight (green) line emanating from $(0,pD)$.
   For fixed $p$, the constraint curves depend on $D$. Three examples are given for each normalization.}
  \label{fig:DSfamily}
\end{figure}

We must first consider how to compare fairly. We do this in two ways, using the
following packet-based measures of incoming workload/information.
\begin{itemize}
  \item  \textit{Packet Processing Rate} (\textbf{PPR}): \ the rate at which
   packets are initially being sampled (and hence require some further processing).
  \item  \textit{Effective Sampling Rate} (\textbf{ESR}): \ the rate at which
   packets are arriving to the flow table (and hence become available
   as information for estimation).
\end{itemize}
PPR is a measure of the processing speed required by the methods, whereas
ESR is a measure of the arrival rate of
packets containing information which is actually used by the method.
Because the methods which discard flows without SYN packets lose (the
majority) of their packets, an equal PPR comparison greatly disadvantages them.
An equal ESR comparison effectively inflates the parameters of such methods to compensate.
\begin{table}[b]
\begin{tabular}{|l||c|c|}
\hline
Method(params)      &  Average sampling rate $p=$  & PPR$(p)=$ \\
\hline\hline
PS$(\pp)$                & $\pp$                      & $p$ \\\hline
PS+SYN$(\pp)$       & $\pp$                      & $p$ \\\hline
FS$(\pf)$                 & $\pf$                       & $p$ \\\hline
SH$(\pp)$                & \!\!$\frac{\pp}{D}\sum_{j=1}^W j\qp^{-j}\sum_{k=j}^W \qp^k 
\theta_k\!\!$ & root$(p(\pp))$ \\\hline
DS$(\pp,\pf)$           & $\pf\big(\frac{1}{D}\big) + \pp\big( 1- \frac{1}{D} \big)$ & 
\!\!$\pf=pD\!-\!\pp(D\!-\!1)$\!\!\!\\
\hline
\end{tabular}
\caption{Average PPR sampling rates and normalizations.}
\label{table:PPR}
\end{table}

Note that given a sampling method $X$, both the PPR and ESR normalizations of X+SEQ and X 
are identical, as the methods only differ in how the packet data is used, not how many 
packets are physically collected.
We now present the normalizations for each method. Here and below we use
$D=\sum_{k=1}^W k\theta_k \ge1$ to denote the average flow size.

Table~\ref{table:PPR} gives the results for PPR. Both the average sampling rate $p$ as a 
function of the method parameter(s), and its inverse, the parameter value required to set 
the average sampling rate at $p$, are shown.  The expression $p(\pp,\vth)$ for SH was derived 
by computing the average number of packets sampled (which is highly dependent on
$\vth$) and dividing the result by the average flow size $D$. It is 
monotonically increasing in $\pp$ with $p(1,\vth)=1$, and $p/\pp\rightarrow\sum_{j=1}^W
j\sum_{k=j}^W \theta_k/D \ge1$ as $\pp\rightarrow0$, which in the worst case
($\theta_W=1$) becomes  $p/\pp=(1+W)/2$. We invert $p(\pp,\vth)$ numerically as required
for choices of $\vth$, denoted by root$(p(\pp))$ in the tables. DS has two parameters. We 
choose to make $\pp$ the independent one.

Table~\ref{table:ESR} gives the results for ESR. Only the methods which involve discarding packets after their initial collection
have changed compared to Table~\ref{table:PPR}.
Note that for PS+SYN $p/\pp\rightarrow 1/D \ge1$ as $\pp\rightarrow0$, and $p(1)=1$.
\begin{table}[t]
\begin{tabular}{|l||c|c|}
\hline
Method(params)      &  Average sampling rate $p=$  & ESR$(p)=$ \\
\hline\hline
PS$(\pp)$                & $\pp$                      & $p$ \\ \hline
PS+SYN$(\pp)$       & $\frac{\pp(\pp(D-1) +1)}{D}$  &\!\!\! $\frac{-1+\sqrt{1+4pD(D-1)}}{2(D-1)}$\!\!\!  \\\hline
FS$(\pf)$                 & $\pf$                       & $p$ \\\hline
SH$(\pp)$                &  \!\!$\frac{\pp}{D}\sum_{j=1}^W j\qp^{-j}\sum_{k=j}^W \qp^k \theta_k$  \!\!& root$(p(\pp))$ \\\hline
DS$(\pp,\pf)$           & $\frac{\pf(\pp(D-1) +1)}{D}$ & $\pf=\frac{pD}{\pp(D-1)+1}$\\
\hline
\end{tabular}
\caption{Average ESR sampling rates and normalizations.}
\label{table:ESR}
\end{table}

We are particularly interested in the ESR case for the DS family, and so repeat its ESR normalization here:
\be
  \pf(\pp;p) = \frac{pD}{\pp(D-1)+1}.
  \label{eq:ESRcurve}
\ee
The denominator $\pp(D-1)+1$ is simply the average sampled flow size, conditional on its SYN being sampled.
For fixed $p$, this equation gives $\pf$ as a monotonically decreasing, in fact convex, function of $\pp$.
Three examples (the blue curves) are given in Figure~\ref{fig:DSfamily} for different values of $D$.
To maintain an ESR fixed at $p$, if we increase $\pp$ we must
move down the curve and decrease $\pf$ to compensate.
The PPR case is similar but simpler as the level curves are straight lines (green lines in Figure~\ref{fig:DSfamily}).  Note that depending on $p$ and $D$, under both PPR and ESR
there are values of $\pf$ and/or $\pp$ that may be disallowed.

\smallskip
To compare methods we also require a performance  metric.
A natural criterion is the set of diagonal elements of the CRLB,
the $k$-th being
\be
  (\con^+)_{kk} = (\uncon^{-1})_{kk} - \theta_k^2,
\ee
since this is a lower bound on the variance of any unbiased estimator of $\theta_k$.
We plot the square root of these values, calculated using the
expressions for the previous section.
\begin{figure}[b]
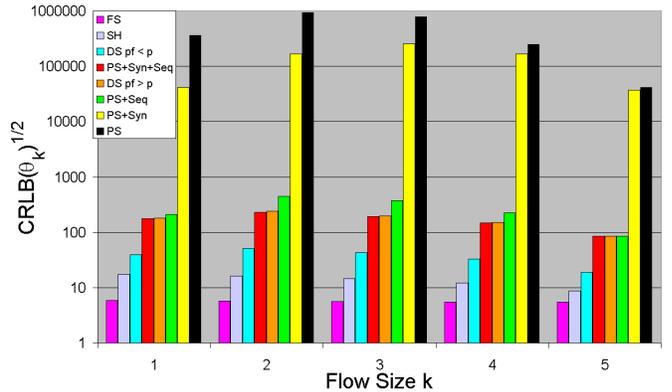

\vspace{-4mm}
  \hspace{-3mm} \shortComparePPR
  \caption{An equal PPR comparison of the CRLB bound for $\vth$ with $p=0.005$.}
  \label{fig:shortComparePPR}
\end{figure}
\begin{figure}[b!]
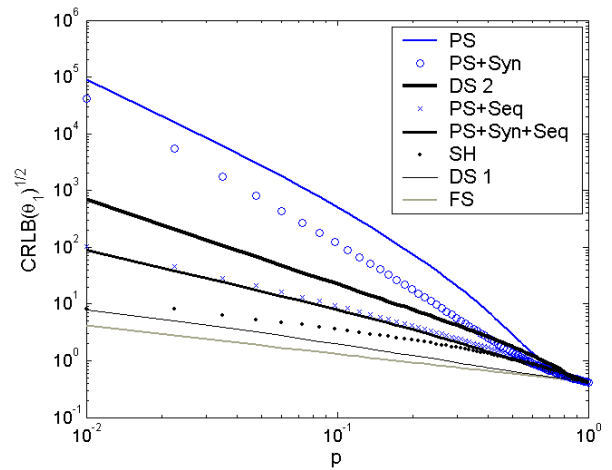

  \sampvaryPPR
  \caption{Dependence of the CLRB bound on $p$, PPR comparison.}
  \vspace{-2mm}
  \label{fig:sampvary}
\end{figure}
\begin{figure}[b]
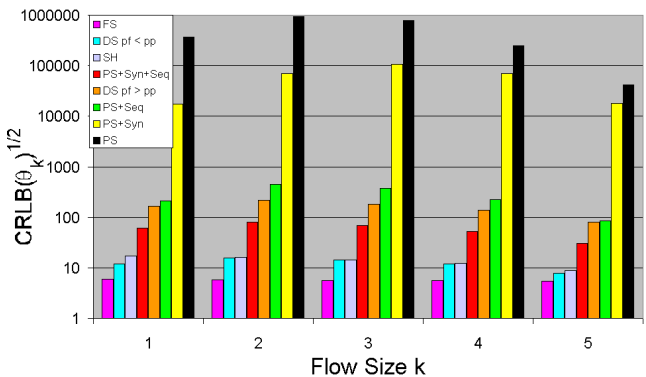

  \shortCompareESR
  \vspace{-2mm}
  \caption{An equal ESR comparison of the CRLB bound for $\vth$ with $p=0.005$.}
  \label{fig:shortCompareESR}
\end{figure}

We begin with some instructive examples.
Consider the results of Figure~\ref{fig:shortComparePPR}, where we set
$W=5$ with
$$\vth = \{0.22, 0.21, 0.20, 0.19, 0.08\}, $$
for which $D= 2.90$ packets.
We use the PPR normalization with $p=0.005$, corresponding to $\pp\approx0.0021845$ for SH.
For DS we give two examples satisfying $p$: $(\pf,\pp)=(0.001,0.0443)$, and $(\pf,\pp)=(0.1, 0.00334)$.
\begin{figure*}[t]
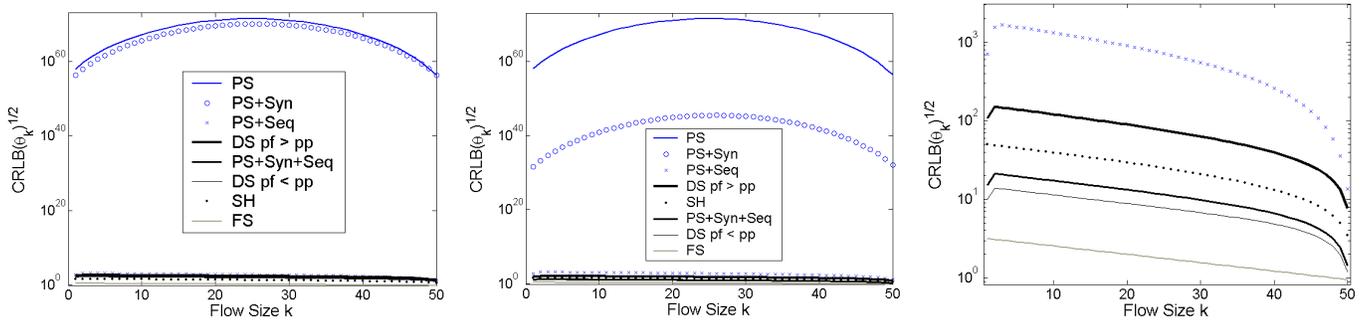

 \longComparePPR \hspace{-3mm} \longCompareESR \hspace{-3mm} \longCompareESRzoom
  \caption{Comparison of the CRLB bound with $W=50$.
Left:    PPR with $p=0.005$. $\pf = 0.001,\ \pp = 0.0052$ for DS ($\pf < \pp$) and  $\pf
		 = 0.01,\ \pp = 0.0047$ for DS ($\pf > \pp$).
Middle:  ESR with $p=0.005$. $\pf = 0.032,\ \pp = 0.1$ for DS ($\pf < \pp$) and  $\pf =
		 0.079,\ \pp = 0.001$ for DS ($\pf > \pp$).
Right:   Zoom of middle plot, the same legend applies.}
  \label{fig:longCompare}
\end{figure*}

As expected from earlier work, the performance of PS is
extraordinarily poor.  In agreement with the results of \cite{Ribeiro06fischer}
and as expected, the inclusion of SEQ improves it enormously, by
orders of magnitude, but it is still orders of magnitude behind FS,
which has the lowest standard deviation bound of all.
In a highly counterintuitive result, PS+SYN performs (much!) better than PS, despite the enormous number of wasted packets. We can offer two reasons for this.
First, in the unconditional framework
information in discarded flows is not entirely wasted, some is
recouped by the (observable) increase in the $j=0$ outcome\footnote{It was pointed out in \cite{Ribeiro06fischer}, under a conditional framework,
that PS+SYN+SEQ outperforms PS+SEQ on actual traces.
Our own numerical evaluation of the CRLB of these methods
(under the conditional framework) also show this in some cases, leading us to
conjecture the same counterintuitive results may hold for the conditional framework.}.
Second, SYN sampled flows (much like FS) have no flow length bias.
Together these lead PS+SYN to perform better in most cases even under PPR.
Three results for DS are given. From best to worst these were with $\pf=0.001<p$,  
the special case PS+SYN+SEQ with $\pf=p$, and $\pf=0.1>p$.
Finally SH performs well in this example as its key disadvantage, a strong bias to large flows, 
does not play a large role for such a small value of $W$.

Using the same PPR based comparison, Figure~\ref{fig:sampvary}
shows the improvement in CRLB as $p$ increases, as we expect.
We see that very high values of $p$ are needed before the performance
of FS is approached.  Here the free parameter $\pf$ for DS was chosen to guarantee meaningful 
examples to either side of PS+SYN+SEQ.
Specifically, the normalization defines for each value of $p$ a feasible range 
$[\pf^l,\pf^u]$ for $\pf$ which includes $\pf=p$. 
For DS1 we set $\pf= p - c(p-\pf^l)<p$, and for DS2 $\pf= p + c(\pf^u-p)>p$, where $c\in[0,1]$. In the figure we use $c=0.1$.

The great merit of the ESR normalization is that it allows us to view the methods as
being different ways in which the same budget of sampled packets can be allocated among
flows. The essential tradeoff is that we can have more sampled flows of poor quality, or
fewer of higher quality. In a class with no flow length bias, flow sampling is at one extreme: for a given ESR$=p$ it gives the minimum number of sampled flows, each of perfect quality.
Among sub-classes of
methods with roughly the same number of flows, we can then distinguish between the finer
details of how the `holes' appear over the flow, which will have an impact on the degree
of improvement which the sequence numbers can bring.


Results using the same scenario as Figure~\ref{fig:shortComparePPR} but ESR normalized
are shown in Figure~\ref{fig:shortCompareESR}. 
As expected, all SYN-based methods improve significantly.
In particular the DS with the smaller $\pf=0.001$ now outperforms SH, and is second only to FS.

Another example with a larger $W=50$ and a truncated exponential distribution for $\vth$
with $D=16.039$ is given in Figure~\ref{fig:longCompare}.
The same general conclusions hold, with the exception of SH, whose performance is now worse 
than PS+SYN+SEQ rather than considerably better.
This is to be expected at larger $W$, as SH expends its packet budget on the rare largest flows. 
At realistic $W$ values in network data this effect is further exaggerated.

It is interesting to note in the right plot of Figure~\ref{fig:longCompare} that variance decreases as $k$ increases. This is consistent with the fact that the ambiguity inherent in an observation of $j$ sampled packets is lower for larger $j$ and disappears at $j=W$, since this observation can only arise in one way.  The dip at very small $k$ we attribute to the influence of the constraint.

\subsection{Positive Semidefinite Comparisons}
\label{ssec:thcompare}

The examples above compared methods using the CRLB of each $\theta_k$ separately. Let two
methods have Fisher information $\uncon_1$ and $\uncon_2$. A more complete, and in a sense
ideal comparison, would be to show that $\uncon_1\ge\uncon_2$, since that would
imply $\con^+_1 \le \con^+_2$, a lower CRLB. What this positive semidefinite comparison
really means is that for any linear combination $f(\vth)= \mathbf{a}^\T \vth = \sum_k
a_k\theta_k$ of the parameters, the (bound on the) variance of $f(\vth)$ under method 1
will be less than that under method 2. A geometric interpretation is that the ellipsoid
corresponding to unit variance of vectors under $\uncon_1$ lies entirely within that of $\uncon_2$. 
In this section we provide a number of comparisons of this type between methods, and explain why
it is not suitable as a universal basis of comparison.

We first confirm the intuition that methods lose information monotonically in their sampling parameter(s).
\begin{thm}
For any method $Z$ surveyed here
if $p_1 \ge p_2$ (for DS $\ppo\ge\ppt$ and $\pfo\ge\pft$)
then $\uncon_Z(p_1) \ge \uncon_Z(p_2)$  (for DS $\uncon_Z(\ppo,\pfo) \ge \uncon_Z(\ppt,\pft)$). 
Equality holds iff $p_1 = p_2$ (for DS $(\ppo,\pfo)=(\ppt,\pft)$).
\label{thm:dpi_sampling}
\end{thm}
\smallskip
The proof (see Appendix \ref{app:comps}), is a straightforward consequence of closure under 
sampling for all methods except SH.

\smallskip
Next, we confirm that the use of sequence numbers increases Fisher information.
\begin{thm}
 \ $\uncon_{Z+\textrm{SEQ}} \ge \uncon_{Z}$ for each of $Z\!\!=$\,PS and $Z=$PS+SYN.
\label{thm:dpi_seq}
\end{thm}
\begin{proof}
We make use of the data processing inequality (DPI) for Fisher
information (see Theorem \ref{thm:fisher_dpi} in Appendix \ref{app:other_lemmas}), which states that if $\vth\rightarrow
Y\rightarrow X$ is a Markov chain, where $X$ is a deterministic function of $Y$, then
$\uncon_Y(\vth)\ge \uncon_X(\vth)$ which holds with equality if $X$ is a sufficient
statistic of $Y$. We first consider PS+SEQ and PS. Flows are selected randomly and are
represented as a random vector of SEQ numbers $\mathbf{V} = \{A,A+1,\ldots,A+K-1\}$ with
realization $\mathbf{v} = \{a,a+1,\ldots,a+k-1\}$. Here, $K$ represents the flow size
(realization $k$) while the $A$, the initial SEQ number (realization $a$) is
uniform over $\lbrack 0, N_a -1 \rbrack$. All operations on them are modulo
$N_a$, thus allowing wrap-arounds. $A$ is {\em independent} of $K$. We also assume that
$W \le N_a$ to avoid problems with multiple wrap-arounds.

After the PS process, we have the SEQ vector $\mathbf{Y} = \{Y_i, 1 \le i \le J\}$,
(realization $\mathbf{y} = \{y_i, 1 \le i \le j\}$). We define two statistics on
$\mathbf{Y}$: $J(\mathbf{Y})$, the actual number of raw packets sampled, and
$L(\mathbf{Y}) = \max_i(Y_i) - \min_i(Y_i) + 1$ (with a sample $l(\mathbf{y})$ and
$L(\mathbf{Y}) = 0$ if $\mathbf{Y}$ is empty), the inferred length of the sequence. The
statistic $L(\mathbf{Y})$ disregards $A$ since it takes the difference of SEQ numbers.
Note that both statistics are deterministic functions of $\mathbf{Y}$ and so form Markov
chains $\vth \to \mathbf{Y} \to L(\mathbf{Y})$ and $\vth \to \mathbf{Y}
\to J(\mathbf{Y})$. A straightforward application of DPI yields $\uncon_\mathbf{Y} \ge
\uncon_{L(\mathbf{Y})}$ and $\uncon_\mathbf{Y} \ge \uncon_{J(\mathbf{Y})}$.

We show that $L(\mathbf{Y})$ is a sufficient statistic of $\mathbf{Y}$ w.r.t.~$\vth$. We
use the Fisher-Neyman factorization theorem \cite[Theorem 6.5, p.~35]{Lehmann98} which
states that if the probability mass function of $\mathbf{Y}$ takes the form
\be
   \Pr(\mathbf{Y}=\mathbf{y}) = \mathrm{h}(\mathbf{y})\mathrm{g}(l(\mathbf{y}),\vth),
   \label{eq:suff_stat}
\ee
where $h$ is independent of the parameters $\vth$, then $L(\mathbf{Y})$ is a sufficient
statistic. Let $m \ge 0$ denote the number of unsampled packets before the first sampled
packet. If $J = j > 0$,
\begin{align*}
	&\Pr(\mathbf{Y}=\mathbf{y})\\
	&= \sum_{k=l}^W \theta_k \sum_{m=0}^{k-l} \Pr(A = y_1 - m)\qp^m\pp (\pp^{j-2} 
	\qp^{l-j}) \pp\qp^{k-l-m}\\
	&= \Pr(A = a) \sum^W_{k=l} \theta_k \sum^{k-l}_{m=0} \qp^{k-j}\pp^j \\
	&= \frac{1}{N_a} \pp^j \qp^{-j} \left(\sum^W_{k=l} (k-l+1)\qp^k \theta_k \right),
\end{align*}
since $A$ is uniform. For the case $j=0$,
\ben
   \Pr(\mathbf{Y}=\emptyset) = \sum^W_{k=1} \qp^k\theta_k .
\een
Clearly, each case satisfies (\ref{eq:suff_stat}), hence $L(\mathbf{Y})$ is a sufficient
statistic of $\vth$.

By the DPI, since $L(\mathbf{Y})$ is a sufficient statistic, $\uncon_\mathbf{Y} =
\uncon_{L(\mathbf{Y})}$. From the previous relation $\uncon_\mathbf{Y} \ge
\uncon_{J(\mathbf{Y})}$, we now have $\uncon_{L(\mathbf{Y})} \ge
\uncon_{J(\mathbf{Y})}$. But $J(\mathbf{Y})$ is equivalent to the statistic used in PS,
proving the result.

As for PS+SYN sampling, we define an additional random variable $S$ taking value 1 if
the SYN packet was sampled, 0 otherwise. $\mathbf{V}$ is defined as before. $\mathbf{Y}
= \emptyset$ if and only if $S = 0$. The same statistics $L(\mathbf{Y})$ and
$J(\mathbf{Y})$ are defined. Then, for $J = j > 0$,
\begin{align*}
\Pr(\mathbf{Y}=\mathbf{y})
&= \sum_{k=l}^W \theta_k \Pr(A = y_1) \pp (\pp^{j-2}\qp^{l-j})\pp \qp^{k-l}\\
&= \Pr(A = a) \sum^W_{k=l} \theta_k \qp^{k-j}\pp^j \\
&= \frac{1}{N_a} \pp^j \qp^{-j} \left(\sum^W_{k=l} \qp^k \theta_k \right),
\end{align*}
and for $j=0$,
\ben
\Pr(\mathbf{Y}=\emptyset) = \sum^W_{k=1} \qp\theta_k = \qp.
\een
Once again, each case satisfies (\ref{eq:suff_stat}), hence $L(\mathbf{Y})$ is a
sufficient statistic. Thus, by DPI, the same relationship $\uncon_{L(\mathbf{Y})} \ge
\uncon_{J(\mathbf{Y})}$ holds, with $J(\mathbf{Y})$ now equivalent to the statistic used
in PS+SYN.
\end{proof}

Intuitively this result seems obvious: if the additional information afforded by the
sequence numbers is available, we should certainly be able to do better by using it.
However, it is tempting to conclude that by the same logic $\uncon_{\textrm{PS}} >
\uncon_{\textrm{PS+SYN}}$ under the PPR normalization,
 since deciding to discard flows without a SYN is also a
deterministic transformation. However, the data processing inequality does not apply
here because some of the information used by PS+SYN (namely the SYN variable $S$),
although available to PS, is not used by it. Indeed we saw in the figures above that, counterintuitively, PS+SYN can
actually outperform PS in terms of the individual variances under PPR, which is a counter-example
to the more general positive semidefinite comparison.  Under ESR this is even more true as we saw from the middle plot in Figure~\ref{fig:longCompare}.

We now give another, even more surprising counter-example which teaches an important
lesson.
\smallskip
\begin{thm}
   $\uncon_{FS} \not\ge \uncon_{PS}$
\end{thm}
\begin{proof}
Proof is by contradiction via counter-example. Let $W=2$ with $\theta_1 = \theta_2 =
1/2$, and let $\pp =\pf = p$. Evaluate $\mathbf{1}_2^\T (\uncon_{FS} -
\uncon_{PS})\mathbf{1}_2$, (the sum of each element of the matrix difference), which we
expect to be nonnegative by assumption. It can be shown that this reduces to
\begin{equation*}
   2 - \frac{2q(1+q^2)}{1+q} - 2\frac{1+3q-4q^3}{1+2q} - q^2
\end{equation*}
\noindent which can be negative, e.g.~when $q = 1/2$, a contradiction.
\end{proof}

Using Lemma~\ref{thm:positive_def_inv}(ii) one can show that $\uncon_{FS} \not\ge 
\uncon_{PS}$ implies $\uncon_{PS}^{-1} \not\ge \uncon_{FS}^{-1}$ which in turn implies
$\con^+_{FS} \not\le \con^+_{PS}$ (see Lemma \ref{lem:uncon_con}). 
Since we earlier showed that PS is enormously worse than FS for variances, this result
is surprising, particularly as experience shows that it is not difficult to find other
examples for much larger $W$.  We can explain this quandary as follows.  When we focus
on the variances in isolation we highlight the poor performance of PS. However, linear
combinations such as $\sum_k a_k\theta_k$ bring in cross terms, which for PS are negative
because of strong ambiguity in its observations. Strong correlations are a
bad feature of a covariance matrix of an estimator, however if they are negative, they can in some cases cancel other
positive terms resulting in a lower total variance. Hence, paradoxically, it is the
poor behaviour of PS which prevents $\uncon_{FS} \ge \uncon_{PS}$ from holding.

\smallskip


The last example reveals that a ranking of methods based on a positive semidefinite
comparison, although very desirable when true, is not possible in general. 
We therefore turn to a more generally applicable approach in the next section which 
we use for the remainder of the paper.

\subsection{Variance Comparisons: a Partial Ranking}
\label{ssec:var_comp}

In this section we focus on comparing methods via the diagonal elements of $\con^{+}$,
corresponding to the optimal variances of $\theta_i$ estimates, just as in the figures
above. As before, it is sufficient to consider the unconstrained covariance matrix
$\uncon^{-1}$ since  $\con^{+} = \uncon^{-1} - \vth\vth^\T$ by (\ref{eq:inverse_simple}).
To the extent possible, we seek to establish a hierarchy between methods based on diagonal
comparisons. Thanks to Theorem~\ref{thm:dpi_seq}, which applies in particular to diagonal
elements, we already have the answer in some cases.

We begin by comparing PS and PS+SYN.  The examples given so far showed PS+SYN as
enormously superior to PS.  However, this is not always the case. A counterexample under
the PPR normalization with $p = 0.05$ is given by
\be
  \vth = \{\theta_1, \theta_1, \theta_{3 \le k \le 10} = \theta_1/100 \}, \  \theta_1\!=\!0.4808,
    \label{eq:degenerate}
\ee
for which $D\approx1.69$.
As seen in Figure~\ref{fig:comp_degenerate}, PS outperforms PS+SYN for $\theta_9$ and 
$\theta_{10}$, but not otherwise. This makes sense given the bias of PS toward sampling
large flows. 

For these same parameters PS+SYN outperforms PS for all $\theta_k$ under ESR.
This is typically the case. For example, with the above parameter set, PS+SYN 
outperforms PS on all sampling rates.
While it may be true that PS outperforms PS+SYN under ESR in some rare situations, 
we believe this is not the case, and conjecture that PS+SYN defeats PS under ESR for all $\vth$. 
For small $p$, we can prove that this is the case.
\begin{thm}
Under PPR and ESR normalization, for small enough $p$, $(\uncon^{-1}_{\text{PS+SYN}})_{jj}
\le (\uncon^{-1}_{\text{PS}})_{jj}$ for any $\vth$,  $j=1,\ldots,W$.
\label{thm:ps_pssyn_comp}
\end{thm}
\begin{figure}[h]
  \PSPSSComp
  \caption{PPR comparison between PS and PS+SYN with $\pp = 0.05$ for a particular
distribution highly skewed to small flows.}
  \label{fig:comp_degenerate}
\end{figure}

\medskip
We now consider the above comparison after both methods have benefitted from the use of sequence numbers.
\smallskip

\begin{thm}
Under PPR and ESR normalization, for every $3 \le j \le W$,
$(\uncon^{-1}_{\text{PS+SYN+SEQ}})_{jj} \le (\uncon^{-1}_{\text{PS+SEQ}})_{jj}$.
\label{thm:pssynseq_comp}
\end{thm}
\smallskip

For each of $j=1$ and $2$ counterexamples can be found for certain $\vth$ and $p$, under
both PPR and ESR. For example PS+SEQ has lower variance for $\theta_1$ for the parameter
set (\ref{eq:degenerate}) with $p=0.05$ under both PPR and ESR, and for $\theta_2$ the
family $\vth = \{a,1-2a,a\}$ with $a<0.1$ gives examples for wide ranges of $p$, including
as $p\rightarrow0$ for $a$ small enough. The counterexamples occur in atypical situations
where $\theta_1$ or $\theta_2$ are large relative to other $\theta_j$, and can be excluded
if these are appropriately controlled. For example if $p<1/2$ it can be shown that when
$\theta_2/\theta_3 <\qp^2$, then PS+SEQ is worse under PPR (and hence ESR).
Given that PS+SYN+SEQ defeats PS+SEQ except for perhaps $\theta_1$ or $\theta_2$ under atypical conditions, 
in general PS+SYN+SEQ can be regarded as superior.

\smallskip
The picture emerging from the above comparisons is that PS+SYN+SEQ is clearly superior to 
PS and PS+SEQ. Since PS+SYN+SEQ is just a special case of DS, we now consider this family
in more detail. We focus on the ESR normalization which, apart from being far more 
important than PPR, is the key to our optimality result in Section~\ref{sec:compcompare}. 
The following is one of our main results, a detailed characterization of the performance 
of DS under ESR. It can be shown that an analogous result does not hold for PPR.

\begin{thm}
The diagonal elements $2 \le j \le W$ of $\uncon^{-1}_{\text{DS}}$ under the ESR 
normalization are monotonically decreasing in $\pp$. The property holds for $j=1$ iff the
condition
\be
   \theta_2 \ge \frac{D-1}{D}\theta_1 (1-\theta_1)
   \label{eq:mono_cond}
\ee
is satisfied. Also, monotonicity holds when $D = 1$ or $W=2$.
\label{thm:ds_mono}
\end{thm}

The above result shows, provided (\ref{eq:mono_cond}) is satisfied, that the variance
bounds for each $\theta_k$ under DS drop as $\pp$ increases. It follows that the optimal
point $\ppopt$ in the DS family lies  at $\ppopt=1$, that is flow sampling! The
superiority of FS clearly demonstrates that the problem of inverting from imperfect
sampled flows is so difficult that in the tradeoff between flow quality and quantity,
quality wins convincingly.  

The exception is for $\theta_1$ when (\ref{eq:mono_cond}) fails, in which case the optimal
DS lies to the left of $\pp=1$. However, this only occurs when $\theta_2/\theta_1$ is
extremely small, and even then in most cases $\ppopt \approx 1$ unless $D$ is also
very large. Consider the region $\pp \to 1$. By solving for the optimal $\pp^\star$ using
(\ref{eq:varone_diff}), we obtain (provided $D > 1$).
\be
    \pp^\star = \sqrt{\frac{\theta_2}{(D-1)\lbrack \theta_1(1-\theta_1)- \theta_2\rbrack}}.
     \label{eq:varone_diff_sol}
\ee
For $\pp^\star$ to be much smaller than 1, we require $D$ to
be large and the ratio $\theta_2/\theta_1$ very small (effectively, $\theta_2$ being
negligible). Such a scenario is very unlikely to occur in networks and in any case only
affects $\theta_1$, and so it is reasonable to assume that the optimal DS corresponds to
FS in practice.

\smallskip
Our next result compares the final method, SH, against FS. 
\begin{thm}
Under PPR and ESR normalization, for every $2 \le j \le W$, 
$(\uncon^{-1}_{\text{FS}})_{jj} \le (\uncon^{-1}_{\text{SH}})_{jj}$. For $j=1$ the
condition holds for any $\vth$ for $p$ sufficiently small. Sufficient conditions
independent of $p$ are either $W=2$, or for $W > 2$, 
\vspace{-1mm}
\be
   \theta_2 \ge \theta_1(1-\theta_1) .
   \label{eq:mono_condSH}
\ee

\label{thm:sh_comp}
\end{thm}
\vspace{-2mm}
The proof for (\ref{eq:mono_condSH}) involves bounds which are tight as $p\rightarrow1$.
Hence, when the condition is violated and $p$ is large enough, SH can indeed outperform FS 
for $\theta_1$. An example is given in Figure~\ref{fig:sh_comp_degenerate} with $W=6$ and
$\pp = 0.4351$, which is large and reasonably close to $p=0.5$. However, this effect is
difficult to see in distributions with larger $W$ as SH shifts most of its packet budget
to larger flows, resulting in $\pp$ being orders of magnitude smaller than $p$ and the 
bounds leading to the sufficient condition becoming very loose. In all cases of interest
therefore, FS can be regarded as superior to SH at all flow sizes.

\begin{figure}[b]
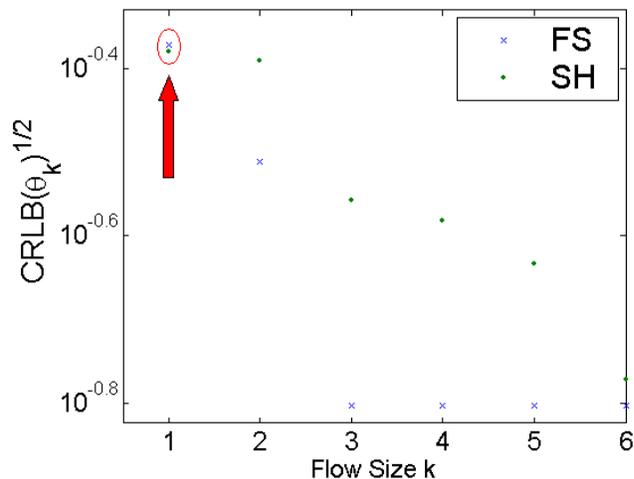

  \SHFSComp
  \caption{PPR and ESR comparison between FS and SH with $p = 0.5$ for a particular distribution highly skewed to small flows.}
  \label{fig:sh_comp_degenerate}
\end{figure}
\smallskip

Finally, we clarify the relationship between DS and SH more generally. DS is a two
parameter family whose performance varies in a wide range. Indeed,
Figure~\ref{fig:shortCompareESR} already shows that it can perform better than SH at the
`FS' end of its range, and worse at the other end. The important observation is that in
all cases where FS outperforms SH for a given $\vth$, it follows from continuity of the
CRLB in $(\pp,\pf(p))$ that there are members of the DS family other than FS which do
likewise. Theorem~\ref{thm:sh_comp} shows that this is true in almost all cases under ESR, 
allowing us to conclude that in general DS outperforms SH at all flow sizes. 

The natural counterexample to this general rule is when (\ref{eq:mono_cond}) is satisfied 
so that the best DS can do is given by FS, and yet $\vth$ is such that SH defeats FS (for 
$\theta_1$), implying that SH defeats all members of the DS family for $\theta_1$. This 
can occur if  (\ref{eq:mono_condSH}) is not satisfied and $p$ is large. Furthermore there
are cases where both SH and the optimal DS defeat FS, in which case the comparison is more 
complex. Since however this battle is contained within a very small and unimportant region 
of $(\vth,p,\pf(p))$ space, we do not characterise these counterexamples fully but instead 
conclude with the following result, which for the first time exhibits specific DS members 
(other than those in a neighbourhood of FS) which defeat SH.

\begin{thm}
Under ESR normalization with rate $p$, if $0\le p_{p,\text{SH}} \le \frac{pD - 1}{D-1}$,
then a sufficient condition for $(\uncon^{-1}_{\text{DS}})_{jj} \le
(\uncon^{-1}_{\text{SH}})_{jj}$ for every $1 \le j \le W$
is that $\ppd$ satisfies
\ben
   p_{p,\text{SH}} \le \ppd .
\een
\label{thm:ds_sh_performance}
\end{thm}

\noindent \textbf{Summary:} We have confirmed that a ranking of methods based on a
positive semidefinite comparison is impossible, since even a diagonal comparison does not 
yield a simple hierarchy. Nonetheless, ignoring the counterexamples which sometimes occur 
for the CRLB variance for $\theta_1$ and perhaps $\theta_2$, a clear overall picture 
emerges from the comparisons above. First, from Theorem~\ref{thm:dpi_seq},  the
utilization of sequence numbers yields consistent information theoretic dividends. In
particular, PS need not be considered further since it has extremely poor performance and
by Theorem~\ref{thm:dpi_seq} PS+SEQ is better. Theorems~\ref{thm:ps_pssyn_comp} and \ref{thm:pssynseq_comp}
show that the use of SYN sampling as a technique to eliminate bias is very powerful, 
which makes PS+SYN+SEQ the leading candidate, and focussed attention on its 
generalization, DS. Theorem~\ref{thm:ds_mono} then shows that DS out-performs PS+SYN+SEQ
under the ESR normalization (the most relevant one for estimation performance) provided 
$\pp>\pp^*$, and furthermore that FS is the favored member of DS.  
Theorem~\ref{thm:sh_comp} shows that FS is also superior to SH. In conclusion, FS is
the best sampling method for all flow sizes. This strongly motivates the adoption of 
sampling methods which approach FS as closely as possible, but which are efficiently
implementable. This refocuses attention back to DS as it enjoys both of these properties.

\section{Analysis of the SEQ Dividend}
\label{sec:seq_div}

The utilization of sequence numbers provides additional `virtual' packets for free. The
extent to which this improves our information on $\vth$ clearly depends on parameters. For
example, flows with only $k=1$ or $2$ packets are not helped at all. Intuitively (for
example, under PS) it is the flows for which $k\gg 1/\pp$ which will receive the most
benefit: since only the packets before the first and after the last physically sampled
packets will not be recovered using sequence numbers, and on average there are $2/\pp$ of
these since $1/\pp$ is the average gap between two physically sampled packets within a
flow. 

To quantify the benefit of sequence numbers better we define the \emph{effective packet
gain}, which is the ratio $r=\E[\tilde N_k]/\E[N_k]$, where $\tilde N_k$ and $N_k$ are
the number of SEQ assisted and unassisted (physical) packets sampled respectively from a 
flow of size $k$. By definition, $R = \tilde N_k / N_k\ge 1$ and so $r\ge1$, and
asymptotically the gain saturates at $r= 1/\pp$.  

For DS it can be shown (see \techrep\ for more details) that to obtain a ratio $\alpha$ of
this maximum gain, a flow must have a size of the order of $\frac{\qp(1+\alpha)}
{\pp(1-\alpha)}$, suggesting that flows must have a size of roughly $1/\pp$ or more before
the sequence number `information dividend' effectively switches on. Although intuitively
appealing, this is however quite misleading. There is a high variance associated with the
random variable $\tilde N_k$, which for large flows under PS+SYN+SEQ goes as 
\ben
   \Var(\tilde N_k) \approx \frac{4}{\pp} + k(k-2)\pp + (2k-3).
\een
Therefore the gain $R$ is not sharply concentrated around its mean and is highly dependent 
on the sampling parameter. 

More generally, the sequence number dividend `switch' is primarily about whether at least 
two packets are sampled. When $p$ is small this is a very rare event, but can nonetheless 
be responsible for most of the available information about flow size. For example in
Figure~\ref{fig:shortComparePPR} the probability of sampling two packets is of the order 
of $10^{-6}$, and yet the ratio of variance of PS to PS+SEQ is around 400 -- even rare 
dividends are worth having! Thus even for smaller to medium sized flows (`mice' and
`rabbits'), the use of sequence numbers provides a huge reduction in estimator variance,
as shown by numerical and empirical evidence throughout this paper.  

It is interesting to compare the SEQ dividend with SH, which implements a kind of extreme
SEQ effect in the sense that if only a single packet is sampled, then all subsequent ones 
will be. This is a more powerful `switch' than the SEQ mechanism which requires at least 
two packets.   For large flows however, packets inferred under DS approach this level of 
packet recovery since, because the SYN packet is sampled, only a geometric number of 
packets will be missed at the end of the flow, compared to a geometric number at the 
beginning under SH. Of course, SH achieves this solely from its real packet budget 
whereas DS does not, which explains why DS outperforms SH even in the latter's area of 
strength, namely very large flows.


\section{Computational Issues}
\label{sec:compcompare}

\subsection{Optimizing DS with Computational Constraints}

From our result on the monotonic behaviour of ESR normalized DS, what is optimal in
terms of information is clear: simply move down the ESR constraint curve to FS. However,
there are other constraints from the resource side which will constrain the parts of the
curve that will be accessible. The solution to the joint information/resource
optimization problem is therefore clear:  move as far down the ESR curve as possible.
Our task now is to determine those constraints and hence the region in the $(\pp,\pf)$
plane which is feasible.

The bottlenecks in the implementation of any sampling method are the memory access time
and memory size. CPU processing is included in the memory access, since the
CPU has to read out values from memory during lookup, performing modifications and
write-backs when necessary. These tasks constitute the main portion of what is required
of the CPU for measurement, as the measurement process is basically all about efficient
counting \cite{VargheseNA}. 
This can be done for example with a hybrid SRAM-DRAM counter achitecture \cite{Shah02,
Ramabhadran03, Zhao06}, where a small amount of (fast but expensive) SRAM is used for the
counter and its value periodically exported to a (cheaper but slower) DRAM store. With
about 10 ns access time for SRAM, such a counter can be implemented even for OC-768 links
which run at 40 Gbps.

We now consider how to optimize DS based on these
constraints. Regarding flows, let $T_{max}$ be the maximum flow table size measured in
terms of flow records and $\lambda_F$ be the flow arrival rate. As for packets, let $C$
be the capacity of the link, $P$ the size of the smallest packet ($\approx$ 40 bytes),
and $\tau$ the access time of memory (nominally DRAM).

Our simple analysis of the above bottleneck constraints is based on the following.
In terms of packet arrivals we assume the worst case, namely the
smallest size packets arriving back-to-back at line rate. By bounding
the processing in such a case we guarantee that the front line of packet processing
(which occurs at the highest speeds) is not under-dimensioned.
In terms of flow arrivals we assume the `average case' based on the
average number of active flows. This can easily be made more
conservative by replacing the average by some quantile to take into
account fluctuations in flow arrivals.
For simplicity, we ignore data export constraints from the measurement center to a central
data collection center. We also do not consider rate adaptation based on the traffic
condition although such schemes are compatible with DS.

Consider the processing of a single packet. The SYN bit is first tested to see which
sampling parameter will apply, the cost of this is negligible, and if a packet is not
sampled no further action is needed. Now consider the cost of a packet which is sampled.
Each SYN packet which is sampled is inserted into the flow table. No prior lookup is
necessary since it must be the first packet of its flow. Each non-SYN packet which is
sampled must first perform a lookup of its flow-ID in the flow table to see if that flow
is being tracked. If not it is discarded. The cost of this wasteful per packet
implementation of the `flow discarding' step (inherent to any SYN based method such as
DS) is not the bottleneck because the following case is the most expensive of all and is
the one we model: a non-SYN packet which is sampled and whose flow \textit{is} being
tracked requires both a lookup followed by an update. 

Using the above, the constraints are
\be
\pf \le \widehat \pf = \min{(\frac{T_{max}}{D\lambda_F}, \frac{P}{\tau C})}, \quad \pp \le
\widehat \pp = \frac{P}{2\tau C}.
\label{eq:comp_constraints}
\ee
The constraint $T_{max}/(D\lambda_F)$ ensures that the average number $D\lambda_F$ of
active flows does not exceed the flow table size. Constraint $P/(\tau C)$ provides the
worst case bound for per packet processing, for a single operation (insert or update).
The factor of 2 that appears in the denominator for $\widehat \pp$ is to account for the
worst case, in which a lookup \emph{and} update is needed. The analysis is based on the
use of a single per flow counter to track flow size. In practice, there may be more
counters needed to track other quantities of interest, necessitating tighter
constraints.

In the sequel, our examples consider traffic mixes that have manageable numbers of SYN
packets, so that from (\ref{eq:comp_constraints}), $\widehat \pf = T_{max}/(D\lambda_F)$.
This is done mainly to illustrate the relationship between the CRLB and parameter $\pf$.
Indeed, at lower link speeds (OC-48 for example), $\pf$ is determined by the number of
flows arriving at the measurement point, rather than packet processing time. Secondly,
to increase accuracy, from the previous discussions, we would like $\pf < \pp$,
i.e.~fewer, but better quality sampled flows.

The constraints form a simple region on the $(\pp,\pf)$ plane that is convex,
since it is rectangular with a corner at the origin. We want to minimize the variance for each $\theta_k$ subject to
these constraints. Since the ESR curve is convex with respect to $\pf$ and
$\pp$, the optimal value must lie on the vertex of the convex constraint set
\cite[Corollary 32.3.1, p.~344]{Rockafellar70}.
For this to hold, we require that the optimum for $\theta_k$ on the ESR curve (Section~\ref{ssec:var_comp}) 
is outside the constrained region.  This will be the case for all $k$ for any reasonable traffic mix.
Under such conditions, the solution is therefore $\pf = \widehat \pf$ and $\pp = \widehat \pp$.

We then have a relationship between the flow table size and link capacity and the
diagonal elements of $\uncon^{-1}_{\text{DS}}$. Since it is apparent from (\ref{eq:DS_diags})
that the diagonals are dominated by $1/\pf$, by substituting the optimal solution, we
have for $1 \le j \le W$,
\ba
\nonumber
\hspace{-4mm} (\uncon^{-1}_{\text{DS}})_{jj}\!\! & = &\!\! O(\frac{1}{\widehat \pf}) = O(
\frac{1}{T_{max}}).
\ea
Thus, with a larger table size (i.e.~more memory available), variance of the estimator
can be reduced.

As for the relation to capacity, we assume that $C$ is large and
use the approximation $\widehat \qp \approx 1$. This can
be justified considering that sampling is required beyond OC-3 link speeds. 
The diagonals become
\ba
\nonumber
\hspace{-4mm} (\uncon_{\text{DS}}^{-1})_{11}\!\! &=&\!\!
   \frac{\theta_1(1-\theta_1)}{\widehat \pf} + \frac{1}{\widehat \pf \widehat \pp}\sum_{k=2}^W \theta_k + \theta^2_1
\\
\nonumber
\hspace{-4mm}  (\uncon_{\text{DS}}^{-1})_{jj}\!\! &=&
    \frac{\theta_j(1-\widehat \pp\theta_j)}{\widehat \pf \widehat \pp} + \frac{1}{\widehat \pf \widehat
    \pp}\sum_{k=j+1}^W 2\theta_k + \theta^2_j \\
\nonumber
    \hspace{-4mm}  (\uncon_{\text{DS}}^{-1})_{WW}\!\! &=&
    \frac{\theta_W}{\widehat \pf \widehat \pp} - \frac{1}{\widehat \pf}\theta_W^2 + \theta^2_W
\ea
which are approximately inversely proportional to $\widehat\pp$, and hence are $O(C)$. 
We conclude that the variance of DS is inversely proportional to memory usage and proportional to
the link capacity.

As an example, consider an OC-192 link with $D\lambda_F = 1 \times 10^6$ flows/sec,
$T_{max} = 100,000$ and an access time of 100 ns for DRAM. Let us assume a further 
100 ns is required for further processes (e.g. sequence number information). Thus, $\widehat\pf =
0.1$ and $\widehat\pp = 0.08$. With our numerical evaluation on the Leipzig-II trace (discussed
in the following section), this would imply that the trace contain at least $9.5 \times 10^8$ original flows in
to achieve a standard deviation of $10^{-8}$ or better. If we compare this to PS
with a sampling rate of $\pp = 0.1$, we require a staggering $5.6 \times 10^{44}$ flows
to achieve the same performance!

To observe the dependence on memory and link capacity, now consider an OC-768 link
instead with $T_{max} = 10,000$. This time, we have $\widehat\pf = 0.01$ and $\widehat\pp = 0.02$. The
number of flows required to achieve the same standard deviation now increases to $8.8
\times 10^9$ which is still orders of magnitude better than PS. 


If we consider the SRAM-DRAM architecture discussed earlier, with state-of-the-art SRAM
having access times of about 5-10 ns, this can only increase the value of $\widehat\pp$.
Ideally, we still keep $\pf$ low to reduce the number of flow entries while increasing
$\pp$ to approach FS performance. Continuing on from the previous example of the OC-192
link and assuming $\tau=5$ ns, $\widehat\pp$ can now approach 1. For the OC-768 example, 
$\widehat\pf =0.01$ while $\widehat\pp = 0.8$, giving huge performance gains. 


\subsection{Other Issues}

Checking for the SYN bit can be done in a simple way by testing the
payload type in the IP header and then verifying the presence of the SYN bit. This takes
much less effort than deep packet inspection systems. Furthermore, sampling decisions
can be implemented using precalculated values, much faster than a straight random
number generator implementation.

We address the problem of flow table overload by using the method proposed by Estan et
al.~\cite{Estan04} by defining discrete measurement time bins, where sampled flows are
exported at the end of each time bin. Consequently, overload of the flow table can be
avoided at the cost of increased export rate.

\section{Case Study on Internet Data}
\label{sec:data}

In this section we test the performance of DS on two traffic traces under the ESR
normalization, and compare it to FS and its closest competitor, SH.
We also further examine the benefit of sequence numbers both in a fully empirical setting, and an idealised one.
Since we work with real data, we require an unbiased estimator for each method. 
We propose closed-form estimators that achieve the CRLB asymptotically. We test their
statistical performance by examining how close they approach the empirical CRLB
on the traces we used.
In general, the empirical results match the theoretical ones from earlier sections well.

\subsection{Data Traces}

We used two publicly available network traces,  Leipzig-II \cite{Leipzig} and
Abilene-III \cite{AbileneIII}, which are each collections of anonymized packet headers
passing through a single router. Since the raw traffic is at packet level, specialized software such 
as CoralReef \cite{CoralreefCAIDA} are required to reconstruct flows. We modified the software
 so that only TCP flows were analyzed.
Summary statistics of these traces appear in Table~\ref{table:leipzig_stats}. Both
traces are unidirectional, presenting problems when constructing a sequence number
function, as elaborated later. 

There are many flows whose SYN packet is missing as they began before 
the measurement interval. To be consistent with our model assumptions, we remove these.
A similar situation was encountered in \cite{duffsamplingToN2005}. 
Fortunately, such flows are in the minority, for example with Leipzig-II they account for only 18\%.
Furthermore, when sampling the trace, we assume an infinite timeout, that is, flows are
expired at the end of the measurement interval. This is in accordance with the fact that we do not 
consider flow splitting. 
In practice, timeouts would split a flow, resulting underestimation of flow size.

Note that some flows may be malformed and have one or more SYN packets within the flow.
Potentially, DS may sample one of these packets and treat it as a new flow, instead of
part of a longer one. However, such cases are rare. In Leipzig-II, there are only 468 of
such flows, or $\approx 0.02\%$, and there are none in Abilene-III. These flows were left in the trace.
\begin{table}[h]
\begin{tabular}{|p{1.3cm}|p{1.3cm}|p{2cm}|p{1.5cm}|p{0.8cm}|}
\hline
\small \centering \textbf{Trace} & \centering \small \textbf{Link Capacity} &
\centering \small \textbf{\!\!Active TCP Flows\!\!} & \centering\small \textbf{Duration
(hh:mm:ss)} & \centering \small $D$
\tabularnewline
\hline \hline
\centering Leipzig-II & \centering 50 Mbps & \centering 2,277,052 & \centering 02:46:01
& \centering 19.76
\tabularnewline
\centering Abilene-III & \centering 10 Gbps & \centering 23,806,285 & \centering 00:59:49
& \centering 16.12
\tabularnewline 
\hline
\end{tabular}
\caption{Summary of the traces used}
\label{table:leipzig_stats}
\end{table}

The value of $D$ in Table \ref{table:leipzig_stats} is the actual average 
flow size. In our experiments, we truncate Leipzig-II to $W=1000$ and Abilene-III to
$W = 2000$, resulting in $D$ being 1.94 and 7.65 packets for each trace respectively.  
Truncation is performed by discarding all flows with size above $W$, which ensures that the assumption $\theta_k > 0$, $k = 1,2,\ldots$, $W$ is met.

\begin{figure}[t]
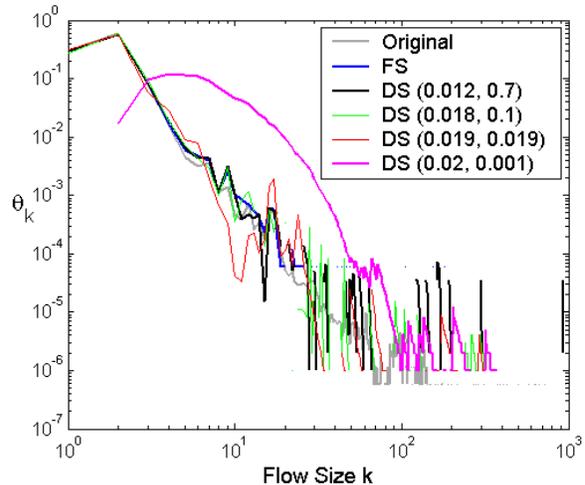

  \ESRcompareone
  \caption{Comparison of DS on Leipzig-II, using $W=1000$ and varying parameters
	under ESR normalization with $p = 0.01$.}
  \label{fig:esrcompare1}
\end{figure}

\begin{figure}[t]
  \ESRcomparetwo
  \caption{Comparison of DS on Leipzig-II, using $W=1000$ and varying parameters
	under ESR normalization with $p = 0.001$.}
  \label{fig:esrcompare2}
\end{figure}

\begin{figure}[t]
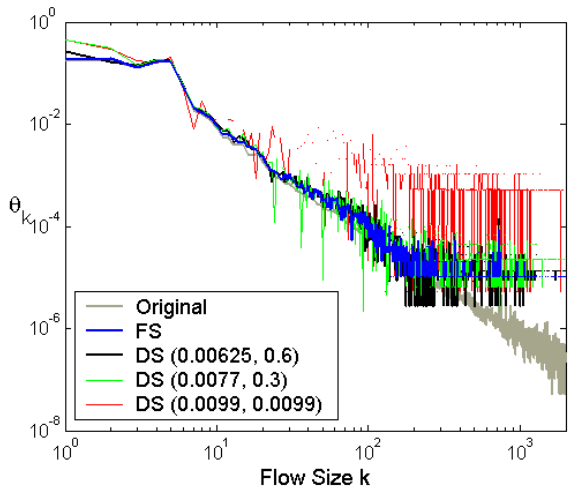

  \ESRabiper
  \caption{Comparison of DS on Abilene-III, using $W=2000$ and varying parameters
	under ESR normalization with $p = 0.005$.}
  \label{fig:esrabiper}
\end{figure}

\begin{figure}[b]
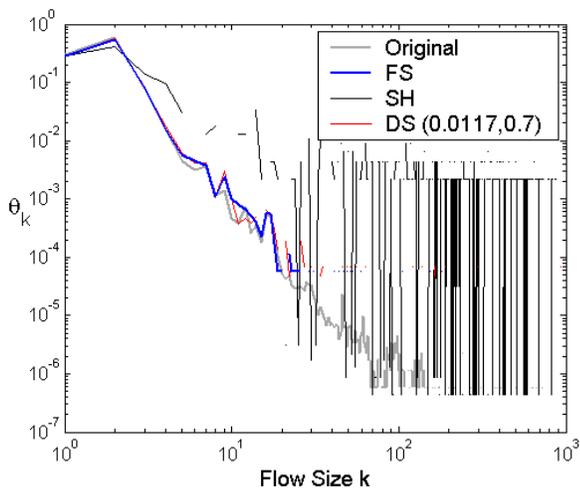

  \ESRshcompare
  \caption{Comparison of FS, SH ($p_r = 0.0002$) and DS ($\pf = 0.0117$, $\pp = 0.7$)
	on Leipzig-II, using $W=1000$ and varying parameters under ESR normalization with $p =
	0.01$.}
  \label{fig:esrleipsh}
\end{figure}

\begin{figure}[b]
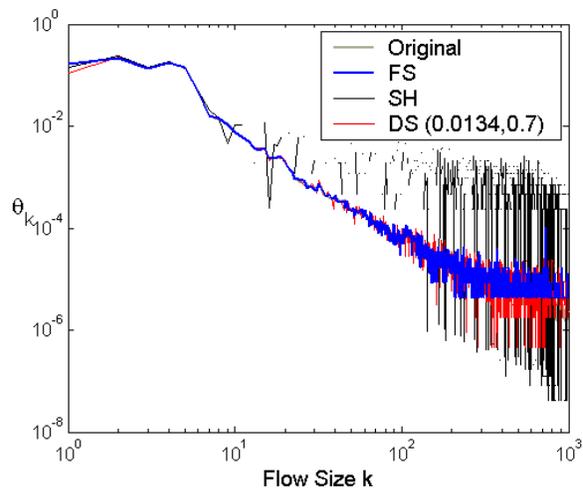

  \ESRabish
  \caption{Comparison of FS, SH ($p_r = 0.00018$) and DS ($\pf = 0.0134$, $\pp = 0.7$)
	on Abilene-III, using $W=1000$ and varying parameters under ESR normalization with $p
	= 0.01$.}
  \label{fig:esrabish}
\end{figure}

\subsection{Closed-Form Unbiased Estimators}
\label{ssec:closed_form_mle}

The analysis of earlier sections were centered on the CRLB, which bounds what is achievable by any unbiased
estimator, but it is neither a construction of such an estimator nor even a proof of its
existence. When working with real data an actual estimator is required, ideally one that
achieves the previously computed CRLB. The maximum likelihood estimator (MLE) is an
attractive candidate as it is {\em asymptotically efficient},
guaranteeing that its performance approaches the CRLB asymptotically \cite{KayVolI}. 
However, the MLE, although unbiased asymptotically, is in general biased, especially
in the small sample regime.

We propose the following estimator for FS and DS (or other SYN based methods such as 
PS+SYN):
\be
   \veth = \frac{\mathbf{\tilde B}^{-1}}{\nf} [M'_1,M'_2,\ldots,M'_W]^\T .
   \label{eq:alt_mle}
\ee
This estimator has a natural interpretation as an empirical histogram based on observed sampled flow counts, inverted by
$\mathbf{\tilde B}^{-1}$ to the original flow distribution. 
It is easy see that it is unbiased since
$\E\lbrack \veth \rbrack = \mathbf{\tilde B}^{-1} \mathbf{\tilde B} \vth = \vth$.  The
matrix $\mathbf{\tilde B}^{-1}$ can be computed explicitly for these methods, as shown in
Section~\ref{sec:methods}.

Our estimator is further motivated by its relation to the MLE, as we now show (all
proofs are deferred to Appendix~\ref{sec:mle}).
\begin{thm}
If the sampling matrix of a method satisfies $\mathbf{b}_0 = q \mathbf{1}_W$, 
then the MLE is
\be
  \veth = \mathbf{\tilde B}^{-1} \left( \frac{p}{\sum^W_{j = 1} M'_j} [M'_1,M'_2, \ldots,M'_W]^\T   \right).
  \label{eq:closed_form}
\ee 
\label{thm:mle_closed}
\end{thm}

By rewriting $\sum^W_{j = 1} M'_j$ as $\nf(1 - M'_0/\nf)$
and observing that $M'_0/\nf$ converges to $q$ with high probability 
(this follows from Hoeffding's inequality \cite[p.~303]{MitzenmacherProb05}), 
the MLE (\ref{eq:closed_form})  tends to our  estimator (\ref{eq:alt_mle}), which therefore 
approaches the CRLB asymptotically. Our proposed estimator only obeys the constraint 
$\onew^\T \veth = 1$ on average, while the estimator in Theorem \ref{thm:mle_closed} always obeys 
the constraint.
(Note that (\ref{eq:closed_form}) remains a viable estimator in the conditional framework, see 
Remark \ref{rem:conditional_mle} in Appendix \ref{sec:mle}.) 

\smallskip
\noindent The following applies to SH.
\begin{thm}
The MLE for SH is given by
\be
  \veth_{\text{SH}} = \frac{\mathbf{\tilde B}_{\text{SH}}^{-1}}{\nf}\cdot[p(M'_0+M'_1),M'_2,\ldots,M'_W]^\T.
  \label{eq:mle_sh_closed_form}
\ee
\label{thm:uncond_sh_mle}
\end{thm}

This time the MLE is unbiased (see Appendix \ref{sec:mle}), so we use it directly on the data. 
It closely resembles the
estimator (\ref{eq:alt_mle}), with a slight difference: for the estimate of $\theta_1$,
the number of missing flows $M'_0$ plays a significant role.
This can be interpreted as follows.
Since SH works by geometrically skipping the first few packets in a flow,  flows of size 1 are those most
likely to have been entirely missed. Hence, the simplest way  to incorporate a
knowledge of $M'_0$ is to assume that it arises solely from evaporated flows originally
of size 1. 

The advantage of simple closed form estimators such as these, 
which only require a matrix multiplication, is that they eliminate the need for iterative 
estimation algorithms such as Expectation-Maximization (EM), often employed in the 
literature \cite{duffsamplingToN2005,Ribeiro06fischer,YangMichailidisGlobe06}. From a
computational viewpoint, this is highly advantageous.

\subsection{Testing with a Perfect Sequence Number Function}

We begin our case study by testing DS with a \textit{perfect} sequence number function,
which returns the exact number of packets between two sampled packets.  
As we have access to the original unsampled flows, this is easily evaluated.
Apart from being a benchmark 
for sequence number functions that may be designed in the future, a perfect function allows clean comparisons between alternative methods employing sequence numbers to be made.

In Figure~\ref{fig:esrcompare1}, for an ESR of $p = 0.01$, values of $\pp$ from $\pp = 1$ (equivalent 
to FS) steadily decreasing to $\pp = 0.001$ are shown (DS with $\pf = \pp = 0.019$ corresponds to
PS+SYN+SEQ).
As $\pp \to 1$ the performance vastly improves.
Similar results were observed in Figure~\ref{fig:esrcompare2}, where $p =
0.001$.    In all cases, a chronic lack of samples at the tail end of the distribution causes
inaccurate estimates, with zero samples showing as discontinuities. 
FS holds to the true distribution the longest, as we expect, and the best DS performs similarly to it.
With a much larger sample set in Figure~\ref{fig:esrabiper}, we can
see better agreement between the estimates and the original distribution.

\begin{figure}[t]
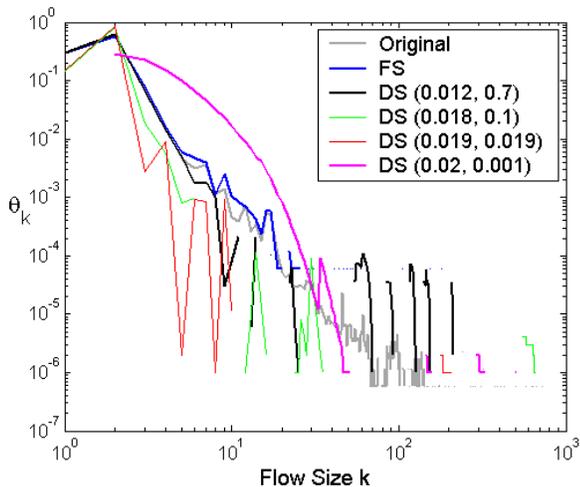

  \ESRseqcompone
  \caption{Comparison of DS on Leipzig-II, using $W=1000$ and varying parameters
	under ESR normalization with $p=0.01$. An imperfect sequence number function is used 	 
	here.}
  \label{fig:esrcompare3}
\end{figure}

\begin{figure}[t]
  \ESRseqcomptwo
  \caption{Comparison of DS on Leipzig-II, using $W=1000$ and varying parameters
	under ESR normalization with $p=0.001$. An imperfect sequence number function is used
	here.}
  \label{fig:esrcompare4}
\end{figure}

\begin{figure}[t]
  \ESRabiseq
  \caption{Comparison of DS on Abilene-III, using $W=2000$ and varying parameters
	under ESR normalization with $p = 0.005$. An imperfect sequence number function is
	used here.}
  \label{fig:esrabiseq}
\end{figure}

\begin{figure}[t]
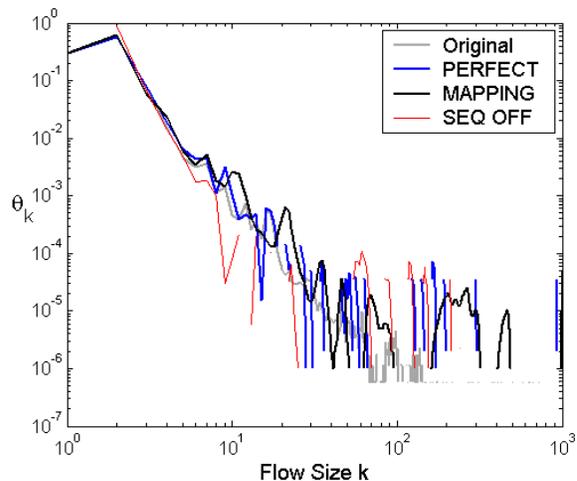

  \datacompare
  \caption{Benefits of using sequence numbers. Three cases are shown: PERFECT uses a
  perfect sequence number function, MAPPING uses an imperfect function and SEQ OFF
  uses no sequence number information.}
  \label{fig:datacompare}
\end{figure} 

We also tested SH, as seen in Figures~\ref{fig:esrleipsh} and \ref{fig:esrabish}. 
To simplify the comparison across the traces, we truncate each to $W =1000$, so that 
$D$ becomes 1.94 and 6.45 packets for Leipzig-II and Abilene-III respectively, and use the
same ESR value $p=0.01$. In both cases the performance is much worse for SH than DS, which
tracks FS and the true distribution well. In Leipzig-II the performance is very poor
almost everywhere, while in Abilene-III the front end of the distribution, up to about
flows of size 8, is estimated quite well, before deteriorating badly at larger sizes.
The good performance of SH at $j=1$ can be attributed to the fact that this event is 
dominated by the sheer number of original flows of size one. Both FS and DS tracks the
distribution much better since both methods have sampled flows of far superior quality to
SH.

\subsection{Testing with an Imperfect Sequence Number Function}

We now test DS with an imperfect sequence number function. Our function
is similar to that outlined in \cite{Ribeiro06fischer}. However, as we do not
have statistics of popular TCP payload sizes available, we infer the most likely payload
size as follows. If the sequence number difference is divisible by a popular payload
size (for example, 1460 bytes), we take this as the most likely payload size.
Otherwise, we use the average payload size. 
This function is subject to errors, especially when a flow has variable payloads, however it suffices for our
purposes. 

In addition to TCP sequence numbers, we
exploit IPID numbers. As mentioned in  \cite{Ribeiro06fischer}, the IPID field of Linux
machines is incremented sequentially for each TCP flow every time a packet in the flow
is transmitted. Given that the majority of web-servers on the Internet are Linux-based,
we exploit IPID numbers to check the accuracy of our estimate when a flow has packets
with variable payloads.

Furthermore, the unidirectional nature of the trace presents a significant challenge.
As one side in a TCP connection usually transmits more data than the other, some sampled
flows may consist mainly of TCP ACK packets, which means that they will have zero-byte TCP payloads.
In this case sequence numbers may not be incremented at all, and so will not provide information about the number of bytes transmitted. A solution is to use the TCP
acknowledgement numbers instead to infer the number of bytes transmitted from the opposite
direction, which would yield an estimate of the number of packets in the TCP ACK flow.
This may not be the most ideal solution, as this method would be susceptible to delayed
ACKs, thus underestimating the size of the flow.

Modern web browsers rely on maintaining TCP connections rather than initiating new connections,
which require more memory. This is the persistent HTTP protocol. The prevalence of this protocol
amongst web browsers presents a challenge, since empty payload packets are periodically sent
to keep the connection alive. These packets do not increment the sequence number. The best we can do in such cases is to infer using IPID numbers, or possibly counting ACK packets coming
in from the other direction, if bidirectional information is available.

Even with the imperfect and relatively simple sequence number function used here, results are consistent with theory.
Figures~\ref{fig:esrcompare3} and \ref{fig:esrcompare4} illustrate this. In both
cases, the imperfection of the function affects the accuracy of DS, but not to a large degree.
A similar observation applies to the
Abilene-III trace, shown in Figure~\ref{fig:esrabiseq}, where the artifacts due to the imperfect function (the sawtooth pattern) are clearly seen.

Finally, Figure~\ref{fig:datacompare} illustrates the effect of using sequence
numbers in recovering the flow size distribution. The three cases shown are for DS
with parameters $\pf = 0.00117$ and $\pp = 0.7$, with an ESR of $p = 0.01$.
The PERFECT case is when DS is given a perfect sequence number function, MAPPING when
using our sequence number function, and SEQ OFF when no sequence numbers were used at all.
It is apparent that using sequence numbers, even with an imperfect function, provides
significantly more information to an estimator.

\subsection{Empirical estimator variance}

Here, we see how closely the estimator variance matches the CRLB by computing the {\em
observed Fisher information} \cite{McLachlanEM08} of the estimator in Figure~\ref{fig:empirical_crlb}. 
To improve readability, smoothing was applied to the tail
end of the observed Fisher information where samples are scarse 
using a simple moving average filter with a window size of 100. 
Even when using an imperfect sequence number function, the variance of the estimator closely matches the
CRLB, effectively proving the benefit of sequence numbers.

\begin{figure}[h]
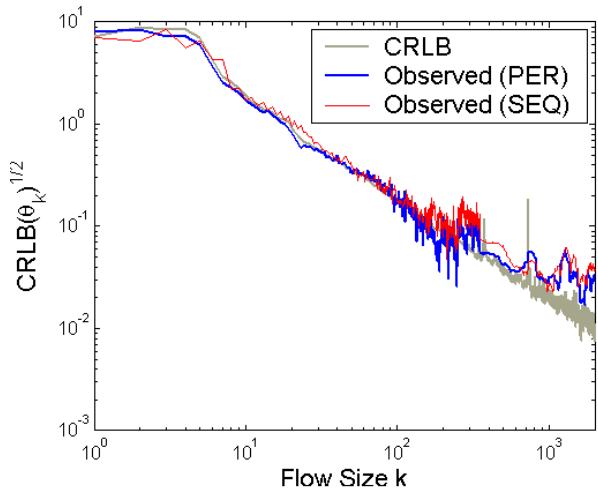

  \ESRabicrlb
  \caption{CRLB versus observed MLE variance of DS with $\pf = 0.00625,\ \pp = 0.6$ on Abilene-III.}
  \label{fig:empirical_crlb}
\end{figure}
\section{Conclusions and Future Work}
\label{sec:conc}

We have re-examined the question of sampling for flow size estimation in the context of
TCP flows from a theoretical point of view. We used the Fisher information to examine
the inherent potential of a number of sampling methods. Most of these had been
examined previously, but we showed how the usual conditional framework can be
made unconditional and thereby simplified, which actually changes the sampling methods
themselves and their performance. The new framework led to a number of new rigorous
results regarding the performance of sampling methods which we studied under two
different normalizations. It also enabled flow sampling to be compared to methods using
TCP sequence numbers and Sample and Hold for the first time, and we showed that it is
far superior to them, except in very special cases which are not important for network measurement applications.

We introduced a new two parameter family of methods, Dual Sampling, which allows the
statistical benefits of flow sampling to be traded off against the computational
advantages of packet sampling. We discussed how, as an unbiased `Hold and Sample' method,
it differs from Sample and Hold, and proved that it is superior to it. We argue that the
scheme is implementable and offers an efficient way of approaching flow sampling in
practice to the extent possible. We also proposed closed-form unbiased estimators for
SYN-based methods and SH which asymptotically achieve the CRLB, saving computational time
in the estimation stage. 

We performed a case study of Dual Sampling and Sample and Hold on two Internet data traces using our proposed
estimators, and found results entirely consistent with the theoretical predictions, despite
the fact that the function which maps sequence numbers to packet counts introduces a new
source of error, and was not highly optimized. Although there is high variation at the
tail end of the distribution, our proposed estimator closely matches the CRLB.

In future work, we intend to search over the space of all possible sampling matrices to
find an optimal sampling method for flow size estimation, and to compare it to flow sampling.
 Our framework is applicable to 
other traffic metrics such as anomaly detection and a future direction is to extend the 
work to those areas. It would also be of interest to improve the sequence number mapping function, 
and also to explore using our approach for the direct estimation of the byte size of
flows, for which sequence numbers are more naturally suited, rather than the packet size.
Finally, we will develop a more detailed case for DS and its implementability in high
speed routers.

\appendices{}

\section{Our Matrix Lemmas}
\label{sec:our_lemmas}

\subsection{General}

\begin{lem}
The matrix $\tilde \uncon(\vth)$ and its inverse $\tilde
\uncon(\vth)^{-1}$ are symmetric and positive definite.
\label{thm:t_uncon_posdef}
\end{lem}
\begin{proof}
For simplicity we omit the $\vth$ dependencies.
Recall that $\tilde \uncon = \mathbf{\tilde B}^\T\mathbf{\tilde
D}\mathbf{\tilde B}$ where $\mathbf{\tilde D}$ is a real diagonal positive definite
matrix with $(\mathbf{\tilde D})_{jj} = d^{-1}_j$ and $\rank(\mathbf{\tilde D}) = W$.
Matrix $\mathbf{\tilde B}$ is a $W \times W$ matrix and $\rank(\mathbf{\tilde B}) =
W$. It follows that an inverse exists for both $\mathbf{\tilde D}$ and
$\mathbf{\tilde B}$. Hence, an inverse also exists for $\tilde \uncon$ since
$\tilde \uncon^{-1} = \mathbf{\tilde B}^{-1}\mathbf{\tilde D}^{-1}(\mathbf{\tilde
B}^\T)^{-1}$.

An equivalent expression for $\tilde \uncon$ is
\be
  \tilde \uncon =  \mathbf{\tilde B}^\T   \mathbf{\tilde D}^{1/2}\mathbf{\tilde
D}^{1/2}\mathbf{\tilde B} = (\mathbf{\tilde D}^{1/2}\mathbf{\tilde B})^\T
(\mathbf{\tilde D}^{1/2}\mathbf{\tilde B})
\ee
since $\mathbf{\tilde D}^{1/2}$ is symmetric.

We can now show that $\tilde \uncon$ is symmetric. We have
\begin{eqnarray*}
\tilde \uncon^\T & = & \lbrack
(\mathbf{\tilde D}^{1/2}\mathbf{\tilde B})^\T(\mathbf{\tilde D}^{1/2}\mathbf{\tilde
B}) \rbrack^\T \\
& = & (\mathbf{\tilde D}^{1/2}\mathbf{\tilde B})^\T (\mathbf{\tilde
D}^{1/2}\mathbf{\tilde B}) = \tilde \uncon .
\end{eqnarray*}
The form $(\mathbf{\tilde D}^{1/2}\mathbf{\tilde B})^\T (\mathbf{\tilde
D}^{1/2}\mathbf{\tilde B})$ is at least positive semidefinite
(Theorem~\ref{thm:pos_def_prod}). However, since $\tilde \uncon^{-1} = \mathbf{\tilde
B}^{-1}\mathbf{\tilde D}^{-1}(\mathbf{\tilde B}^\T)^{-1}$,  $\tilde \uncon$ is invertible,
it is also positive definite. By definition of symmetric, positive definite matrices,
its inverse is also symmetric, positive definite.
\end{proof}

\begin{lem}
The unconstrained Fisher information matrix $\uncon(\vth)$ and its inverse
$\uncon(\vth)^{-1}$ are symmetric and positive definite.
\label{thm:uncon_posdef}
\end{lem}
\begin{proof}
Recall that $\uncon = \mathbf{B}^\T\mathbf{D}\mathbf{B}$ where $\mathbf{D}$ is a real
diagonal positive definite matrix with $(\mathbf{D})_{jj} = d^{-1}_{j-1}$ and
$\rank(\mathbf{D}) = W+1$. Matrix $\mathbf{B}$ is a $(W+1) \times W$ matrix and
$\rank(\mathbf{B}) = W$.
An equivalent expression for $\uncon$ is
\begin{eqnarray*}
\uncon & = & \mathbf{B}^\T\mathbf{D}^{1/2}\mathbf{D}^{1/2}\mathbf{B} =
(\mathbf{D}^{1/2}\mathbf{B})^\T(\mathbf{D}^{1/2}\mathbf{B})
\end{eqnarray*}
since $\mathbf{D}^{1/2}$ is symmetric.  Now
\begin{eqnarray*}
\uncon^\T & = & \lbrack
(\mathbf{D}^{1/2}\mathbf{B})^\T(\mathbf{D}^{1/2}\mathbf{B})\rbrack^\T \\
& = & (\mathbf{D}^{1/2}\mathbf{B})^\T (\mathbf{D}^{1/2}\mathbf{B}) = \uncon .
\end{eqnarray*}
From Theorem~\ref{thm:pos_def_prod}, $(\mathbf{D}^{1/2}\mathbf{B})^\T
(\mathbf{D}^{1/2}\mathbf{B})$ is positive semidefinite. 
Moreover, $\uncon$ is invertible by Proposition~\ref{prop:sub}, implying
it is positive definite. By definition of symmetric, positive definite
matrices, its inverse is also symmetric and positive definite.
\end{proof}


\smallskip
\begin{lem}
$\uncon_1 \ge \uncon_2$ if and only if $\con^{+}_1 \le \con^{+}_2$.
\label{lem:uncon_con}
\end{lem}
\begin{proof}
Since by Lemma~\ref{thm:positive_def_inv}(ii), $\uncon_1 \ge \uncon_2$ iff $\uncon^{-1}_1
\le \uncon^{-1}_2$, then by definition of positive semidefinite matrices, $\mathbf{x}^\T
\uncon^{-1}_1 \mathbf{x} \le \mathbf{x}^\T \uncon^{-1}_2 \mathbf{x}$. Hence, this implies
$\mathbf{x}^\T (\uncon^{-1}_1 - \vth \vth^\T) \mathbf{x} \le \mathbf{x}^\T  (\uncon^{-1}_2
- \vth \vth^\T) \mathbf{x}$, implying $\uncon_1 \ge \uncon_2$.
\end{proof}

\subsection{Proof of Theorem~\ref{thm:upper_bound_inv}}

\noindent Let $\mathbf{E} = (1/d_0) \mathbf{b}_0\mathbf{b}_0^\T$ from (\ref{eq:mtx_j_alt}).
It has rank 1 and is therefore positive semidefinite
since its eigenvalues are $\tr(\mathbf{E}) = \sum^W_{k=1} b^2_{0k}/d_0$ with
multiplicity 1 and $0$ with multiplicity $W-1$.
It follows from Lemma~\ref{thm:possem_results} that $\uncon(\vth) \ge
\tilde \uncon(\vth)$ since $\uncon(\vth) = \mathbf{E} + \tilde
\uncon(\vth)$, and from Lemma~\ref{thm:positive_def_inv}, this implies
that $\tilde \uncon^{-1}(\vth) - \uncon^{-1}(\vth) \ge \mathbf{0}_{W
\times W}$ and therefore $\uncon^{-1}(\vth) \le \tilde
\uncon^{-1}(\vth)$. Equality can only hold iff $\mathbf{E} = \mathbf{0}_{W \times W}$ since
this is the only case where all eigenvalues of $\mathbf{E}$ are zero. This implies that 
$\mathbf{b}_0 = \mathbf{0}_{W \times 1}$ is required for equality, that is that no flow can `evaporate'.

\section{Other Matrix Lemmas}
\label{app:other_lemmas}

We collect some useful results required in this paper here. This first result comes from
\cite{Harville97}.
\begin{lem}
Let $\mathbf{A}$ be a $n \times n$ symmetric positive definite matrix and
 $\mathbf{B}$ an $n \times n$ positive definite matrix. Then
\begin{enumerate}
  \item If $\mathbf{B}-\mathbf{A}$ is positive definite, then so is $\mathbf{A}^{-1} -
\mathbf{B}^{-1}$,
  \item If $\mathbf{B}-\mathbf{A}$ is symmetric and positive semidefinite (implying
$\mathbf{B}$ is symmetric), then $\mathbf{A}^{-1} - \mathbf{B}^{-1}\ge0$.
\end{enumerate}
\label{thm:positive_def_inv}
\end{lem}

\noindent The following theorem appears in \cite[Theorem 6.3, p.~161]{Zhang99}.
\begin{lem}
The following statements are equivalent:
\begin{enumerate}
   \item $\mathbf{A}$ is positive semidefinite;
   \item $\mathbf{A} = \mathbf{B}^* \mathbf{B}$ for some matrix $\mathbf{B}$, where
   $\mathbf{B}^*$ is the conjugate transpose of $\mathbf{B}$. 
\end{enumerate}
\label{thm:pos_def_prod}
\end{lem}

\smallskip
\noindent The following result gives more properties of positive semidefinite matrices
\cite[Theorem 6.5, p. 166]{Zhang99}.
\begin{lem}
Let $\mathbf{A} \ge 0$ and $\mathbf{B} \ge 0$ have same size. Then
\begin{enumerate}
\item $\mathbf{A} + \mathbf{B} \ge \mathbf{B}$,
\item $\mathbf{A}^{1/2}\mathbf{B}\mathbf{A}^{1/2} \ge 0$,
\item $\tr(\mathbf{A}\mathbf{B}) \le \tr(\mathbf{A})\tr(\mathbf{B})$,
\item the eigenvalues of $\mathbf{A}\mathbf{B}$ are all nonnegative.
\end{enumerate}
\label{thm:possem_results}
\end{lem}


\medskip
\noindent The matrix inversion lemma (also known as Woodbury's formula) can be found in
\cite[Theorem 18.2.8, p.~424]{Harville97}.

\begin{lem}[Matrix Inversion Lemma]
Let $\mathbf{R}$ be a $n \times n$ matrix, $\mathbf{S}$ a $n \times m$
matrix, $\mathbf{T}$ a $m \times m$ matrix, and $\mathbf{U}$ a $m \times n$ matrix.
Suppose that $\mathbf{R}$ and $\mathbf{T}$ are nonsingular. Then,
$$(\mathbf{R} + \mathbf{S}\mathbf{T}\mathbf{U})^{-1} = \mathbf{R}^{-1} -
\mathbf{R}^{-1}\mathbf{S}(\mathbf{T}^{-1} + \mathbf{U}\mathbf{R}^{-1}\mathbf{S})^{-1} \mathbf{U}\mathbf{R}^{-1}.
$$
\label{lem:matrix_inversion}
\end{lem}
\vspace{-2mm}
\noindent  The data processing inequality for Fisher information from \cite{Zamir98} is as follows.
\begin{thm}
If $\Theta \to Y \to X$ satisfies a relation of the form $f(y, x | \Theta) = f_{\Theta}
(y)f(x|y)$ (i.e.~the conditional distribution of $X$ given $Y$ is independent of
$\Theta$), then $\uncon_X(\Theta) \le \uncon_Y(\Theta)$ with the deterministic version
being $\uncon_{\gamma(Y)}(\Theta) \le \uncon_Y(\Theta)$. Equality holds if $\gamma(Y)$ is
a {\em sufficient statistic} relative to the family ${f_{\Theta}(\mathbf{y})}$,
i.e.~$\Theta \to \gamma(Y) \to Y$ forms a Markov chain.
\label{thm:fisher_dpi}
\end{thm}

\section{Sampling Methods}
\label{app:sampling}

\subsection{Proof of equation \eqref{eq:jinv_ps}}

\noindent  Expanding $(1+(-1))^k$ leads to the useful identity
\be
   \sum_{\ell=1}^k (-1)^{k-\ell}\dbinom{k}{\ell} = (-1)^{k-1}.
   \label{eq:binom_sum}
\ee
Using \eqref{eq:gen_inverse_diag} and $b'_{jk} = 
(-1)^{k-j}\dbinom{k}{j}\qp^{k-j}\pp^{-k}$ for PS, we have
\begin{align*}
\nonumber
 ({\mathbf{J}}^{-1})_{jj} &= \sum_{k=j}^W \dbinom{k}{j}^2 \qp^{2(k-j)} \pp^{-2k}
	d_k \\
	&- \tiny{\text{$\frac{\Big( \sum_{k=j}^W \sum_{\ell=1}^k (-1)^{2k-j-\ell} \binom{k}{j} 
	\binom{k}{\ell}	\qp^{2k-j} \pp^{-2k} d_k \Big)^2}{d_0 + \sum_{k=1}^W d_k \Big(\qp^k \pp^{-k}
	\sum_{\ell=1}^k (-1)^{k-\ell} \binom{k}{\ell} \Big)^2}	
	$}}.
\end{align*}
The denominator can be simplified as follows:
\begin{align*}
&d_0 + \sum_{k=1}^W d_k \Big(\qp^k \pp^{-k} \sum_{\ell=1}^k (-1)^{k-\ell} \binom{k}{\ell} 
\Big)^2\\
&= d_0 + \sum_{k=1}^W d_k \qp^{2k} \pp^{-2k} \Big( \sum_{\ell=1}^k (-1)^{k-\ell} \binom{k}{\ell} 
\Big)^2\\
&= d_0 + \sum_{k=1}^W (-1)^{2k-2} d_k \qp^{2k} \pp^{-2k}\\
&= \sum_{k=0}^W \qp^{2k} \pp^{-2k} d_k,
\end{align*}
where we used identity \eqref{eq:binom_sum} in the second line.

Similarly, the numerator can be simplified:
\begin{align*}
&\left( \sum_{k=j}^W \sum_{\ell=1}^k (-1)^{2k-j-\ell} \binom{k}{j} 
	\binom{k}{\ell}	\qp^{2k-j} \pp^{-2k} d_k \right)^2\\
&= \left( \sum_{k=j}^W  d_k \dbinom{k}{j}(-1)^{k-j} \qp^{2k-j} \pp^{-2k} 
\sum_{\ell=1}^k (-1)^{k-\ell} \dbinom{k}{\ell} \right)^2\\
&=\left( \sum_{k=j}^W  d_k \dbinom{k}{j}(-1)^{k-j} \qp^{2k-j} \pp^{-2k} 
(-1)^{k-1} \right)^2\\
&=\left( \sum_{k=j}^W  (-1)^{2k-j-1} d_k \dbinom{k}{j} \qp^{2k-j} \pp^{-2k} \right)^2,
\end{align*}
where \eqref{eq:binom_sum} is used in the third line. 
This proves \eqref{eq:jinv_ps}.

\section{Comparisons}
\label{app:comps}

\subsection{Proof of Theorem \ref{thm:dpi_sampling}}

\noindent Let $Z$ denote any method except SH, and $Y$ the complete outcome  
(i.e.~the vector of SEQ numbers and the SYN variable in the richest case) of sampling with $Z$ using parameter $p_1$ (for DS $(\ppo,\pfo)$).
Now sample  $Y$ using $Z$ with parameter $p_2/p_1$ to form $X$ (for DS $(\ppt/\ppo,\pft/\pfo)$). 
Since $X$ is a function only of $Y$ and the new sampling, the DPI for Fisher applies  
(Theorem~\ref{thm:fisher_dpi}) (and $X \subseteq Y$). 
Furthermore, it is easy to see that $X$ is statistically equivalent to the outcome of sampling the original data 
using $Z$ with probability $p_1(p_2/p_1)=p_2$ (for DS $(\ppo(\ppt/\ppo),\pfo(\pft/\pfo))=(\ppt,\pft)$).
It follows that $\uncon_Z(p_1) \ge \uncon_Z(p_2)$.
Equality holds iff $p_2 = p_1$ (for DS  $\ppo=\ppt$ and $\pfo=\pft$) implying $X = Y$, since sampling inherently discards information.

The above proof does not apply to SH as it does not have a closure
property, meaning that $X$ is not equivalent to applying SH with some $\pp$.
We turn instead to a direct method, exploiting the tridiagonal structure of $\uncon^{-1}_{\text{SH}}$.

Denote $(\uncon^{-1}_{\text{SH}})_{ij} = a_{i,j}$, and consider the quadratic form 
$\mathbf{x}^\T \uncon^{-1}_{\text{SH}} \mathbf{x}$ for any non-zero vector $\mathbf{x} \in \Real^W$:
\begin{align}
\nonumber
 \mathbf{x}^\T \uncon^{-1}_{\text{SH}} \mathbf{x} &= (a_{1,1} + a_{1,2}) x_1^2 + (a_{W,W} + a_{W-1,W})x^2_W \\
  \label{eq:sh_quadratic}
    &+\sum_{j=2}^{W-1} (a_{j,j} + a_{j-1,j} + a_{j,j+1}) x^2_j \\\nonumber
    & + \sum_{j=1}^W (-a_{j,j+1}) (x_j - x_{j+1})^2.
\end{align}
From Lemma~\ref{lem:prod_rowstoch}  we know $\uncon^{-1}\onew = \vth$. 
It follows that
\begin{align*}
  (a_{1,1} + a_{1,2}) &= \theta_1\\
  (a_{j,j} + a_{j-1,j} - a_{j,j+1}) &= \theta_j,\ j = 2,\ldots ,W-1\\
  (a_{W,W} + a_{W-1,W}) &= \theta_W
\end{align*}
which are each independent of $\pp$.  Consider then the term involving 
$-a_{j,j+1} = \pp^{-2} \qp d_{j+1} = \frac{1}{\pp} \sum_{k=j+1}^W \qp^{k-j} \theta_k$, $1 \le j < W$. 
Differentiating with respect with $\pp$, we obtain
\ben
   -\frac{1}{\pp^2} \sum_{k=j+1}^W \qp^{k-j} \theta_k - \frac{1}{\pp} \sum^W_{k=j+1} (k-j) \qp^{k-j-1} \theta_k.
\een
Since this is negative, it follows that  $\mathbf{x}^\T \uncon^{-1}_{\text{SH}} \mathbf{x}$ 
decreases monotonically in $\pp$, and so $\uncon^{-1}_{\text{SH}}(p_1) \le \uncon^{-1}_{\text{SH}}(p_2)$ for any $p_1\ge p_2$.
The result then follows by Lemma~\ref{thm:positive_def_inv}(ii).

\subsection{Proof of Theorem \ref{thm:ps_pssyn_comp}}

\noindent We first consider PPR normalization, where $\pp = p$ for both methods. Recall from 
(\ref{eq:jinv_ps}) that 
\begin{align*}
 ({\mathbf{J}}^{-1}_{\text{PS}})_{jj} &= \sum_{k=j}^W \dbinom{k}{j}^2 q^{2(k-j)}
	p^{-2k} d_{k, \text{PS}} \\
	&\hspace{5mm} - \frac{\Big( \sum_{k=j}^W  (-1)^{2k-j-1} d_{k, \text{PS}} \dbinom{k}{j} 
	q^{2k-j} p^{-2k} \Big)^2}{\sum_{k=0}^W q^{2k} p^{-2k} d_{k, \text{PS}}}	\\
       & \ge \underbrace{\sum_{k=j}^W \dbinom{k}{j}^2 q^{2 (k-j)} p^{-2k} d_{k, \text{PS}}}_{T_1} 
           - q^{-2j} (W\!-\!j\!+\!1)W! 
\end{align*}
where the lower bound follows by, in the second term, 
dropping the first $j$ terms in the denominator and upper bounding $\dbinom{k}{j}$ by $W!$.
Also, from (\ref{eq:jinv_pssyn}),
\ben
   ({\mathbf{J}_{\text{PS+SYN}}^{-1}})_{jj} = \underbrace{\sum_{k=j}^W	\dbinom{k-1}
{j-1}^2 q^{2(k-j)} p^{-2k} d_{k, \text{PS+SYN}}}_{T_2} - \frac{q}{p} \theta^2_j.
\een
Since $\binom{k}{j} = \binom{k-1}{j-1} + \binom{k-1}{j} \ge \binom{k-1}{j-1}$,
\begin{align*}
    d_{j, \text{PS}} &\ge \sum_{k=j}^W \dbinom{k-1}{j-1} p^{j} q^{k-j} \theta_k = d_{j,\text{PS+SYN}},
\end{align*}
implying that $T_1 \ge T_2$.
 
Now under the limit $p\to 0$, for $j \ge 1$, the second term for PS becomes
$-(W-j+1)W!$ which is finite, whereas for PS+SYN, $-\frac{q}{p} \theta_j^2 \approx -\frac{1}{p}
\theta_j^2 \to -\infty$. It follows that $({\mathbf{J}_{\text{PS+SYN}}^{-1}})_{jj} \le 
({\mathbf{J}}^{-1}_{\text{PS}})_{jj}$. Since $\pp> p$ for PS+SYN under ESR, the result also holds 
for ESR by Theorem~\ref{thm:dpi_sampling}.

\subsection{Proof of Theorem \ref{thm:pssynseq_comp}}

\noindent We begin by examining the simplest case, where $j=W$. We have
\be
   \nonumber
   (\uncon^{-1}_{\text{PS+SEQ}})_{WW} = \frac{\theta_W}{\pp^2}
\ee
and
\be
   \nonumber
   (\uncon^{-1}_{\text{PS+SYN+SEQ}})_{WW} = \frac{\theta_W}{\pp^2} - \frac{\qp} {\pp}\theta^2_W,
\ee
and so $(\uncon^{-1}_{\text{PS+SYN+SEQ}})_{WW} \le (\uncon^{-1}_{\text{PS+SEQ}})_{WW}$ under PPR.

Now consider the case $3 \le j \le W-1$. Let $d_{Q,j}$ and $d_{SQ,j}$ be
the proportion of sampled flows of size $j$ after sampling by PS+SEQ and PS+SYN+SEQ
respectively. For $j \ge 1$, a straightforward comparison would yield $d_{Q,j} \ge
d_{SQ,j}$. Intuitively, since more flows are discarded in the PS+SYN+SEQ scheme, the
proportion of sampled flows must be less than the proportions of PS+SEQ. Therefore,
expressing the diagonals in terms of $d_{Q,j}$ and $d_{SQ,j}$, we have for $3 \le j \le
W-1$,
\be
   \nonumber
   (\uncon^{-1}_{\text{PS+SEQ}})_{jj} = \pp^{-4}d_{Q,j} + 4\qp^2\pp^{-4}d_{Q,j+1} + \qp^4 \pp^{-4} d_{Q,j+2}
\ee
and
\begin{align*}
   (\uncon^{-1}_{\text{PS+SYN+SEQ}})_{jj}
          & = \pp^{-4}d_{SQ,j} + \qp^2\pp^{-4}d_{SQ,j+1} - \qp\pp^{-1}\theta^2_j \\
          & \le \pp^{-4}d_{SQ,j} + \qp^2\pp^{-4}d_{SQ,j+1},
\end{align*}
which, by a direct comparison, shows $(\uncon^{-1}_{\text{PS+SYN+SEQ}})_{jj} \\ \le
(\uncon^{-1}_{\text{PS+SEQ}})_{jj}$. Similarly, the ESR comparison follows since the
sampling rate for PS+SYN+SEQ must increase, thereby reducing its CRLB.

\subsection{Proof of Theorem \ref{thm:ds_mono}}

\noindent First consider the simplest case, $j=W$. By substituting (\ref{eq:ESRcurve}) into
(\ref{eq:DS_diags}) and then differentiating w.r.t $\pp$
\begin{align*}
   \frac{d}{d\pp}(\uncon_{\text{DS}}^{-1})_{W\!W} =  - \frac{\theta_W}{\pp^2 pD} -
\frac{(D-1)\theta_W^2}{pD} < 0,
\end{align*}
implying $(\uncon_{\text{DS}}^{-1})_{W\!W}$ is monotonically decreasing with $\pp$.

For $2 \le j \le W-1$ the derivative $\frac{d}{d\pp}(\uncon_{\text{DS}}^{-1})_{jj}$ is given by
\begin{align*}
    & -\frac{\theta_j}{\pp^2 pD} - \frac{(D-1)\theta_j^2}{pD} - \frac{1}{\pp^2 pD}
         \left( \sum^W_{k=j+1} \qp^{k-j} (1+\qp)\theta_k \right) \\
    & \quad - \frac{1}{\pf\pp} \left(\sum^W_{k=j+1} \qp^{k-j-1}((k-j) + (k-j+1)\qp) \theta_k
     \right),
\end{align*}
which is negative since each term is negative for $0 < \pp \le 1$.

Finally, for $j=1$ we have
\begin{align}
\nonumber
  \frac{d}{d\pp}(\uncon_{\text{DS}}^{-1})_{11} & = \frac{D-1}{pD} \theta_1 (1-\theta_1) -\frac{1}{\pp^2 pD} \left( \sum_{k=2}^W
        \qp^{k-1} \theta_k \right) \\
 \label{eq:varone_diff}
                                                                 & - \frac{\pp(D-1)+1}{\pp pD} \left( \sum^W_{k=2} (k-1)\qp^{k-2}\theta_k \right).
\end{align}
For small values of $\pp$ the expression is dominated by terms in $1/\pp$ and is therefore
again negative, but as the first term is positive, for large $\pp$ it may change sign.
It is not hard to show that $\frac{d}{d\pp^2}(\uncon_{\text{DS}}^{-1})_{11} > 0$,  so at
most one sign change is possible. Setting $\pp=1$ in (\ref{eq:varone_diff}) yields
(\ref{eq:mono_cond}) as the necessary and sufficient condition for this not to occur in
the feasible region $\pp\le 1$. The special cases follow simply from (\ref{eq:mono_cond}).

\subsection{Proof of Theorem \ref{thm:sh_comp}}

\noindent First consider the case $2 \le j \le W$.  From \eqref{eq:fsinv_alt} and (\ref{eq:SHj})
\ben
   (\uncon^{-1}_{\text{SH}})_{jj}  \ge \frac{\theta_j}{\pp}  \ge \frac{\theta_j}{p}  \ge \frac{\theta_j}{p} - \frac{q}
             {p}\theta_j^2 = (\uncon^{-1}_{\text{FS}})_{jj}
\een
since $\pf=p$ and $\pp=\pp(p)\le p$ under both PPR and ESR.

\smallskip
\noindent Now consider $j=1$.  It is convenient to recall (\ref{eq:SHone}) and 
\eqref{eq:fsinv_alt}:
 $(\uncon^{-1}_{\text{SH}})_{11} = \theta_1 + \frac{1}{\pp} \sum_{k=2}^W \qp^{k-1} \theta_k$, 
 and  $(\uncon^{-1}_{\text{FS}})_{11} = \frac{1}{\pf}\theta_1 - \frac{\qf}{\pf} \theta_1^2$.
It follows that $(\uncon^{-1}_{\text{FS}})_{11} \le (\uncon^{-1}_{\text{SH}})_{11}$ when
\be
   \pp \le \frac{p \sum_{k=2}^W \qp^{k-1} \theta_k}{\theta_1(1-\theta_1)(1-p)}.
\label{eq:actual}
\ee
A sufficient condition implying (\ref{eq:actual}) is obtained by using the lower bound 
$\qp\theta_2 \le \sum_{k=2}^W \qp^{k-1} \theta_k$ and rearranging,
yielding
\be
  \pp \le \frac{p \theta_2}{\theta_1(1-\theta_1)(1-p) + p\theta_2}.
\label{eq:suff_cond_sh}
\ee
Furthermore, since $\pp \le p$, a more restrictive sufficient condition is given by replacing the l.h.s.~with $p$,  
which reduces to $p(1-p)\Big( \theta_1(1-\theta_1) - \theta_2\Big) \le 0$, which shows that for any  
$0 \le p \le 1$, if $\theta_2 \ge \theta_1(1-\theta_1)$, then  (\ref{eq:actual}) holds and 
$(\uncon^{-1}_{\text{FS}})_{11} \le (\uncon^{-1}_{\text{SH}})_{11}$.  The condition is satisfied if $W=2$.

Now let $\vth$ be arbitrary and consider the small $p$ (and hence small $\pp$) limit.  
Then (\ref{eq:actual}) becomes $\pp\le p/\theta_1$, which is always satisfied since $\pp\le p$.

\section{Maximum Likelihood Estimators}
\label{sec:mle}

\subsection{Proof of Theorem \ref{thm:mle_closed}}
\label{ssec:closed_mle}

\noindent The likelihood function for $\nf$ flows is
\ben
    f(\vth, \nf) = \prod_i f(j_i;\vth) = \prod_{j \ge 0} d_j^{M'_j}.
\een
The MLE is the $\vth$ which maximizes the log-likelihood 
\begin{align*}
    l(\vth, \nf) &= \sum_{j \ge 0} M'_j \log d_j 
\end{align*}
subject to the constraint $\sum_{k \ge 1} \theta_k = 1$, $\theta_k > 0,\ \forall k$. The
optimization problem admits a feasible solution by the Bolzano-Weierstrass theorem 
\cite[p.~517]{TaylorCalculus83}, since the log-likelihood function is concave and
continuous, and optimization is performed over a compact, convex set. Furthermore, the
solution obtained will be unique under our assumptions, since the Hessian of the
log-likelihood is the Fisher information, which is positive definite given  $0 < \theta_k 
< 1$ for all $k$.

Given the assumptions, the method of Lagrange multipliers would yield the optimal solution 
since strong duality holds as the problem satisfies Slater's constraint qualification 
\cite[Section 5.2.3, p.~226]{Boyd04}. 
The Lagrangian is
\begin{align*}
\mathcal{L}(\vth, \lambda, \boldsymbol{\nu}) &= \sum_{j\ge 0} M'_j \log d_j -
\lambda\Big( \sum_{k \ge 1} \theta_k - 1 \Big) - \boldsymbol{\nu}^\T \vth,
\end{align*}
where the vector $\boldsymbol{\nu}$ has elements $\nu_k \ge 0$ and $\lambda \in \Real$. By
differentiating with respect to $\theta_k,\ \forall k$ and the
multipliers,
\begin{align*}
\frac{\partial}{\partial \theta_k} \mathcal{L}(\vth, \lambda, \boldsymbol{\nu}) &=
\sum_{k \ge j} \frac{M'_j} {d_j} b_{jk} - \lambda
- \nu_k = 0, \\
\frac{\partial}{\partial \lambda} \mathcal{L}(\vth, \lambda, \boldsymbol{\nu}) &=
1 - \sum_{k \ge 1} \theta_k = 0,\\
\frac{\partial}{\partial \nu_k} \mathcal{L}(\vth, \lambda, \boldsymbol{\nu}) &= \theta_k
= 0.
\end{align*}
The second equation is just the equality constraint while the third yields a solution
$\vth = \mathbf{0}_W$, which lies on the boundary, yielding an unbounded solution 
(observed by substituting the solution into the likelihood function). That leaves
the first equation, and by the Karush-Kuhn-Tucker condition, $\boldsymbol{\nu}^\T \vth = 0$ 
\cite[Section 5.5.3, p.~243]{Boyd04}, implying that $\boldsymbol{\nu} = \mathbf{0}_W$ 
(our assumptions require that the original parameters $0 < \theta_k < 1$ for all $k$,
hence the optimal must lie within the region where the constraints are inactive). Thus, we have,
in matrix form,
\be
\mathbf{\tilde B}^\T \mathbf{\tilde D}\, \diag(M'_1, \ldots, M'_W) \onew = \lambda
\onew - \frac{M'_0}{d_0} \mathbf{b}_0,
\label{eq:general_mle_eq}
\ee
recalling that $\mathbf{\tilde D} = \diag(d^{-1}_1, \ldots, d^{-1}_W)$.

We proceed to solve for $\lambda$ using the equality constraint $\sum_{k \ge 1} \theta_k =
1$ and $\mathbf{\tilde d} = \mathbf{\tilde B} \vth$, as follows, by multiplying both
sides of (\ref{eq:general_mle_eq}) by $\vth^\T$ to obtain
\begin{align*}
\vth^\T \mathbf{\tilde B}^\T \mathbf{\tilde D}\, \diag(M'_1, \ldots, M'_W) \onew &= 
\lambda - M'_0\\
\mathbf{\tilde d}^\T \mathbf{\tilde D}\, \diag(M'_1, \ldots, M'_W) \onew &= \lambda - 
M'_0\\
\onew^\T \diag(M'_1, \ldots, M'_W) \onew &= \lambda - M'_0,
\end{align*}
implying $\lambda = \nf$.

For methods that pick flows in an unbiased manner, $\mathbf{b}_0 = q \onew$, and
thus $\mathbf{\tilde B}^\T \onew = p \onew$, implying $(\mathbf{\tilde B}^{-1})^\T \onew =  
p^{-1} \onew$, therefore (\ref{eq:general_mle_eq}) reduces to
\be
  \mathbf{\tilde B}^\T \mathbf{\tilde D} \diag(M'_1, \ldots, M'_W) \onew = (\nf -M'_0) \onew.
  \label{eq:optimal}
\ee
which can be rewritten as
\begin{align*}
   \mathbf{\tilde D}^{-1} (\mathbf{\tilde B}^{-1})^\T \onew &= \frac{1}{\nf - M'_0} \diag(M'_1, 
   \ldots, M'_W) \onew\\
   \mathbf{\tilde B} \vth &= \frac{p}{\sum^W_{j = 1} M'_j} [M'_1,\ldots,M'_W]^\T.
\end{align*}  
The result follows.

\begin{rem}
In the conditional framework, expressing the log-likelihood function is difficult, due
the fact that normalization of the likelihood involves division by random variables. 
However, the estimator above would, with high probability, be close to the actual MLE. 
The flow selection process is a Bernoulli process. The denominator, $\sum_{j=1}^W M'_j$
encapsulates information about $\nf$, because asymptotically, the deviation between $p
\nf$ and $\sum_{j=1}^W M'_j$ is extremely small, a consequence of the concentration of
Bernoulli samples around its mean. This property is not found amongst other methods,
such as PS, where samples are biased towards large flows.
\label{rem:conditional_mle}
\end{rem}

\subsection{Proof of Theorem \ref{thm:uncond_sh_mle}}
\label{ssec:sh_mle_closed}

\noindent We begin with the optimization equation~(\ref{eq:general_mle_eq}),
\ben
\mathbf{\tilde B}^\T \mathbf{\tilde D}\, \diag(M'_1, \ldots, M'_W) \onew +
\frac{M'_0}{d_0} \mathbf{b}_0= \nf \onew.
\een
Using properties of the sampling matrix for SH, we have
\begin{align}
\label{eq:b0_reduction}
\hspace{-3mm}\mathbf{\tilde D}\, \diag(M'_1, \ldots, M'_W) \onew +
\frac{M'_0}{d_0} \cdot \frac{q}{p} \mathbf{e}_1 &= \nf (\mathbf{\tilde B}^{-1})^\T\onew\\
\label{eq:d0_simplify}
\hspace{-3mm}\mathbf{\tilde D}\, \diag(M'_0 + M'_1, \ldots, M'_W) \onew &= \nf (\mathbf{\tilde 
B}^{-1})^\T\onew\\
\label{eq:B_simplify}
\hspace{-3mm}\mathbf{\tilde D}\, \diag(p(M'_0 + M'_1), \ldots, M'_W) \onew &= \nf \onew.
\end{align}
where in (\ref{eq:b0_reduction}) we use the property
$\mathbf{b}^\T_0 = \frac{q}{p} \mathbf{\tilde B}^\T \mathbf{e}_1$, where $\mathbf{e}_i$ is
the canonical basis vector, in (\ref{eq:d0_simplify}) we use $d_0 = \frac{q}{p}d_1$, and
in (\ref{eq:B_simplify}) we use $(\mathbf{\tilde B}^{-1})^\T\onew = \diag(p^{-1},1\ldots,1)
\onew$. All these properties can be obtained by a straightforward evaluation using the
sampling matrix. The final line reduces to
\ben
\mathbf{\tilde B}\vth = \frac{1}{\nf}\cdot[p(M'_0+M'_1),\ldots,M'_W]^\T,
\een
proving the result. 

The estimator is unbiased, as by taking the expectation, we obtain
$\E\lbrack p(M'_0+M'_1)\rbrack = \nf (p d_0 + p d_1) = \nf (p+q) d_1 = d_1$, by using the
identity $d_1 = \frac{\pp}{\qp}d_0$ while clearly $\E\lbrack M'_j\rbrack = d_j$ for all $j
\ge 2$. Thus $\E \lbrack \hat \vth_{\text{SH}} \rbrack = \mathbf{\tilde
B}^{-1}_{\text{SH}} \mathbf{\tilde B}_{\text{SH}} \vth = \vth$.

%
\small
\bibliographystyle{IEEEtran}
\bibliography{darryls_publications,darryl}

\begin{biography}
[{\includegraphics[width=1in,height=1.25in,clip,keepaspectratio]{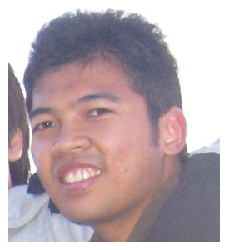}}]
{Paul Tune (M'09) received the B.E.~(Hon.) and B.Sc.~degrees in Electrical and
Electronics Engineering and Computer Science respectively, both in 2005, from the
University of Melbourne, Australia. He has recently completed a Ph.D.~degree at
the ARC Special Center for Ultra-Broadband Information Networks (CUBIN) at the
University of Melbourne, Australia. His research interests are in communication networks,
particularly in traffic inversion and sampling, information theory, statistics and
signal processing, specifically sparse signal recovery.}
\end{biography}
\vspace*{-2\baselineskip}
\begin{biography}
[{\includegraphics[width=1in,height=1.25in,clip,keepaspectratio]{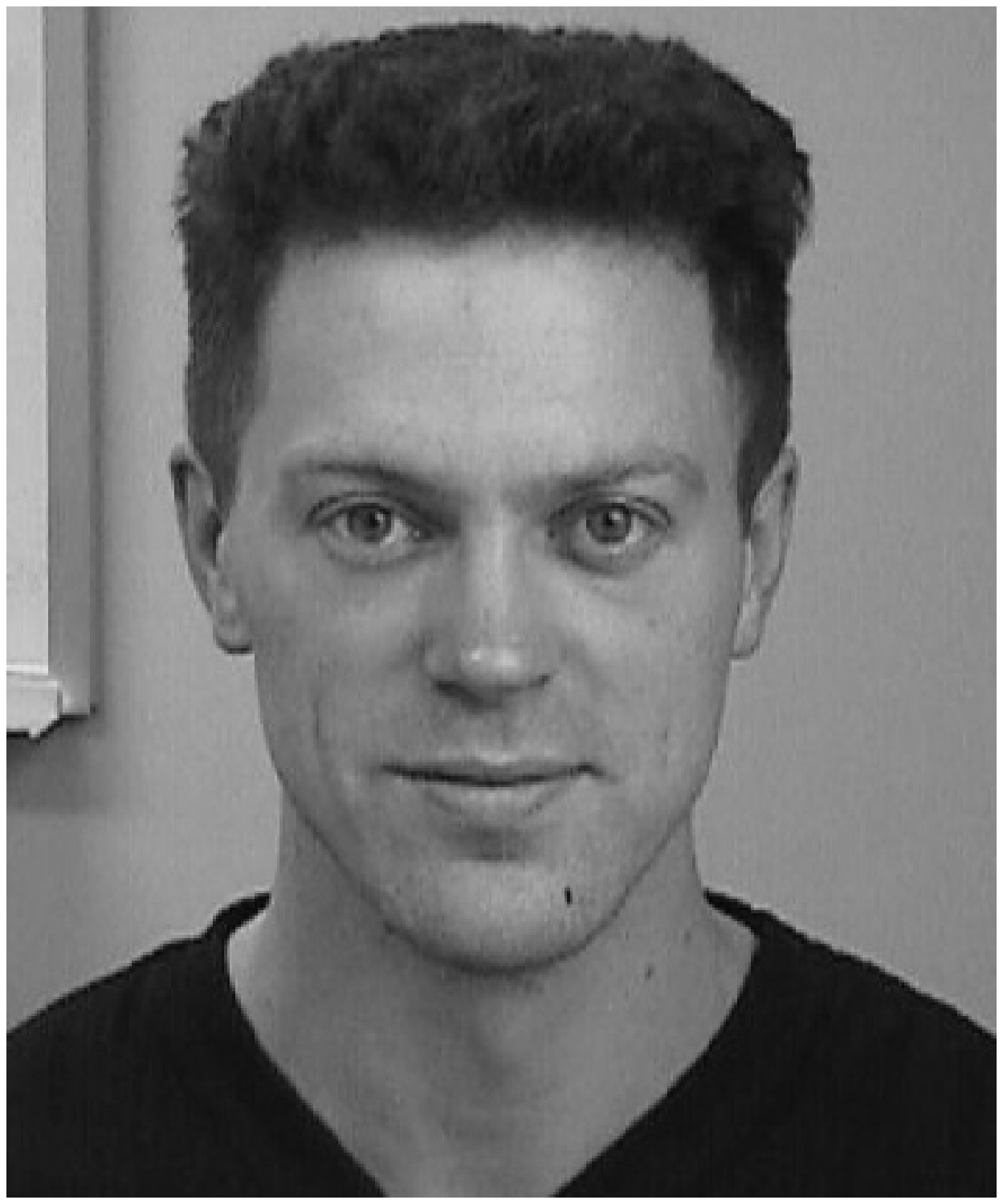}}]
{Darryl Veitch (F'10) completed a BSc.~Hons.~at Monash University, Australia (1985)
and a mathematics Ph.D.~from Cambridge, England (1990). He worked  at
TRL (Telstra, Melbourne), CNET (France Telecom, Paris), KTH
(Stockholm), INRIA (Sophia Antipolis, France), Bellcore (New Jersey), RMIT (Melbourne)
and EMUlab and CUBIN at The University of
Melbourne, where he is a Professorial Research Fellow.
His research interests include traffic modelling, parameter
estimation, active measurement, traffic sampling,  and clock synchronization.}
\end{biography}

\end{document}